\documentclass[journal,10pt]{IEEEtran}               %Two column
\usepackage[T1]{fontenc}
\usepackage{amsmath}
\usepackage{amssymb}
\usepackage{stmaryrd}
\usepackage{epic,eepic}
\usepackage{theorem}
\usepackage{pifont}  %needed by dingautolist
\usepackage{euscript}
\usepackage{calc}
\usepackage{cite}
\usepackage[titles]{tocloft}
\usepackage[usenames]{color}          % Need the color package
\usepackage{xcolor}
\definecolor{Brown}{rgb}{0.55,0.0,0.10}
\definecolor{dgreen}{rgb}{0.00,0.56,0.00}
\definecolor{vertmoinsfonce}{rgb}{0.00,0.50,0.00}
\definecolor{vert}{rgb}{0.00,0.60,0.00}
\definecolor{llightggray}{rgb}{0.97,0.97,0.97}
\definecolor{lightggray}{rgb}{0.9,0.9,0.9}
\definecolor{ggray}{rgb}{0.5,0.5,0.5}
\definecolor{darkggray}{rgb}{0.25,0.25,0.25}
\definecolor{ddarkggray}{rgb}{0.1,0.1,0.1}
\definecolor{bleu}{rgb}{0.00,0.00,1.00}
\definecolor{darkblue}{rgb}{0,0,0.7}
 \definecolor{newlightblue}{rgb}{0,1,1}

\usepackage{shortvrb,psfrag}
\usepackage{epsf}           % dessin
\usepackage{graphicx}       % dessin
\usepackage{epsfig}         % dessin
\usepackage{pstricks}
\usepackage{pstricks-add}
\usepackage{pst-plot}
\usepackage{dsfont}
\usepackage{array}
\usepackage{mdwmath}
\usepackage{mdwtab}

\theoremheaderfont{\color{black}\normalfont\bfseries}

\newtheorem{lemma}{Lemma}
\newtheorem{theorem}{Theorem}[section]
\newtheorem{definition}[theorem]{Definition}
\newtheorem{corollary}[theorem]{Corollary}
\newtheorem{proposition}[theorem]{Proposition}
\newtheorem{example}[theorem]{Example}

\theoremstyle{plain}{\theorembodyfont{\rmfamily}%
}
\theoremstyle{plain}{\theorembodyfont{\rmfamily}%
}
\theoremstyle{plain}{
\theorembodyfont{\rmfamily}

	\newtheorem{remark}[theorem]{Remark}

	}

\newcommand{\R}{\mathbb{R}}

\newcommand{\N}{\mathbb{N}}

\newcommand{\E}{\mathbb{E}}
\newcommand{\Q}{\mathbb{Q}}

\font\dsrom=dsrom10 scaled 1200 \def \indic{\textrm{\dsrom{1}}}
\newcommand{\UN}{\indic}

\newcommand{\C}{\mathcal{C}}

\newcommand{\QQ}{\mathcal{Q}}

\newcommand{\D}{\mathcal{D}}

\newcommand{\PP}{\mathcal{P}}

\newcommand{\mc}{\mathcal}
\newcommand{\bd}{\textbf}

{\Large\normalsize}
{\Large\normalsize}

\ifCLASSINFOpdf
\else
\fi
\begin{document}

%\tableofcontents

\title{Joint Empirical Coordination of Source and Channel}

%\author{\IEEEauthorblockN{Ma\"{e}l Le Treust}\\
%\IEEEauthorblockA{
%ETIS UMR 8051, Université Paris Seine, Université Cergy-Pontoise, ENSEA, CNRS,\\
%6, avenue du Ponceau,\\
%95014 CERGY-PONTOISE CEDEX,\\
%FRANCE\\
%Email: mael.le-treust@ensea.fr}
%\thanks{Manuscript received March 26, 2014; revised July 09, 2015; September 21, 2016; accepted May 17, 2017. This paper was presented in part at the Paris Game Theory Seminar June 23, 2014. Ma\"{e}l Le Treust acknowledges financial support from INS2I CNRS through projects  JCJC CoReDe 2015 and PEPS StrategicCoo 2016.}}

\author{\IEEEauthorblockN{Ma\"{e}l Le Treust, \emph{Member, IEEE}}
\thanks{Manuscript received March 26, 2014; revised July 09, 2015; September 21, 2016; accepted May 17, 2017. Date of publication XXX; date of current version June 02, 2017. The author acknowledges financial support from INS2I CNRS through projects  JCJC CoReDe 2015 and PEPS StrategicCoo 2016. This paper was presented in part at the Paris Game Theory Seminar June 23, 2014. 

Ma\"{e}l Le Treust is with ETIS UMR 8051, Université Paris Seine, Université Cergy-Pontoise, ENSEA, CNRS, 6, avenue du Ponceau, 95014 Cergy-Pontoise CEDEX, France (e-mail: mael.le-treust@ensea.fr).

Communicated by H. Permuter, Associate Editor for Shannon Theory. 

Digital Object Identifier 10.1109/TIT.2017.XXX}}

%\author{\IEEEauthorblockN{Ma\"{e}l Le Treust\IEEEauthorrefmark{1} and 
%Tristan Tomala \IEEEauthorrefmark{2}}\\
%\IEEEauthorblockA{\IEEEauthorrefmark{1}
%ETIS, UMR 8051 / ENSEA, Université Cergy-Pontoise, CNRS,\\
%6, avenue du Ponceau, 95014 Cergy-Pontoise CEDEX, FRANCE\\
%Email: mael.le-treust@ensea.fr}
%\thanks{\IEEEauthorrefmark{1} Ma\"{e}l Le Treust acknowledges financial support from INS2I CNRS through projects  JCJC CoReDe 2015 and PEPS StrategicCoo 2016.}\\
%\IEEEauthorblockA{\IEEEauthorrefmark{2}
%HEC Paris, GREGHEC UMR 2959\\
%1 rue de la Libération, 78351 Jouy-en-Josas CEDEX, FRANCE\\
%Email: tomala@hec.fr}
%\thanks{\IEEEauthorrefmark{2} Tristan Tomala acknowledges financial support from the HEC foundation.}}

%\author{Maël~Le~Treust,~\IEEEmembership{Member,~IEEE,}
%        Abdellatif Zaidi,~\IEEEmembership{Member,~IEEE},
%        \thanks{Laboratoire d'Informatique de l'Institut Gaspard Monge,
%Universit\'{e} Paris-Est Marne La Vall\'{e}e,77454, Marne La Vall\'{e}e Cedex 2, France.
%Email: \{mael.letreust\},\{abdellatif.zaidi\}@univ-mlv.fr}}

%% The paper headers
%\markboth{Journal of \LaTeX\ Class Files,~Vol.~6, No.~1, January~2007}%
%{Shell \MakeLowercase{\textit{et al.}}: Bare Demo of IEEEtran.cls for Journals}

%and will be distributed control systems,  sensor networks, repeated games, resource allocation, tasks assignment.

\maketitle

\begin{abstract}
In a decentralized and self-configuring network, the communication devices are considered as  autonomous decision-makers that sense their environment and that implement optimal transmission schemes. It is essential that these autonomous devices cooperate and coordinate their actions, to ensure the reliability of the transmissions and the stability of the network. We study  a point-to-point scenario in which the encoder and the decoder implement decentralized policies that are coordinated. The coordination is measured in terms of  empirical frequency of symbols of source and channel. The encoder and the decoder perform a coding scheme such that the empirical distribution of the symbols is close to a target joint probability distribution. We characterize the set of achievable target probability distributions for a point-to-point source-channel model, in which the encoder is non-causal and the decoder is strictly causal \textit{i.e.}, it returns an action based on the observation  of the past channel outputs. The objectives of the encoder and of the decoder, are captured by some utility function, evaluated with respect to the set of achievable target probability distributions. In this article, we investigate the maximization problem of a utility function that is common to both encoder and decoder. We show that the compression and the transmission of information are particular cases of the empirical coordination. 
%, but these results will soon be extended  to the case of distinct utility functions. 

\end{abstract}

%where the anticipation speed $\gamma$ is chosen appropriately



% For peer review articles, you can put extra information on the cover
% \fi
%
% For peerreview articles, this IEEEtran command inserts a page break and
\IEEEpeerreviewmaketitle

\section{Introduction}\label{sec:Introduction}

A decentralized network is composed of communication devices that sense their environment and that choose autonomously the best transmission scheme to implement. \textcolor[rgb]{0.00,0.00,0.0}{The decision process in large and self-configuring networks composed by different communication technologies is decentralized and does not require a central controller.} However, it is essential that the communication devices cooperate and coordinate their actions, in order to ensure the reliability of the transmissions and the stability of the network. We investigate the problem of the coordination of two autonomous devices, by considering a point-to-point model, represented by Fig. \ref{fig:1CwithChannel}, with an information source and a noisy channel. In the classical scenario, both encoder and decoder have the same objective: to implement a reliable transmission scheme. We wonder how this simple network operates when the devices are autonomous and try to coordinate their actions in order to achieve a broader common objective.

%We wonder how this simple network operates when the devices are autonomous and aim at coordinating their actions in order to achieve a more general objective.
%are not inclined to agree on the same coding scheme. 

We study this problem using a two-step approach. {First}, we characterize the coordination possibilities available for the encoder and the decoder, by using the concepts of empirical  distribution and empirical frequencies of symbols. Based on their observations, the encoder and the decoder choose the sequences of channel input and decoder's output. We require that the empirical distribution of all the sequences of symbols, converges to a target joint probability distribution. The aim is to determine the minimal amount of information to exchange such that  the symbols of both transmitters  are coordinated  with the symbols of the source and of the channel.  From an information theoretic point of view, this problem is closely related to the joint source-channel coding problem with two-sided state information and correlated source and state \cite{LeTreust(ISIT-TwoSided)15}.  We characterize the set of achievable joint probability distributions using a single-letter information constraint that is related to the compression and to the transmission of information.
\begin{figure}[!ht]
\begin{center}
\psset{xunit=0.9cm,yunit=0.9cm}
\begin{pspicture}(0,-0.35)(8.5,1.7)
\pscircle(0,0.5){0.45}
\psframe(2,0)(3,1)
\pscircle(5,0.5){0.45}
\psframe(7,0)(8,1)
\psline[linewidth=1pt]{->}(0.5,0.5)(2,0.5)
\psline[linewidth=1pt]{->}(3,0.5)(4.5,0.5)
\psline[linewidth=1pt]{->}(5.5,0.5)(7,0.5)
\psline[linewidth=1pt]{->}(8,0.5)(9,0.5)
%\psline[linewidth=1pt]{->}(0,0)(0,-0.5)(5,-0.5)(5,0)
%\psline[linewidth=1pt]{->}(2.5,-0.5)(2.5,0)
%\psline[linewidth=1pt]{->}(0,1)(0,1.5)(7.5,1.5)(7.5,1)
\rput[u](1,0.8){$U^n$}
\rput[u](3.75,0.8){$X^n$}
\rput[u](6.25,0.8){$Y^{i-1}$}
\rput[u](8.5,0.8){$V_i$}
\rput(0,0.5){$\PP_{\sf{u}}$}
\rput(2.5,0.5){$\C$}
\rput(5,0.5){$\mc{T}_{\sf{y|x}}$}
\rput(7.5,0.5){$\D$}
\end{pspicture}
\caption{The information source has i.i.d. probability distribution $\PP_{\sf{u}}(u)$ and the channel $\mc{T}(y|x)$ is memoryless. The encoder and the decoder implement a coding scheme such that, the empirical frequencies of symbols are close to the target joint probability distribution $\QQ$, defined over the symbols $\mc{U} \times \mc{X} \times \mc{Y} \times \mc{V} $, with high probability. Equivalently, the sequences of symbols $(U^n ,X^n,Y^n,V^n) \in A_{\varepsilon}^{{\star}{n}}(\QQ)$ are jointly typical for the probability distribution $\QQ$, with high probability. The encoder is non-causal $X^n = f(U^n)$ and the decoder is strictly causal $V_i = g_i(Y^{i-1})$, for all instant $i\in\{1, \ldots ,n\}$. We characterize the set of joint probability distributions $\QQ(u,x,y,v)$, that are achievable.}
\label{fig:1CwithChannel}
\end{center}
\end{figure}
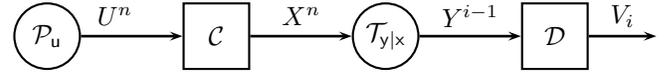
{Second}, we consider that the autonomous devices are equipped with some utility function, capturing their objective. The set of achievable values of the utility function is the  image, by the expectation operator, of the set of achievable probability distributions. A utility value is achievable if and only if it corresponds to the expected utility, with respect to some achievable probability distribution. This approach simplifies the optimization of the long-run utility function, since the whole set of possible codes of large block-length reduces to the set of achievable target probability distributions, expressed with a single-letter formulation. As a particular case, our results boil down to the classical results of Shannon \cite{shannon-bell-1948}, when considering the  minimal distortion for the information source or the minimal cost for the channel inputs. In this paper, we consider a  utility function that is common to both encoder and decoder. The problem of strategical coordination, involving distinct utility functions, is considered in \cite{LeTreustTomala(Allerton)16}, using a repeated game approach \cite{AM95}, \cite{MertensSorinZamir15}, \cite{GossnerTomala(RG)07}. 

The notion of target probability distribution has been proposed by Wyner in  \cite{Wyner(CommonInfo)1975} for determining the common information of two correlated random variables. In the framework of quantum coding, the authors of \cite{KramerSavari02} and \cite{KramerSavari07} prove the existence of a code with minimal amount of information  exchanged, such that the empirical distribution of the sequences of symbols is close to the target probability distribution. The problem of empirical coordination is studied in  \cite{CoverPermuter07} for a three-node cascade network, for a multi-agent chain and for a multiple description setting. A stronger definition of coordination is considered in \cite{Cuff08} and \cite{Cuff(DistributedChannelSynth)13} for the problem of simulating and synthesizing a discrete memoryless channel, also related to the ``reverse Shannon Theorem'' \cite{Bennett02} and ``channel resolvability'' \cite{HanVerdu93}. The concept of ``coordination capacity'' is introduced in \cite{CuffPermuterCover10}, as a measure of the minimal amount of  information transmitted,  such that the nodes of a network can coordinate their actions. The authors consider two different notions of coordination, referred to as empirical coordination and strong coordination. For some networks, both notions coincide if the nodes have enough common randomness. Coordination over a network is also related to the multi-agent control problem with common randomness \cite{AnantharamBorkar07}. In \cite{GohariAnantharam11}, the authors investigate strong coordination of the actions, assuming that the nodes have multiple rounds of noise-free communication. Empirical distribution of  sub-codewords of length $k$ of a \textit{good code} is considered  in \cite{WeissmanOrdentlichISIT04} and \cite{WeissmanOrdentlich05},  and the authors prove that it converges to the product of the optimal probability distribution \cite{ShamaiVerdu97}.  Polar codes are under investigation  for empirical coordination in \cite{BlascoThobabenSkoglund12}, and for strong coordination  in \cite{BlochLuzziKliewer12}. In \cite{ChouBlochKliewer15}, the authors provide encoding and decoding algorithms that are based on polar codes, and that achieve the empirical and the strong coordination capacity. In \cite{CerviaLuzziBlochLeTreust(ITW)16}, the authors construct a polar code for empirical coordination with a noisy channel. Empirical coordination for triangular multi-terminal network is investigated in \cite{BereyhiBahramiMirmohseniAref13}. Strong coordination is studied for a multi-hop line network in \cite{VellambiKliewerBloch(ITW)15} and \cite{VellambiKliewerBloch(CISS)16}, for a three-terminal line network in \cite{SatpathyCuff13} and \cite{BlochKliewer13}, for a three-terminal relay network in \cite{BlochKliewer14}, for two-way communication with a relay in \cite{HaddadpourYassaeeGohariAref12}, and for signal's coordination in \cite{CerviaLuzziLeTreustBloch(ISIT)17}. The source coding problem of Ahlswede-K\"{o}rner \cite{Ahlswede-korner-it-1975} is investigated in \cite{GoldfeldPermuterKramer14}, with a  coordination requirement. The results for empirical coordination are extended to general alphabets in \cite{RaginskyISIT10} and \cite{Raginsky13}, by considering standard Borel spaces. The problems of zero-error coordination  \cite{AbroshanGohariJaggi(ITW)15} and  of strong coordination with an evaluation function \cite{OrlitskyRoche01} are both related to graph theory.

%\textcolor[rgb]{0.00,0.00,1.00}{

Coordination is also investigated in the literature of game theory \cite{GossnerVieille02}, \cite{GossnerTomala06}, \cite{GossnerTomala07}, \cite{GossnerLarakiTomala09}, \cite{GossnerHernandezNeyman06} using the notion of \textit{implementable probability distribution}, that is related to empirical coordination. In \cite{GossnerHernandezNeyman06}, the authors consider a point-to-point scenario with an encoder that observes the sequence of symbols of source, called ``state of the nature'', and that chooses a sequence of actions. The channel is perfect and the decoder is strictly causal \textit{i.e.}, it returns an action based on the observation of the past actions of the encoder and past symbols of the source. The objective is to coordinate the actions of both players together with the symbols of the source. The main difference with the  settings described previously is that the channel inputs are also coordinated with the symbols of the source and the decoder's actions. The encoder chooses a sequence of channel inputs that conveys some information and that is coordinated with the sequences of symbols of source.  The authors  characterize the set of implementable target joint probability distributions and evaluate the long-run utility function of the players, by considering the expected utility. 
%Empirical coordination was also investigated in the literature of game theory  \cite{GossnerHernandezNeyman06} for a  point-to-point communication model with an information source and a perfect channel. The encoder observes the sequence of symbols of source and chooses a sequence of actions. The channel is perfect and the decoder is strictly causal \textit{i.e.}, it returns an action based on the observation of the past symbols of the encoder and of the past symbols of the source. The authors of \cite{GossnerHernandezNeyman06} considered the empirical coordination of the actions of both players together with the symbols of the source. 
The results of \cite{GossnerHernandezNeyman06} are extended  in \cite{SamsonBenjaminISIT13}, \cite{LarrousseLasaulceBloch(IT)14}, by considering a noisy channel. The authors characterize the set of implementable probability distributions and apply their result to the interference channel in which the power control is used to encode embedded data about the channel state information. This approach is further applied to the two-way channel in \cite{LarrousseLasaulceWigger(ITW)15}, and to the case of causal encoding and decoding in \cite{LarrousseLasaulceWigger(ISIT)15}. The results of \cite{GossnerHernandezNeyman06} have also been extended in \cite{Cuff(ImplicitCoordination)11} by considering the notion of empirical coordination and by removing the observation by the decoder of the past symbols of source. The tools for empirical coordination with a cascade of controllers \cite{Cuff(ImplicitCoordination)11} are also used in  \cite{AsnaniPermuterWeissman13}, for  the problem of cooperation in multi-terminal source coding with strictly causal, causal, and non-causal cribbing. In \cite{LetreustZaidiLasaulce(Allerton)11} and \cite[pp. 121]{LetreustThese11}, the authors investigate the empirical correlation for two dependent sources and a broadcast channel with an additional secrecy constraint. The problem of empirical coordination for a joint source-channel coding problem is solved in \cite{CuffSchieler11}, for strictly causal and causal encoder with non-causal decoder. These results are based on hybrid coding \cite{MineroLimKim(ISIT)11}, \cite{LimMineroKim(Allerton)10}, and are closely related to the problem of state communication  under investigation in \cite{ChoudhuriKimMitra10}, \cite{ChoudhuriKimMitra11}, \cite{ChoudhuriMitra(Allerton)12}. The results stated in \cite{CuffSchieler11} are extended in  \cite{LeTreust(ISITfeedbacks)15} by considering channel feedback available at the encoder. Channel feedback improves the coordination possibilities and simplifies the information constraint. For this problem, the authors of \cite{LeTreustBloch(ISIT)16} characterize of the region of achievable triple of information rate, empirical distribution and state leakage, with and without state information and noisy channel feedback.  The problem of empirical coordination for non-causal encoder and decoder is not yet completely solved, but the optimal solutions are characterized in \cite{LeTreust(CorrelationITW)14} for lossless decoding and in \cite{LeTreust(ISIT-TwoSided)15} for perfect channel and for independent source and channel, based the separation result of \cite{MerhavShamai03}. The duality \cite{CoverChiang02} between the channel coding of Gel'fand Pinsker  \cite{gelfand-it-1980} and the source coding  of Wyner Ziv \cite{wyner-it-1976} induces some similarities in the information constraints for lossless decoding \cite{LeTreust(CorrelationITW)14} and for perfect channel \cite{LeTreust(ISIT-TwoSided)15}. This open problem is closely related to the problem of non-causal state communication, under investigation in \cite{SutivongChiangCoverKim05}, \cite{ChoudhuriKimMitra13} and \cite{SutivongPhD03}. The problem of empirical coordination is a first step towards a better understanding of decentralized communication networks, in which the devices have different utility functions \cite{LeTreustTomala(Allerton)16} and choose autonomously the transmission power \cite{LeTreustLasaulce(PowerControlRG)10}, \cite{BelmegaLasaulceDebbah09} and the transmission rates \cite{BerryTse(ShannonMetNash)11}, \cite{PerlazaTandonPoorHan12}, depending on their observation structure \cite{LetreustLasaulce(DGAA)12}.

In this paper, we investigate a point-to-point coordination problem involving an i.i.d. information source and a memoryless channel, represented by Fig. \ref{fig:1CwithChannel}. The encoder and the decoder choose their sequences of actions \textit{i.e.}, channel input and decoder output, so that the empirical distribution of the symbols converges to a target probability distribution. We assume that the decoder is strictly causal \textit{i.e.}, at each instant, it returns a symbol, also called an action, based on the observation of the \textit{past} channel outputs. This on-line coordination assumption  is related to the game theoretical framework \cite{GossnerHernandezNeyman06}, in which the encoder and the decoder are the players that choose their actions simultaneously, based on their past observations. Strictly causal decoding  has no impact on the information constraint for reliable transmission but it modifies the information constraint for empirical coordination. We characterize the set of achievable target joint probability distributions for non-causal encoder and strictly causal decoder and we relate the corresponding information constraint to the previous results from the literature, especially with the problem of source distortion and channel cost  \cite[pp. 47, 57 and 66]{ElGammalKim(book)11}. We analyze the optimization of some utility function over the set of achievable target probability distributions and we prove that this problem is convex. We also characterize the information constraint corresponding to causal decoder instead of strictly causal decoder. In that case, the actions of the decoder may also be coordinated with the \textit{current} channel output.

The article is organized as follows. Sec. \ref{sec:ModelDefinition} presents the channel model under investigation and defines the notion of achievable empirical distribution and strictly causal decoding. In Sec. \ref{sec:SCD}, we characterize the set of achievable target probability distributions  for non-causal encoder and strictly causal decoder. In Sec. \ref{sec:ParticularCases}, we compare our characterization to the previous results of the literature, for perfect channel and for independent random variables of source and channel. In Sec. \ref{sec:ExtensionsSF}, we investigate empirical coordination with source feedforward and in Sec. \ref{sec:ExtensionsT} we characterize the trade-off between empirical coordination and information transmission. In Sec. \ref{sec:Utilities}, we characterize the set of achievable utilities and we prove that the corresponding optimization problem is convex. We investigate two examples: the coordination game in Sec. \ref{sec:ExampleMatchingPennies} and the trade-off between source distortion and channel cost in Sec. \ref{sec:SourceDistortionCost}. In Sec. \ref{sec:CwithChannel}, we characterize the information constraint  for  causal decoding, instead of strictly causal decoding. Sec. \ref{sec:Conclusion} concludes the article. The proof of the main results are presented in App. \ref{sec:ProofDecomposition}-\ref{sec:ProofTheoConvexCD}.

%%%%%%%%Table of contents
%\ref{sec:ModelDefinition}
%\ref{sec:MainResult}, 
%
%\ref{sec:SCD}, 
%\ref{sec:ParticularCases}
%\ref{sec:Extensions}, 
%
%%%%
%
%\ref{sec:ConvexOptimization}, 
%
%\ref{sec:Utilities}
%\ref{sec:ExampleMatchingPennies}, 
%\ref{sec:SourceDistortionCost}
%%%%
%
%\ref{sec:CwithChannel}, 
%
%%%%
%
%\ref{sec:Conclusion}
%
%%%%
%
%\ref{sec:ProofDecomposition}, 
%\ref{sec:ProofAchievability}
%\ref{sec:ProofEqualityIC}
%\ref{sec:ProofConverse}, 
%\ref{sec:CardinalityBound}, 
%\ref{sec:ProofCorollaryTrans}, 
%\ref{sec:ProofCoroUtility}, 
%\ref{sec:ProofTheoConvex}
%\ref{sec:ProofDecompositionCD}, 
%\ref{sec:ProofAchievabilityC}
%\ref{sec:ProofCEqualityIC}
%\ref{sec:ProofConverseC}, 
%\ref{sec:CardinalityBoundCD}, 
%\ref{sec:ProofTheoConvexCD}
%
%%%%%%%%%%%%%%%

\section{System model}\label{sec:ModelDefinition}

The problem under investigation is depicted in Fig. \ref{fig:1CwithChannel}. Capital letters like $U$ denote random variables, calligraphic fonts like $\mc{U}$ denote alphabets and lowercase letters like $u\in\mc{U}$ denote the realizations of random variables. We denote by $U^n$, $X^n$, $Y^n$, $V^n$ the sequences of random variables of the source symbols $u^n=(u_1,\ldots,u_n)\in\mc{U}^n$, of channel inputs $x^n\in\mc{X}^n$, of channel outputs $y^n\in\mc{Y}^n$ and of outputs of the decoder $v^n\in\mc{V}^n$. We assume the sets $\mc{U}$, $\mc{X}$, $\mc{Y}$ and $\mc{V}$ are discrete and $\mc{U}^n$ denotes the $n$-time cartesian product of set $\mc{U}$.  The notation $\Delta(\mc{X})$ stands for the set of the probability distributions over the set $\mc{X}$. The variational distance between two probability distributions $\QQ$ and $\PP$ is denoted by $||\QQ - \PP||_{\sf{tv}}= 1/2\cdot \sum_{x\in\mc{X}} |\QQ(x) - \PP(x)|$, see in \cite[pp. 370]{cover-book-2006} and in \cite[pp. 44]{CsiszarKorner(Book)11}. With a slight abuse of notation, we denote by $\QQ(x,v) \in \Delta(\mc{X}\times \mc{V})$ the joint probability distribution over $\mc{X}\times \mc{V}$. The notation $\UN(v=u)$ denotes the indicator function, that is equal to 1 if $v=u$ and 0 otherwise. We use the notation $Y  -\!\!\!\!\minuso\!\!\!\!-X    -\!\!\!\!\minuso\!\!\!\!-  U$ to denote the Markov chain property: $\PP(y|x,u) = \PP(y|x)$ for all $(u,x,y)$. The notation $A_{\varepsilon}^{{\star}{n}}(\QQ)$ denotes the set of sequences $(u^n,x^n,y^n,v^n)$ that are jointly typical with tolerance $\varepsilon>0$, for the probability distribution $\QQ \in \Delta(\mc{U} \times \mc{X}\times \mc{Y}\times \mc{V} )$, as stated  in \cite[pp. 25]{ElGammalKim(book)11}. The information source has i.i.d. probability distribution $\PP_{\sf{u}}(u)$  and the channel is memoryless  with conditional probability distribution $\mc{T}(y|x)$. The statistics of $\PP_{\sf{u}}(u)$ and $\mc{T}(y|x)$ are known by both encoder $\C$ and decoder $\D$.
%$\mc{T}_{\sf{y|x}}$

Coding Process: A sequence of source symbols $u^n\in\mc{U}^n$ is drawn from the i.i.d. probability distribution denoted by $\PP_{\sf{u}}^{\times n} (u^n) \in \Delta(\mc{U}^n)$ and   defined by equation \eqref{eq:SourceProba}. The non-causal encoder $\C$ observes $u^n\in\mc{U}^n$  and sends a sequence of channel inputs $x^n\in\mc{X}^n$. The sequence of channel outputs $y^n \in \mc{Y}^n$ is drawn according to the discrete and memoryless channel whose i.i.d. conditional probability distribution is denoted by $\mc{T}^{\times n}(y^n|x^n) : \mc{X}^n \rightarrow \Delta(\mc{Y}^n)$ and defined by equation \eqref{eq:TransitionProba}. 
\begin{eqnarray}
\mc{P}_{\sf{u}}^{\times n}(u^n)&=& \prod_{i=1}^n \mc{P}_{\sf{u}}(u_i),\label{eq:SourceProba}\\
\mc{T}^{\times n}(y^n|x^n)&=& \prod_{i=1}^n \mc{T}(y_i|x_i).\label{eq:TransitionProba}
\end{eqnarray}

We consider that the decoder $\D$ is strictly causal. At instant $i\in\{1,\ldots,n\}$, it observes the sequence of past channel outputs $y^{i-1} = (y_1,\ldots,y_{i-1}) \in \mc{Y}^{i-1}$ and returns an output symbol $v_i \in \mc{V}$. The objective of this work is to characterize the set of empirical distributions $\QQ \in \Delta(\mc{U} \times \mc{X} \times \mc{Y} \times \mc{V}  ) $ that are achievable \textit{i.e.}, for which the encoder and the decoder can implement sequences of symbols  $(U^n ,X^n,Y^n,V^n) \in A_{\varepsilon}^{{\star}{n}}(\QQ)$ that are jointly typical for the probability distribution $\QQ$, with high probability. 
%The definition of typical sequences is stated in \cite[pp. 25]{ElGammalKim(book)11}.

\begin{definition}\label{def:Code}
A  code $c\in\mc{C}(n)$ with non-causal encoder and strictly-causal decoder is a tuple of functions $c=(f,\{g_i\}_{i=1}^n)$ defined 
by:
%by equations \eqref{eq:1CausalCodeSource1} and \eqref{eq:1CausalCodeSource2} .
\begin{eqnarray}
f &:& \mc{U}^n  \longrightarrow \mc{X}^n ,\label{eq:1CausalCodeSource1}\\
g_i &:& \mc{Y}^{i-1}  \longrightarrow \mc{V},\qquad i \in\{ 1,\ldots,n\} .\label{eq:1CausalCodeSource2}
\end{eqnarray}
We denote by $\textsf{N}(u|u^n) = \sum_{i=1}^n \UN(u_i=u)$ the number of occurrences of the symbol $u \in \mc{U}$ in the sequence $u^n \in \mc{U}^n$. The empirical distribution ${Q}^n \in \Delta(\mc{U}\times \mc{X}\times \mc{Y} \times\mc{V})$ of sequences $(u^n,x^n,y^n,v^n)\in \mc{U}^n\times \mc{X}^n \times \mc{Y}^n \times\mc{V}^n$ is defined 
by:
%by equation \eqref{eq:EmpiricalDistribution}.
\begin{eqnarray}
{Q}^n(u,x,y,v) &=& \frac{\textsf{N}(u,x,y,v|u^n,x^n,y^n,v^n)}{n},\nonumber \\
 \forall  (u,x,y,v)& \in& \mc{U}\times \mc{X}\times \mc{Y} \times\mc{V}.\label{eq:EmpiricalDistribution}
\end{eqnarray}
The probability distributions of the source $\PP_{\sf{u}}(u)$, of the channel $\mc{T}(y|x)$ and the code $c \in \mc{C}(n)$ generate the random sequences of symbols $(U^n,X^n,Y^n,V^n)$. Hence, the empirical distribution $Q^n \in \Delta(\mc{U}\times \mc{X}\times \mc{Y} \times\mc{V})$ is a random variable. For an $\varepsilon>0$ and a target single-letter probability distribution $\QQ \in \Delta(\mc{U} \times \mc{X} \times \mc{Y} \times \mc{V}  )$, the error probability $\PP_{\textsf{e}}(c)$ of the code $c\in\mc{C}(n)$ is defined as:
%by equation \eqref{eq:ErrorProba}.
\begin{eqnarray}
\PP_{\textsf{e}}(c) = \PP_c\bigg(\Big|\Big|Q^n - \QQ \Big|\Big|_{\sf{tv}}> \varepsilon\bigg).\label{eq:ErrorProba}
\end{eqnarray}
%where $Q^n \in \Delta(\mc{U}\times \mc{X}\times \mc{Y} \times\mc{V})$ is the random variable of the empirical distribution of the sequences of symbols $(U^n,X^n,Y^n,V^n)$, induced by the code $c \in \mc{C}(n)$ and the probability distributions of the source $\PP_{\sf{u}}$ and of the channel $\mc{T}_{\sf{y|x}}$.
\end{definition}
In this scenario, the error probability of the code is based on the empirical coordination property instead of the lossless or lossy reconstruction of the symbols of source. 

%In the following section, we will see that empirical coordination is closely related to rate distortion theory.

%\begin{remark}[Recall of the Decoder]
%The following two decoding functions are equivalent. 
%\begin{eqnarray}
%&g_i& : \mc{X}^{i-1}  \longrightarrow \mc{V}(i),\qquad i \in \{ 1,\ldots,n\}\\
%&g_i& : \mc{X}^{i-1} \times \mc{V}^{i-1}  \longrightarrow \mc{V}(i),\qquad i \in \{ 1,\ldots,n\}.
%\end{eqnarray}
%Indeed, the first one is a particular case of the second one and knowing only the sequence $x^{i-1}$, it is possible to reconstruct the sequence $v^{i-1}$ by using the past decoding functions $\{g_j\}_{j=1}^{i-1}$.
%\end{remark}

\begin{definition}\label{def:AchievableDistribution}
The target probability distribution $\QQ \in  \Delta(\mc{U} \times \mc{X} \times \mc{Y} \times \mc{V}  )$ is achievable if for all $\varepsilon>0$, there exists an $\bar{n}\in \N$, such that for all $n \geq \bar{n}$, there exists a code $c\in\mc{C}(n)$ with strictly-causal decoder that satisfies:
\begin{eqnarray}
\PP_{\textsf{e}}(c)  = \PP_c\bigg(\Big|\Big|Q^n - \QQ \Big|\Big|_{\sf{tv}}> \varepsilon\bigg)\leq \varepsilon.\label{eq:AchievableDistribution}
\end{eqnarray}
\end{definition}
The notion of convergence of the Definition \ref{def:AchievableDistribution} is based on the  Ky Fan metric and it is equivalent to the convergence in probability, see in \cite[pp. 289, Theorem 9.2.2]{Dudley02}. If the error probability $\PP_{\textsf{e}}(c)$ is small, the empirical distribution $Q^n(u,x,y,v)$ is close to the probability distribution $\QQ(u,x,y,v)$, with high probability. In that case, the sequences of symbols are coordinated empirically. 
%The notion of convergence stated in Definition \ref{def:AchievableDistribution}, is based on the  Ky Fan metric and Theorem 9.2.2, pp. 289, in \cite{Dudley02}, states that it is equivalent to the convergence in probability.

%%%%%%%%%%%%%%%%%%%%%%%%%%%%%%%%%%%%%%%%%%%%%%%%%%%%%%%%%%%%%%%%%%%%%%%%%%%%%%%%%%%%%%%%%%%%

%\section{Optimal solution and particular cases}\label{sec:MainResult}
%\subsection{Information constraint for strictly causal decoding}\label{sec:SCD}

\section{Main results}\label{sec:MainResult}
\subsection{Characterization of the set of achievable probability distributions}\label{sec:SCD}

We fix the probability distribution of the source $\PP_{\sf{u}}(u) $ and the conditional probability distribution of the channel  $\mc{T}(y|x)$. %We fix the probability distribution of the source $\PP_{\sf{u}} \in  \Delta(\mc{U})$ and the conditional probability distribution of the channel $\mc{T}: \mc{X} \rightarrow \Delta(\mc{Y})$. 
For a non-causal encoder and a strictly causal decoder, we characterize the set of achievable probability distributions $\QQ(u,x,y,v) \in  \Delta(\mc{U} \times \mc{X} \times \mc{Y} \times \mc{V}  )$.

\begin{theorem}[Strictly causal decoding] \label{theo:1CwithChannel}
$\qquad$\\
A joint probability distribution $\QQ(u,x,y,v) \in  \Delta(\mc{U} \times \mc{X} \times \mc{Y} \times \mc{V}  )$ is achievable if and only if the two following conditions are satisfied:\\
1) It decomposes as follows:
\begin{eqnarray}
\QQ(u,x,y,v) = \PP_{\sf{u}}(u)   \times \QQ(x,v | u) \times  \mc{T}(y | x ),\label{eq:1CwithChannel0}
\end{eqnarray}
2) There exists an auxiliary random variable $W \in \mc{W}$ such that: 
\begin{eqnarray}
\max_{{\QQ}\in \Q} \bigg( I( W;Y  |V )  -   I( U ; V  ,W  )  \bigg) \geq 0,
\label{eq:1CwithChannel1}
\end{eqnarray}
where $\Q$ is the set of  joint probability distributions ${\QQ}(u,x,w,y,v)\in  \Delta(\mc{U} \times \mc{X}  \times \mc{W} \times \mc{Y} \times \mc{V}  )$ that decompose as follows:
\begin{eqnarray}
\PP_{\sf{u}}(u)   \times \QQ(x,v | u)  \times \QQ(w| u,x,v) \times  \mc{T}(y | x ) \label{eq:1CwithChannel5}
\end{eqnarray}
and the support of $W$ is bounded by: $ |\mc{W}| \leq  |\mc{U} \times   \mc{X} \times \mc{V}| +1$. 
\end{theorem}

%%%%%%%%%%%%%%%%%%%%%%%%%%%%%%%%%%%%%%%%%%%%%%%%%%%%%%
%%%%%%%%%%%%%%%%%%%%%%%%%%%%%%%%%%%%%%%%%%%%%%%%%%%%%%

The proof of Theorem \ref{theo:1CwithChannel} is stated in App. \ref{sec:ProofDecomposition}- \ref{sec:CardinalityBound}. App. \ref{sec:ProofDecomposition} proves the decomposition of the probability distribution \eqref{eq:1CwithChannel0}, App. \ref{sec:ProofAchievability} proves the achievability result when the information constraint is strictly positive. The case of information constraint equal to zero is stated in App. \ref{sec:ProofEqualityIC}. App. \ref{sec:ProofConverse} provides the converse result and App. \ref{sec:CardinalityBound} proves the upper bound on the cardinality of the support of the auxiliary random variable $W$. The information constraint \eqref{eq:1CwithChannel1} was also obtained in \cite{LarrousseLasaulceBloch(IT)14} and \cite{LarrousseLasaulceWigger(ISIT)15}, for  implementable probability distribution instead of empirical coordination. Since the target probability distribution $\PP_{\sf{u}}(u)   \times  {\QQ}(x,v | u) \times  \mc{T}(y | x )$ is fixed, the maximum in equation \eqref{eq:1CwithChannel1} can be also taken over the conditional probability distributions ${\QQ}(w | u,x,v)$ such that $ |\mc{W}| \leq  |\mc{U} \times   \mc{X} \times \mc{V}| +1$.

\textit{Proof ideas:} 

$\bullet$ The achievability proof is inspired by the block-Markov code represented by Fig. 6 in \cite{LetreustZaidiLasaulce(Allerton)11}. It is the concatenation of a lossy source coding, for the block $b+1$ and a channel coding with two-sided state information, for the block $b$.  The auxiliary random variable $W$ characterizes the trade-off between the correlation of  the channel input $X$ with the pair $(U,V)$, considered as channel-states.

$\bullet$ The converse is based on Csisz\'{a}r Sum Identity \cite[pp. 25]{ElGammalKim(book)11}, with the identification of the auxiliary random variable $W_{i} = (Y^{i-1}  , U^n_{i+1})$, and using Fano's inequality for empirical coordination.

\begin{example}[Empirical distribution of sequences]\label{Example:EmpDistr}
$\qquad$\\
We consider a binary information source with uniform probability distribution $\PP(U=0) = \PP(U=1)=0.5$, and a perfect channel $Y=X$. Suppose that a coding scheme induces the following sequences of symbols $ \textcolor[rgb]{0.00,0,0.00}{(U^n,X^n,V^n)}$, of length $n=12$:
\begin{small}
 \begin{figure}[ht]
\begin{center}
\psset{xunit=0.5cm,yunit=0.5cm}
\begin{pspicture}(0,1)(12,4)
\psframe[linecolor = black](0,1)(12,4)
\psline[linecolor = black](0,2)(12,2)
\psline[linecolor = black](0,3)(12,3)
\psline[linecolor = black](0,4)(12,4)
\psline[linecolor = black, linewidth = 1pt](1,4)(1,1)
\psline[linecolor = black, linewidth = 1pt](2,4)(2,1)
\psline[linecolor = black, linewidth = 1pt](3,4)(3,1)
\psline[linecolor = black](4,4)(4,1)
\psline[linecolor = black](5,4)(5,1)
\psline[linecolor = black](6,4)(6,1)
\psline[linecolor = black](7,4)(7,1)
\psline[linecolor = black](8,4)(8,1)
\psline[linecolor = black](9,4)(9,1)
\psline[linecolor = black](10,4)(10,1)
\psline[linecolor = black](11,4)(11,1)
\rput(-1,3.5){$ \textcolor[rgb]{0.00,0,0.00}{U^n}$}
\rput(-1,2.5){$ \textcolor[rgb]{0.00,0,0.00}{X^n}$}
\rput(-1,1.5){$ \textcolor[rgb]{0.00,0,0.00}{V^n}$}
\rput(0.5,3.5){$0$}
\rput(0.5,2.5){$0$}
\rput(0.5,1.5){$0$}
\rput(1.5,3.5){$0$}
\rput(1.5,2.5){$1$}
\rput(1.5,1.5){$0$}
\rput(2.5,3.5){$1$}
\rput(2.5,2.5){$1$}
\rput(2.5,1.5){$0$}
\rput(3.5,3.5){$1$}
\rput(3.5,2.5){$1$}
\rput(3.5,1.5){$1$}
\rput(4.5,3.5){$0$}
\rput(4.5,2.5){$0$}
\rput(4.5,1.5){$1$}
\rput(5.5,3.5){$0$}
\rput(5.5,2.5){$0$}
\rput(5.5,1.5){$0$}
\rput(6.5,3.5){$1$}
\rput(6.5,2.5){$1$}
\rput(6.5,1.5){$1$}
\rput(7.5,3.5){$1$}
\rput(7.5,2.5){$0$}
\rput(7.5,1.5){$0$}
\rput(8.5,3.5){$0$}
\rput(8.5,2.5){$0$}
\rput(8.5,1.5){$0$}
\rput(9.5,3.5){$0$}
\rput(9.5,2.5){$1$}
\rput(9.5,1.5){$1$}
\rput(10.5,3.5){$1$}
\rput(10.5,2.5){$0$}
\rput(10.5,1.5){$1$}
\rput(11.5,3.5){$1$}
\rput(11.5,2.5){$1$}
\rput(11.5,1.5){$1$}
\rput(5,0.5){$n=12$}
\psframe*[linecolor = yellow](0,1)(1,4)
\psframe[linecolor = black](0,1)(1,4)
\psline[linecolor = black](0,2)(1,2)
\psline[linecolor = black](0,3)(1,3)
\small
\rput(0.5,3.5){$0$}
\rput(0.5,2.5){$0$}
\rput(0.5,1.5){$0$}
\psframe*[linecolor = yellow](5,1)(6,4)
\psframe[linecolor = black](5,1)(6,4)
\psline[linecolor = black](5,2)(6,2)
\psline[linecolor = black](5,3)(6,3)
\rput(5.5,3.5){$0$}
\rput(5.5,2.5){$0$}
\rput(5.5,1.5){$0$}
\psframe*[linecolor = yellow](8,1)(9,4)
\psframe[linecolor = black](8,1)(9,4)
\psline[linecolor = black](8,2)(9,2)
\psline[linecolor = black](8,3)(9,3)
\rput(8.5,3.5){$0$}
\rput(8.5,2.5){$0$}
\rput(8.5,1.5){$0$}
\psframe*[linecolor = newlightblue](3,1)(4,4)
\psframe[linecolor = black](3,1)(4,4)
\psline[linecolor = black](3,2)(4,2)
\psline[linecolor = black](3,3)(4,3)
\small
\rput(3.5,3.5){$1$}
\rput(3.5,2.5){$1$}
\rput(3.5,1.5){$1$}
\psframe*[linecolor = newlightblue](6,1)(7,4)
\psframe[linecolor = black](6,1)(7,4)
\psline[linecolor = black](6,2)(7,2)
\psline[linecolor = black](6,3)(7,3)
\rput(6.5,3.5){$1$}
\rput(6.5,2.5){$1$}
\rput(6.5,1.5){$1$}
\psframe*[linecolor = newlightblue](11,1)(12,4)
\psframe[linecolor = black](11,1)(12,4)
\psline[linecolor = black](11,2)(12,2)
\psline[linecolor = black](11,3)(12,3)
\rput(11.5,3.5){$1$}
\rput(11.5,2.5){$1$}
\rput(11.5,1.5){$1$}
\end{pspicture}
\end{center}
\end{figure}
\end{small}\\
The empirical distribution $\textcolor[rgb]{0,0,0}{Q^n} \in \Delta( \mc{U} \times   \mc{X} \times   \mc{V} )$ induced by the sequences of symbols $ \textcolor[rgb]{0.00,0.0,0.00}{(U^n,X^n,V^n)}$, is represented by Fig. \ref{fig:Distribution}. 
\begin{small}
 \begin{figure}[ht]
\begin{center}
\psset{xunit=0.75cm,yunit=0.75cm}
\begin{pspicture}(13,0.5)(23,4.3)
\psframe[linecolor = black](14.5,1)(17.5,4)
\psframe[linecolor = black](19,1)(22,4)
\psline[linecolor = black](16,1)(16,4)
\psline[linecolor = black](20.5,1)(20.5,4)
\psline[linecolor = black](14.5,2.5)(17.5,2.5)
\psline[linecolor = black](19,2.5)(22,2.5)
\footnotesize
\rput(13.75,1.75){$X=1$}
\rput(13.75,3.25){$X=0$}
\rput(18.25,1.75){$X=1$}
\rput(18.25,3.25){$X=0$}
\rput(16,0.4){$U=0$}
\rput(20.5,0.4){$U=1$}
\rput(19.75,4.25){$V=0$}
\rput(21.25,4.25){$V=1$}
\rput(15.25,4.25){$V=0$}
\rput(16.75,4.25){$V=1$}
\normalsize
\rput(15.25,3.25){$\textcolor[rgb]{1,0,0}{\frac{3}{12}}$}
\rput(21.25,1.75){$\textcolor[rgb]{1,0,0}{\frac{3}{12}}$}
\rput(21.25,3.25){$\textcolor[rgb]{0,0,1}{\frac{1}{12}}$}
\rput(15.25,1.75){$\textcolor[rgb]{0,0,1}{\frac{1}{12}}$}
\rput(16.75,3.25){$\textcolor[rgb]{0,0,1}{\frac{1}{12}}$}
\rput(16.75,1.75){$\textcolor[rgb]{0,0,1}{\frac{1}{12}}$}
\rput(19.75,3.25){$\textcolor[rgb]{0,0,1}{\frac{1}{12}}$}
\rput(19.75,1.75){$\textcolor[rgb]{0,0,1}{\frac{1}{12}}$}
\end{pspicture}
\caption{Empirical distribution $ \textcolor[rgb]{0,0,0}{Q^n}$ of the sequences of symbols $ \textcolor[rgb]{0,0,0}{(U^n,X^n,V^n)}$.}\label{fig:Distribution}
\end{center}
\end{figure}
\end{small}
We evaluate the information constraint \eqref{eq:1CwithChannel1} corresponding to the empirical distribution $ \textcolor[rgb]{0,0,0}{Q^n}$: 
\begin{align}
&\max_{{\QQ}\in \Q} \big( I( W;Y  |V )  -   I( U ; V  ,W  )  \big) \nonumber \\
&= H(X |V)   - I(U ;X,V ) =H(X,U|V) - H(U) \\
&= 1 + \frac{1}{2} \cdot \log_2(3) - 1 = \frac{1}{2} \cdot \log_2(3)\simeq 0.79\geq0.\label{eq:PositiveIC}
\end{align}
The first equality comes from the hypothesis of perfect channel, see \cite{GossnerHernandezNeyman06} and Corollary \ref{coro:PerfectChannel}. The empirical distribution $ \textcolor[rgb]{0,0,0}{Q^n}$ has a positive information constraint \eqref{eq:PositiveIC}.
\end{example}

\begin{remark}[Markov chain]
The strictly causal decoding $V_i = g_i(Y^{i-1})$ induces a Markov chain $Y  -\!\!\!\!\minuso\!\!\!\!-X    -\!\!\!\!\minuso\!\!\!\!-  V$, since at each instant $i\in\{1,\ldots,n\}$, the symbol $V_i$ is generated by the decoder before it observes $Y_i$.
\end{remark}

\begin{remark}[Asynchronous source and channel]\label{remark:bandwidth}
We consider the setting of bandwidth expansion and compression, for which $k_s \in \N $ symbols of source are synchronized with $k_c \in \N $ channel uses. When we introduce super-symbols of source $(\tilde{U}, \tilde{V})  = ( U^{k_s} ,  V^{k_s}) $ corresponding to the sequences of length  $k_s \in \N$ and super-symbols of channel $(\tilde{X}, \tilde{Y})  = ( X^{k_c} ,  Y^{k_c}) $ corresponding to the sequences of length $k_c \in \N$, Theorem III.1 characterizes the achievable empirical distributions between the super-symbols of source and channel $(\tilde{U},\tilde{X}, \tilde{Y}, \tilde{V})$.
\end{remark}

We compare the result of Theorem \ref{theo:1CwithChannel} with previous results stated in the literature.

%%%%%%%%%%%%%%%%%%%%%%%%%%%%%%%%%%%%%%%%%%%%%%%%%%%%%%%%%%%%%%%%%%%%%%%%%%%%%%%%%%%%%%%%%%%%%%%%%%%%%%%%%%%%%%%%%%%%%%%%%%%%%%%%%%%%%%%%%%%%%%%%%%%%%%%%%%%%%%%%%%%%%%%%%%%%%%%%%%%%

%\subsection{Particular case: Independence between the random variables of the source and  of the channel}\label{sec:IndependentSC}
%\subsection{Particular cases:  Perfect channel and source feedforward}\label{sec:PerfectChannel}

\subsection{Particular cases}\label{sec:ParticularCases}

The case of perfect channel $Y=X$ was characterized for \textit{implementable probability distribution} in \cite{GossnerHernandezNeyman06} and for empirical coordination in \cite{Cuff(ImplicitCoordination)11}. When the channel is perfect, the information constraint \eqref{eq:1CwithChannel1} of Theorem \ref{theo:1CwithChannel} reduces to the one of \cite{GossnerHernandezNeyman06} and \cite{Cuff(ImplicitCoordination)11}.

\begin{corollary}[Perfect channel]\label{coro:PerfectChannel}
We consider a perfect channel $Y  =X $. The information constraint \eqref{eq:1CwithChannel1} of Theorem \ref{theo:1CwithChannel} reduces to:
%\begin{small}
\begin{align}
&\max_{{\QQ}\in \Q} \bigg( I( W;Y  |V )  -   I(  U ; V  ,W  )  \bigg)\nonumber \\
&= H(X |V)   - I(U ;X,V ).
\end{align}
%\end{small}
The optimal auxiliary random variable is $W =X $.
\end{corollary}

\begin{proof}[Corollary \ref{coro:PerfectChannel}] 
We consider the information constraint \eqref{eq:1CwithChannel1} of Theorem \ref{theo:1CwithChannel}, with a perfect channel $Y  =X $.
\begin{align}
&I( W;X  |V )  -   I(  U ; V  ,W  )\\
&=  I( W;X  |V )  -   I(  U ; W | V ) -  I(  U ; V  ) \\
&=  H( W| V, U )  -   H( W| V, X )-  I(  U ; V  ) \\
&=  I( W;X  |V, U )  -   I(  W ; U | V, X ) -  I(  U ; V  ) \\
&\leq  H(X  |V, U )  -  I(  U ; V  )  \\
&=  H(X  |V )  -  I(  U ; X,V  ) .
\end{align}
For all auxiliary random variable $W$, the information constraint stated in \cite{GossnerHernandezNeyman06}, \cite{Cuff(ImplicitCoordination)11} is larger than equation \eqref{eq:1CwithChannel1} of Theorem \ref{theo:1CwithChannel}. We conclude the proof of Corollary \ref{coro:PerfectChannel}, by replacing the auxiliary random variable $W =X $, by the channel input.
\begin{eqnarray}
 I( W;X  |V )  -   I(  U ; V  ,W  )  = H(X |V)   - I(U ;X,V )  .
\end{eqnarray}
\end{proof}

The joint source-channel coding result of Shannon \cite{shannon-bell-1948} states that the source $U$ can be recovered by the decoder if and only if:
\begin{eqnarray}
 \max_{\PP(x)} I(X;Y) - H(U) \geq0. \label{eq:ShannonLossless}
\end{eqnarray}
We denote by $\PP^{\star}(x) \in \Delta(\mc{X})$ the probability distribution that achieves the maximum in equation \eqref{eq:ShannonLossless}. Although the sequences of symbols $(U^n,X^n,Y^n,V^n)$ satisfy the Markov chain $U^n  -\!\!\!\!\minuso\!\!\!\!- X^n   -\!\!\!\!\minuso\!\!\!\!-  Y^n  -\!\!\!\!\minuso\!\!\!\!-  V^n $, the empirical distribution writes as a product:
\begin{align}
&\QQ(u,v) \times \QQ(x,y) \nonumber \\ 
&=  \PP_{\sf{u}} (u)\times \UN(v=u)\times\PP^{\star}(x) \times \mc{T}(y|x) ,
\end{align}
In that case, the random variables of the source $(U,V)$ are independent of the random variable of the channel $(X,Y)$.
%distribution for the source $\QQ(u,v) = \PP_{\sf{u}} (u)\times \UN(v=u) $ is not of the distribution $ \QQ(x,y) =\PP_{\sf{x}}^{\star}(x) \times \mc{T}(y|x) $, for the channel. 
Corollary \ref{coro:Shannon} establishes that the information constraint \eqref{eq:1CwithChannel1} of Theorem \ref{theo:1CwithChannel} reduces to the one of Shannon \cite{shannon-bell-1948}, when the target distribution $\QQ(u,x,y,v) =\QQ(u,v) \times \QQ(x,y)$ decomposes as a product.

\begin{corollary}\label{coro:Shannon}
Suppose that the random variables $(U,V)$ are independent of $(X,Y)$. The information constraint \eqref{eq:1CwithChannel1} of Theorem \ref{theo:1CwithChannel} reduces to:
\begin{align}
&\max_{{\QQ}\in \Q} \bigg( I( W;Y  |V )  -   I(  U ; V  ,W  )  \bigg) \nonumber \\
&=  \quad I(X  ;Y  )   -  I(U ;V ) .\label{eq:ShannonIC}
\end{align}
The optimal auxiliary random variable is $W =X $.
\end{corollary}

\begin{proof}%[Corollary \ref{coro:Shannon}]
We consider the following equations:
\begin{align}
&\max_{{\QQ}\in \Q} \bigg( I( W;Y  |V )  -   I(  U ; V  ,W  )  \bigg) \nonumber \\
&=\max_{{\QQ}\in \Q} \bigg( I( W;Y  |V )  -   I(  U ; W  |V )  \bigg)  -     I(  U ; V   )\label{eq:InequalityIC1}  \\
&\leq\max_{{\QQ}\in \Q} \bigg( I( W,V;Y   )    \bigg)  -     I(  U ; V   )\label{eq:InequalityIC1b}  \\
&\leq  \quad I(X  ;Y  )   -  I(U ;V ) .\label{eq:InequalityIC2}
\end{align}
Equation \eqref{eq:InequalityIC2} comes from the Markov chain property $Y  -\!\!\!\!\minuso\!\!\!\!-X    -\!\!\!\!\minuso\!\!\!\!-  (U,V,W)$ induced by the memoryless channel. We conclude the proof of Corollary \ref{coro:Shannon} by choosing the auxiliary random variable $W =X $, independent of $(U ,V )$:
\begin{align}
 I( W;Y  |V )  -  I(  U ; V  ,W  )  
 &=  I( X;Y  |V )  -  I(  U ; V  ,X  ) \nonumber \\
  &=   I( X;Y   )  -  I(  U ; V    ).
\end{align}
\end{proof}

Corollary \ref{coro:Shannon} also proves that strictly causal decoding has no impact on the information constraint stated by Shannon in \cite{shannon-bell-1948}. In fact the information constraint $I(X  ;Y  )   -  I(U ;V )\geq0$ characterizes the optimal solution for non-causal encoding and strictly causal, causal or non-causal decoding. This is not true anymore when the random variables $(X,Y)$ are empirically coordinated with $(U,V)$.

\begin{remark}%[Relationship between independent target distributions and separated coding]
\bd{(Product of target distributions vs. separation source-channel)}
Corollary \ref{coro:Shannon} shows that when the target distribution $\QQ(u,v)\times \QQ(x,y)$ is a product, then separation of source coding and channel coding is optimal. This remark also holds when considering two-sided channel state information, see \cite[Theorem IV.2]{LeTreust(ISIT-TwoSided)15}. We can ask whether, for the point-to-point model, the optimality of the separation of source coding and channel coding is equivalent to the decomposition of the target distribution into a product.
\end{remark}

\begin{remark}%[Coordination is more restrictive than transmission]
\label{remark:CoordRestrictive}
\bd{(Coordination is more restrictive than information transmission)}
Equation \eqref{eq:InequalityIC2} implies that the information constraint corresponding to $\QQ(x,v|u) $ is lower than the one corresponding to the product of marginals $\QQ(x) \times \QQ(v|u) $. We conclude that the empirical coordination is more restrictive than lossless or lossy transmission of the information source. 
\end{remark}

%%%%%%%%%%%%%%%%%%%%%%%%%%%%%%%%%%%%%%%%%%%%%%%%%%%%%%%%%%%%%%%%%%%%%%%%%%%%%%%%%%%%%%%%%%%%%%%%%%%%%%%%%%%%%%%%%%%%%%%%%%%%%%%%%%%%%%%%%%%%%%%%%%%%%%%%%%%%%%%%%%%%%%%%%%%%%%%%%%%%%%%%%%%%%%%%%%%%%%%%%%%%%%%%%%%%%%%%%%%%%%%%

\subsection{Empirical coordination with source feedforward}\label{sec:ExtensionsSF}

We consider the scenario with \textit{source feedforward}  corresponding to the definition \ref{def:CodeTransmission}, represented by Fig. \ref{fig:feedforward}.

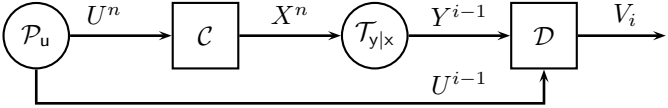
\begin{figure}[!ht]
\begin{center}
\psset{xunit=0.9cm,yunit=0.9cm}
\begin{pspicture}(0.3,-0.5)(8.5,1.5)
\pscircle(0,0.5){0.449}
%\pscircle(0,1.5){0.449}
\psframe(2,0)(3,1)
\pscircle(5,0.5){0.45}
\psframe(7,0)(8,1)
\psline[linewidth=1pt]{->}(0.5,0.5)(2,0.5)
\psline[linewidth=1pt]{->}(3,0.5)(4.5,0.5)
\psline[linewidth=1pt]{->}(5.5,0.5)(7,0.5)
\psline[linewidth=1pt]{->}(8,0.5)(9.3,0.5)
\psline[linewidth=1pt]{->}(0,0)(0,-0.5)(7.5,-0.5)(7.5,0)
%\psline[linewidth=1pt]{->}(2.5,-0.5)(2.5,0)
%\psline[linewidth=1pt]{->}(0.5,1.5)(2.5,1.5)(2.5,1)
\rput[u](1,0.8){$U^n$}
\rput[u](3.75,0.8){$X^n$}
\rput[u](6.25,0.8){$Y^{i-1}$}
\rput[u](6.25,-0.2){$U^{i-1}$}
\rput[u](8.7,0.8){$V_i  $}
%\rput[u](8.5,0.1){$\hat{M} $}
\rput(0,0.5){$\PP_{\sf{u}}$}
%\rput(0,1.5){$\PP_{M}$}
\rput(2.5,0.5){$\C$}
\rput(5,0.5){$\mc{T}_{\sf{y|x}}$}
\rput(7.5,0.5){$\D$}
\end{pspicture}
\caption{Source feedforward: the decoder generates a symbol $V_i$ based on the observation of the pair of sequences $(Y^{i-1},U^{i-1})$. }\label{fig:feedforward}
\end{center}
\end{figure}

\begin{definition}\label{def:CodeTransmission}
A  code $c\in\mc{C}_{\sf{s}}(n)$ with source feedforward is defined by:
\begin{eqnarray}
f &:& \mc{U}^n   \longrightarrow \mc{X}^n ,\label{eq:CodeTransS1}\\
g_i &:& \mc{Y}^{i-1}  \times \mc{U}^{i-1}\longrightarrow \mc{V},\qquad i \in \{ 1,\ldots,n\},\label{eq:CodeTransS2}.\label{eq:CodeTrans3}
\end{eqnarray}
A target probability distribution $\QQ(u,x,y,v)$ is achievable if
\begin{eqnarray}
&&\forall \varepsilon>0,\;\; \exists\bar{n}\in \N,\;\; \forall n\geq \bar{n},\;\; \exists c\in\mc{C}_{\sf{s}}(n),\nonumber\\
&&\PP_{\textsf{e}}(c)  = \PP_c\bigg(\Big|\Big|Q^n - \QQ \Big|\Big|_{\sf{tv}}> \varepsilon\bigg)\leq \varepsilon.\label{eq:AchievableDistributionS}
\end{eqnarray}
\end{definition}

The set of implementable probability distributions was characterized for a perfect channel in \cite{GossnerHernandezNeyman06} and for a noisy channel in \cite{LarrousseLasaulceBloch(IT)14}, where the target probability distribution $\PP_{\sf{u}} (u)\times \QQ(x,v|u) \times \mc{T}(y|x)$ is achievable if and only if: $I(X;Y|U,V)\geq I(U;V)$.
\begin{corollary}[Source feedforward]\label{coro:FeedbackSource}
We consider a probability distribution $\PP_{\sf{u}} (u)\times \QQ(x,v|u) \times \mc{T}(y|x)$. Then we have: 
\begin{eqnarray}
&&\max_{{\QQ}\in \Q} \bigg( I( W; U,Y  |V )  -   I(  U ; V  ,W  )  \bigg)\label{eq:ICfeedforward2}\\
 &=&I(X  ;Y  |   U  ,V )   - I(U ;V )\label{eq:ICfeedforward1}.
\end{eqnarray}
Here $\Q$  is the set of probability distributions ${\QQ}   \in  \Delta(\mc{U}   \times \mc{W} \times \mc{X} \times \mc{Y} \times \mc{V}  )$ defined in Theorem \ref{theo:1CwithChannel}. The optimal auxiliary random variable  in \eqref{eq:ICfeedforward2} is $W =X $.
\end{corollary}

The main difference with the problem stated in Sec. \ref{sec:MainResult} is that the decoder observes the sequences $(Y^{i-1}, U^{i-1})$ instead of $Y^{i-1}$. If we replace the symbol $Y$ by the pair $(Y,U)$ in the information constraint \eqref{eq:1CwithChannel1} of Theorem \ref{theo:1CwithChannel}, then we obtain equation \eqref{eq:ICfeedforward2} that boils down to the information constraint \eqref{eq:ICfeedforward1} of \cite{LarrousseLasaulceBloch(IT)14}.

\begin{proof}[Corollary \ref{coro:FeedbackSource}]
We have the following equations:
\begin{align}
&   \max_{{\QQ}\in \Q} \bigg( I( W;Y , U  |V )  -   I(  U ; V  ,W  )  \bigg) \\
&=  \max_{{\QQ}\in \Q} \bigg(  I( W;Y | U  ,V ) +  I( W; U  |V ) -   I(  U ; V  ,W  )  \bigg)\\
&=  \max_{{\QQ}\in \Q} \bigg(  I( W;Y | U  ,V )  -   I(  U ; V  )  \bigg)\\
&=   I( X;Y | U  ,V )  -   I(  U ; V  )  .
\end{align}
The Markov chain of the channel $Y  -\!\!\!\!\minuso\!\!\!\!-X    -\!\!\!\!\minuso\!\!\!\!-  (W ,U,V)$ implies that the term $ I( W;Y | U  ,V ) $ is maximal for the auxiliary random variable $W =X$. This concludes the proof of Corollary \ref{coro:FeedbackSource}.
\end{proof}

\subsection{Trade-off between empirical coordination and information transmission}\label{sec:ExtensionsT}

We investigate the trade-off between reliable transmission of a message $M\in \mc{M}$ and empirical coordination, as depicted in Fig. \ref{fig:CoordTrans}. We consider a positive target information rate $\textsf{R} \geq 0$  and a target joint probability distribution that decomposes as:  $\QQ(u,x,y,v)  = \PP_{\sf{u}}(u)   \times \QQ(x,v | u) \times  \mc{T}(y | x )$.  Corollary \ref{coro:CoordTrans} characterizes the set of achievable pairs of rate and probability distribution $(\textsf{R} , \QQ)$. 

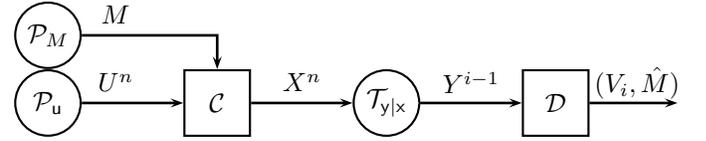
\begin{figure}[!ht]
\begin{center}
\psset{xunit=0.9cm,yunit=0.9cm}
\begin{pspicture}(0.3,0)(8.5,2)
\pscircle(0,0.5){0.449}
\pscircle(0,1.5){0.449}
\psframe(2,0)(3,1)
\pscircle(5,0.5){0.45}
\psframe(7,0)(8,1)
\psline[linewidth=1pt]{->}(0.5,0.5)(2,0.5)
\psline[linewidth=1pt]{->}(3,0.5)(4.5,0.5)
\psline[linewidth=1pt]{->}(5.5,0.5)(7,0.5)
\psline[linewidth=1pt]{->}(8,0.5)(9.3,0.5)
%\psline[linewidth=1pt]{->}(0,0)(0,-0.5)(5,-0.5)(5,0)
%\psline[linewidth=1pt]{->}(2.5,-0.5)(2.5,0)
\psline[linewidth=1pt]{->}(0.5,1.5)(2.5,1.5)(2.5,1)
\rput[u](1,0.8){$U^n$}
\rput[u](1,1.8){$M$}
\rput[u](3.75,0.8){$X^n$}
\rput[u](6.25,0.8){$Y^{i-1}$}
\rput[u](8.7,0.8){$(V_i ,\hat{M}) $}
%\rput[u](8.5,0.1){$\hat{M} $}
\rput(0,0.5){$\PP_{\sf{u}}$}
\rput(0,1.5){$\PP_{M}$}
\rput(2.5,0.5){$\C$}
\rput(5,0.5){$\mc{T}_{\sf{y|x}}$}
\rput(7.5,0.5){$\D$}
\end{pspicture}
\caption{Simultaneous information transmission and empirical coordination.}\label{fig:CoordTrans}
\end{center}
\end{figure}

\begin{definition}\label{def:CodeTransmission}
A  code $c\in\mc{C}_{\sf{i}}(n)$ with  information transmission and strictly-causal decoder is a tuple of functions $c=(f,\{g_i\}_{i=1}^n, g)$ defined by equations \eqref{eq:CodeTrans1}, \eqref{eq:CodeTrans2} and \eqref{eq:CodeTrans3} .
\begin{eqnarray}
f &:& \mc{U}^n \times \mc{M}  \longrightarrow \mc{X}^n ,\label{eq:CodeTrans1}\\
g_i &:& \mc{Y}^{i-1}  \longrightarrow \mc{V},\qquad i \in \{ 1,\ldots,n\},\label{eq:CodeTrans2}\\
g &:& \mc{Y}^{n}  \longrightarrow  \mc{M} .\label{eq:CodeTrans3}
\end{eqnarray}
\end{definition}

\begin{definition}
The pair of rate and probability distribution $(\textsf{R} , \QQ)$ is achievable if for all $\varepsilon>0$, there exists an $\bar{n}\in \N$, such that for all $n \geq \bar{n}$, there exists a code  with information transmission $c\in\mc{C}_{\sf{i}}(n)$, that satisfies:
%\begin{small}
\begin{eqnarray}
&& \PP_c\bigg(\Big|\Big|Q^n - \QQ \Big|\Big|_{\sf{tv}}> \varepsilon\bigg) + \PP_c\bigg( M \neq \hat{M}\bigg)\leq \varepsilon, \label{eq:AchieTrans1} \\
&&\frac{\log_2 |\mc{M}|}{n} \geq  \textsf{R} -\varepsilon.  \label{eq:AchieTrans2}
\end{eqnarray}
%\end{small}
\end{definition}

\begin{corollary}[Information transmission] \label{coro:CoordTrans}
The pair of rate and probability distribution $(\textsf{R} , \QQ)$ is achievable if and only if: 
\begin{eqnarray}
 \max_{{\QQ}\in \Q} \bigg( I( W;Y  |V )  -   I( U ; V  ,W  )  \bigg) \geq&  \textsf{R}& \geq 0.  \label{eq:CoordTrans}
\end{eqnarray}
Here $\Q$ is the set of probability distributions ${\QQ}   \in  \Delta(\mc{U}   \times \mc{W} \times \mc{X} \times \mc{Y} \times \mc{V}  )$ defined in Theorem \ref{theo:1CwithChannel}.
\end{corollary}

The proof of Corollary \ref{coro:CoordTrans} relies on the proof of Theorem \ref{theo:1CwithChannel} and a sketch is stated in App. \ref{sec:ProofCorollaryTrans}. The achievability is based on rate splitting and the converse involves the auxiliary random variable $W_i = (U^n_{i+1} , M, Y^{i-1} )$. Corollary \ref{coro:CoordTrans} characterizes the optimal trade-off between the transmission of information and the empirical coordination. In case of strictly positive information constraint  \eqref{eq:1CwithChannel1}, it is possible to transmit reliably an additional message $M\in \mc{M}$ to the decoder.

\begin{remark}
\label{remark:StateAmplification}
\bd{(Causal Encoding)}
The trade-off between information transmission and empirical coordination was also characterized in \cite[eq. (5)]{LeTreustBloch(ISIT)16}, for the case of causal encoding and non-causal decoding.  
\end{remark}

%%%%%%%%%%%%%%%%%%%%%%%%%%%%%%%%%%%%%%%%%%%%%%%%%%%%%%%%%%%%%%%%%%%%%%%%%%%%%%%%%%%%%%%%%%%%%%%%%%%%%%%%%%%%%%%%%%%%%%%%%%%%%%%%%%%%%%%%%%%%%%%%%%%%%%%%%%%%%%%%%%%%%%%%%%%%%%%%%%%%

%%%%%%%%%%%%%%%%%%%%%%%%%%%%%%%%%%%%%%%%%%%%%%%%%%%%%%%%%%%%%%%%%%%%%%%%%%%%%%%%%%%%%%%%%%%%%%%%%%%%%%%%%%%%%%%%%%%%%%%%%%%%%%%%%%%%%%%%%%%%%%%%%%%%%%%%%%%%%%%%%%%%%%%%%%%%%%%%%%%%%%%%%%%%%%%%%%%%%%%%%%%%%%%%%%%%%%%%%%%%%%%%%%%%%%%%%%%%%%%%%%%%%%%%
%%%%%%%%%%%%%%%%%%%%%%%%%%%%%%%%%%%%%%%%%%%%%%%%%%%%%%%%%%%%%%%%%%%%%%%%%%%%%%%%%%%%%%%%%%%%%%%%%%%%%%%%%%

%%%%%%%%%%%%%%%%%%%%%%%%%%%%%%%%%%%%%%%%%%%%%%%%%%%%%%%%%%%%%%%%%%%%%%%%%%%%%%%%%%%%%%%%%%%%%%%%%%%%%%%%%%%%%%%%%%%%%%%%%%%%%%%%%%%%%%%%%%%%%%%%%%%%%%%%%%%%%%%%%%%%%%%%%%%%%%%%%%%%%%%%%%%%%%%%%%%%%%%%%%%%%%%%%%%%%%%%%%%%%%%%%%%%%%%%%%%%%%%%%%%%%%%%%%%%%%%%%%%%%%%%%%%%%%%%%%%%%%%%%%%%%%%%%%%%%%%%%%%%%%%%%%%%%%%%%%%%%%%%%%%%

%\subsection{Performance of the coding scheme}\label{sec:PerformancesCodingScheme}

\section{Convex optimization problem}\label{sec:ConvexOptimization}

\subsection{Characterization of achievable utilities}\label{sec:Utilities}

In this section, we evaluate the performance of a coding scheme $c\in\mc{C}(n)$, by considering a utility function $\Phi(u,x,y,v)$, defined over the symbols of the source and of the channel:
\begin{eqnarray}
\Phi : \mc{U} \times \mc{X} \times \mc{Y} \times \mc{V}   \longrightarrow \R.\label{eq:UtilityFunction}
\end{eqnarray}
The utility function $\Phi(u,x,y,v)$ is general  and captures the different objectives of the coding process. It can be a distortion function $\Phi(u,x,y,v)=d(u,v)$ for the source coding problem, a  cost function $c(x)$ for the channel coding problem, or a payoff function $\pi(u,x,y,v)$ for the players of a repeated game \cite{GossnerHernandezNeyman06}. The probability distribution of the source $\PP_{\sf{u}}(u)$, the channel conditional probability distribution $ \mc{T}(y | x )$, and the code with strictly causal decoder $c\in\mc{C}(n)$, induce  sequences of random variables $(U^n,X^n,Y^n,V^n)$. Each stage $i \in \{1,\ldots,n\}$, is associated with a stage utility $\Phi (U_i,X_i,Y_i,V_i)$. We evaluate the performance of a code $c\in\mc{C}(n)$ using the $n$-stage utility $\Phi^n(c)$.

\begin{definition}
The  $n$-stage utility $\Phi^n(c)$ of the code $c\in\mc{C}(n)$  is defined by:
\begin{eqnarray}
\Phi^n(c) &=& \E\Bigg[ \frac{1}{n} \cdot  \sum_{i=1}^n \Phi(U_i ,X_i ,Y_i ,V_i )\Bigg].\label{eq:UtilityFunction}
\end{eqnarray}
The expectation is taken over the sequences of random variables $(U^n,X^n,Y^n,V^n)$, induced by the code $c\in\mc{C}(n)$ and by the source $\PP_{\sf{u}}(u)$ and the channel $ \mc{T}(y | x )$. A utility value $\phi\in\R$ is achievable if for all $\varepsilon >0$, there exists a $\bar{n} \in \N$, such that for all $n \geq \bar{n}$, there exists a code $c \in \mc{C}(n)$ such that:
\begin{eqnarray}
\Bigg| \phi - \E\bigg[ \frac{1}{n} \cdot  \sum_{i=1}^n \Phi(U_i ,X_i ,Y_i ,V_i )\bigg] \Bigg| \leq \varepsilon.\label{eq:AchievableUtilityFunction}
\end{eqnarray}
We denote by $\textsf{U}$ the set of achievable utility values $\phi \in \textsf{U}$.
\end{definition}

\begin{theorem}[Set of achievable utilities]\label{theo:AchievableUtility}
The set of achievable utilities $\textsf{U}$ is:
\begin{align}
 \textsf{U} =  \bigg\{\phi\in\R,\text{ s.t. } \exists \;\QQ(x,v|u) , \;\; \E_{ \QQ} \big[\Phi(U,X,Y,V)\big]  = \phi , \nonumber  \\
  \text{ and }  \max_{{\QQ}(w|u,v,x), \atop  |\mc{W}| \leq  |\mc{U} \times   \mc{X} \times \mc{V}| +1 } \bigg( I( W;Y  |V )  -   I( U ; V  ,W  )  \bigg) \geq 0 \bigg\}. \label{eq:AchievableUtilityFunction0} 
\end{align}
\end{theorem}

The proof of Theorem \ref{theo:AchievableUtility} is stated in App. \ref{sec:ProofCoroUtility}. The result of Theorem \ref{theo:AchievableUtility} extends to multiple utility functions $(\Phi_1,\Phi_2,\ldots,\Phi_K)$, as for the trade-off between source distortion $\Phi_1(u,x,y,v) = d(u,v)$ and channel cost $\Phi_2(x) = c(x)$, under investigation in Sec. \ref{sec:SourceDistortionCost}.

\begin{example}
We consider the utility function $\Phi (u,x,y,v) = \UN(v = u)$ that achieves its maximum when the symbols of source and decoder's output $U=V$ are equal. From Shannon's separation result (see \cite{shannon-bell-1948} and Corollary \ref{coro:Shannon}) the maximal expected utility $\E\big[ \Phi (u,x,y,v) \big] = \E\big[ \UN(v = u)  \big] = 1$ is achievable if and only if $\max_{\PP(x)}I( X;Y )  - H(U )\geq 0$. This information constraint corresponds to the target distribution is $\QQ(u,x,y,v)  = \PP_{\sf{u}}(u)   \times \UN(v= u) \times \PP^{\star}(x) \times\mc{T}(y | x )$ where $\PP^{\star}(x) $ achieves the maximum in $\max_{\PP(x)}I( X;Y ) $. The problem of empirical coordination generalizes the problem of information transmission of Shannon.
\end{example}

We define the set $\mc{A}$ of achievable target distributions $\QQ(x,v | u)$, characterized by  Theorem \ref{theo:1CwithChannel}:
\begin{align}
\mc{A} = \bigg\{  \QQ(x,v | u) ,\text{ s.t. }, \qquad\qquad\qquad\qquad\qquad \nonumber  \\ 
\max_{\QQ(w|u,v,x),\atop  |\mc{W}| \leq  |\mc{U} \times   \mc{X} \times \mc{V}| +1} \bigg( I( W;Y  |V )  -   I(  U ; V  ,W  )  \bigg) \geq 0\bigg\}.\label{eq:AchievableSet}
\end{align}
The set $\mc{A}$ is closed since the information constraint \eqref{eq:1CwithChannel1} is not strict. The set of symbols $|\mc{U}\times \mc{X}\times \mc{V} |< + \infty$ are discrete, $\mc{A}$ is a closed and bounded subset of $[0,1]^{|\mc{U}\times \mc{X}\times \mc{V} |}$, hence $\mc{A}$ is a compact set. The set of achievable utilities $ \textsf{U} $ is the image, by the expectation operator, of the set of achievable distributions $\mc{A} $:
\begin{eqnarray}
\mc{A}&\longrightarrow&  \textsf{U},\nonumber \\
\QQ(x,v|u)  &\longrightarrow& \E_{\QQ(u,x,y,v) }\bigg[ \Phi(U,X,Y,V) \bigg].\label{eq:ImageFunction}
\end{eqnarray}
The set $\textsf{U}$ is a closed and bounded subset of $\R$, hence it is also a compact set. Since the probability distributions of the source $\PP_{\sf{u}}(u)$ and of the channel $ \mc{T}(y | x )$ are fixed,  the conditional probability distribution $\QQ(x,v | u)\in\mc{A}$ is the unique degree of freedom for the optimization of the expected utility $\E_{\QQ} \big[ \Phi(U ,X ,Y ,V )\big]$. 
\begin{theorem}\label{theo:ConvexProblem}
The set ${\mc{A}}$ is convex and the optimization problem stated in equation \eqref{eq:optimizationProblem} is a convex optimization problem:
\begin{align}
\max_{\QQ(x,v | u) \in{\mc{A}}}\E_{\QQ} \bigg[ \Phi(U ,X ,Y ,V )\bigg].\label{eq:optimizationProblem}
\end{align}
The information constraint \eqref{eq:1CwithChannel1} is concave with respect to the conditional probability distribution $\QQ(x,v | u)$.
\end{theorem}
The proof of Theorem \ref{theo:ConvexProblem} is stated in App. \ref{sec:ProofTheoConvex}. Since the expectation $\E_{\QQ} \big[ \Phi(U ,X ,Y ,V )\big]$ is linear in $\QQ(x,v | u)$ and the set $\mc{A}$ is convex, the optimal solution of the optimization problem \eqref{eq:optimizationProblem} lies on the boundary of the set $\mc{A}$. Denote by $\texttt{bd}(\mc{A}) $  the subset of the boundary of $\mc{A}$ where the information constraint is zero:
\begin{align}
\texttt{bd}(\mc{A}) = \bigg\{  \QQ(x,v | u) ,\text{ s.t. }\qquad\qquad\qquad\qquad\qquad \nonumber  \\ 
\max_{\QQ(w|u,v,x),\atop  |\mc{W}| \leq  |\mc{U} \times   \mc{X} \times \mc{V}| +1} \bigg( I( W;Y  |V )  -   I(  U ; V  ,W  )  \bigg) = 0\bigg\}.\label{eq:Boundary}
\end{align}
In the following sections, we provide numerical results for the coordination game of \cite{GossnerHernandezNeyman06} and for the trade-off between source distortion and channel cost.

%%%%%%%%%%%%%%%%%%%%%%%%%%%%%%%%%%%%%%%%%%%%%%%%%%%%%%
%%%%%%%%%%%%%%%%%%%%%%%%%%%%%%%%%%%%%%%%%%%%%%%%%%%%%%
%%%%%%%%%%%%%%%%%%%%%%%%%%%%%%%%%%%%%%%%%%%%%%%%%%%%%%%
%%%%%%%%%%%%%%%%%%%%%%%%%%%%%%%%%%%%%%%%%%%%%%%%%%%%%%%

%%%%%%%%%%%%%%%%%%%%%%%%%%%%%%%%%%%%%%%%%%%%%%%%%%%%%%
%%%%%%%%%%%%%%%%%%%%%%%%%%%%%%%%%%%%%%%%%%%%%%%%%%%%%%
%%%%%%%%%%%%%%%%%%%%%%%%%%%%%%%%%%%%%%%%%%%%%%%%%%%%%%
%%%%%%%%%%%%%%%%%%%%%%%%%%%%%%%%%%%%%%%%%%%%%%%%%%%%%%

\subsection{Coordination game}\label{sec:ExampleMatchingPennies}

We consider a binary information source and a binary symmetric channel with the set of two symbols $\mc{U} = \mc{X} = \mc{Y} = \mc{V} = \{0,1\}$, as represented by Fig. \ref{fig:JointSourceChannelProblem2x2b}. The information source depends on the parameter $p\in [0,1]$ and the channel depends on the parameter $\varepsilon \in [0,1]$.  
  \begin{figure}[h!]
\begin{center}
\psset{xunit=0.6cm,yunit=0.6cm}
\begin{pspicture}(-1,-1)(15,3.5)
\rput(0,1.5){$p$}
\rput(0,0.5){$1-p$}
\rput(1,2.5){$U$}
\rput(1,1.5){$0$}
\rput(1,0.5){$1$}
\psframe(2,0)(3,2)
\rput(2.5,1){$\C$}
\rput(4,2.5){$X$}
\rput(4,1.5){$0$}
\rput(4,0.5){$1$}
\psline{->}(5,1.5)(9,1.5)
\psline{->}(5,0.5)(9,0.5)
\psline[linewidth=0.5pt, linestyle = dashed]{->}(5,1.5)(9,0.5)
\psline[linewidth=0.5pt, linestyle = dashed]{->}(5,0.5)(9,1.5)
\rput(7,1.8){$1 - \varepsilon$}
\rput(7,0.1){$1 - \varepsilon$}
\rput(8.5,1){$\varepsilon$}
\rput(10,2.5){$Y$}
\rput(10,1.5){$0$}
\rput(10,0.5){$1$}
\psframe(11,0)(12,2)
\rput(11.5,1){$\D$}
\rput(13,2.5){$V$}
\rput(13,1.5){$0$}
\rput(13,0.5){$1$}
\end{pspicture}
\caption{Binary information source and binary symmetric channel with parameters $p\in [0,1]$ and $\varepsilon \in [0,1]$}\label{fig:JointSourceChannelProblem2x2b}
\end{center}
\end{figure}
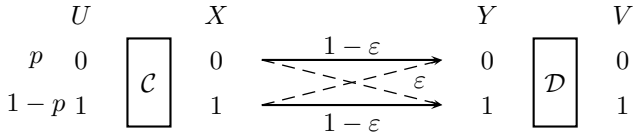
The goal of both encoder and decoder is to coordinate their actions $X$ and $V$ with the information source $U$, in order to maximize the utility function defined by Fig. \ref{fig:MatchingPennies}. In fact, the maximal utility can be achieved when the encoder and the decoder implement the same symbol as the source symbol $X =V = U$ \textit{i.e.}, the sequences of symbols $U^n$, $X^n$ and $V^n$ are jointly typical for the probability distribution $\PP_{\sf{u}}(u) \times \UN(\sf{x} =\sf{v} = \sf{u})$. In \cite{GossnerHernandezNeyman06}, the authors proved that this distribution is not achievable.
 \begin{figure}[h!]
 \begin{small}
\begin{center}
\psset{xunit=0.6cm,yunit=0.6cm}
\begin{pspicture}(-19,-1)(0,3.5)
\psframe(-10,0)(-6,3)
\psline(-10,1.5)(-6,1.5)
\psline(-8,0)(-8,3)
\rput(-9,2.25){$0$}
\rput(-7,2.25){$0$}
\rput(-9,0.75){$0$}
\rput(-7,0.75){$1$}
%\rput(-6.4,0.5){${\star}$}
%\rput(-10.7,1.5){$X=0$}
\rput(-10.95,2.25){$X=0$}
\rput(-10.95,0.75){$X=1$}
\rput(-17,2.25){$X=0$}
\rput(-17,0.75){$X=1$}
\rput(-9,3.3){$V=0$}
\rput(-7,3.3){$V=1$}
\rput(-15,3.3){$V=0$}
\rput(-13,3.3){$V=1$}
\rput(-14,-0.8){$U=0$}
\rput(-8,-0.8){$U=1$}
\psframe(-16,0)(-12,3)
\psline(-16,1.5)(-12,1.5)
\psline(-14,0)(-14,3)
\rput(-15,2.25){$1$}
\rput(-13,2.25){$0$}
\rput(-15,0.75){$0$}
\rput(-13,0.75){$0$}
 \end{pspicture}
 \caption{Utility function $\Phi:\mc{U} \times \mc{X} \times \mc{V} \mapsto \R $ corresponding to the coordination game of \cite{GossnerHernandezNeyman06}. } \label{fig:MatchingPennies}
\end{center}
\end{small}
\end{figure}
The objective is to determine the probability distribution $\QQ(x,v|u)$ that is achievable and that maximizes the expected utility function represented by Fig. \ref{fig:MatchingPennies}. 
 \begin{figure}[ht]
\begin{center}
\begin{small}
\psset{xunit=0.6cm,yunit=0.6cm}
\begin{pspicture}(-19,-1)(5,3.5)
\psframe(-10,0)(-6,3)
\psline(-10,1.5)(-6,1.5)
\psline(-8,0)(-8,3)
\rput(-9,2.25){$\frac{1 - \gamma}{6}$}
\rput(-7,2.25){$\frac{1 - \gamma}{6}$}
\rput(-9,0.75){$\frac{1 - \gamma}{6}$}
\rput(-7,0.75){$\frac{\gamma}{2}$}
%\rput(-6.4,0.5){${\star}$}
%\rput(-10.7,1.5){$X=0$}
\rput(-10.9,2.25){$X=0$}
\rput(-10.9,0.75){$X=1$}
\rput(-17,2.25){$X=0$}
\rput(-17,0.75){$X=1$}
\rput(-9,3.3){$V=0$}
\rput(-7,3.3){$V=1$}
\rput(-15,3.3){$V=0$}
\rput(-13,3.3){$V=1$}
\rput(-14,-0.8){$U=0$}
\rput(-8,-0.8){$U=1$}
\psframe(-16,0)(-12,3)
\psline(-16,1.5)(-12,1.5)
\psline(-14,0)(-14,3)
\rput(-15,2.25){$\frac{\gamma}{2}$}
\rput(-13,2.25){$\frac{1 - \gamma}{6}$}
\rput(-15,0.75){$\frac{1 - \gamma}{6}$}
\rput(-13,0.75){$\frac{1 - \gamma}{6}$}
 \end{pspicture}
 \end{small}
 \caption{The optimal probability distribution for the  utility function stated in Fig. \ref{fig:MatchingPennies}, depends on the parameter $\gamma\in [0,1]$. The value of the expected utility is equal to this parameter $\gamma\in [0,1]$.}\label{fig:OptimalProbaMatchingPennies}
\end{center}
\end{figure}
We suppose that the source parameter is $p=\frac{1}{2}$. As mentioned in \cite{GossnerHernandezNeyman06}, the utility function presented in Fig. \ref{fig:MatchingPennies} is symmetric, hence the empirical distribution that maximizes the utility function is given by Fig. \ref{fig:OptimalProbaMatchingPennies}, with parameter $\gamma\in [0,1]$.We consider lower and upper bounds on the information constraint  \eqref{eq:1CwithChannel1} that do not involve an auxiliary random variable $W$. 

$\bullet$ The lower bound is obtained by letting $W =X$ in the information constraint \eqref{eq:1CwithChannel1} of Theorem \ref{theo:1CwithChannel}:
%\begin{footnotesize}
\begin{align}
& \max_{{\QQ}\in \Q} \bigg( I( W;Y  |V )  -   I(  U ; V  ,W  )  \bigg) \nonumber \\
&\geq I( X;Y  |V )  -   I(  U ; V  ,X  )  \geq0. \label{eq:ICMatchingPenniesLow}
\end{align}
%\end{footnotesize}

$\bullet$ The upper bound comes from the result with source feedforward stated in  \cite{LarrousseLasaulceBloch(IT)14}, in which the decoder observes the pair $(Y,U)$:
%\begin{footnotesize}
\begin{align}
& I( X;Y  | U ,V )  -   I(  U ; V   ) \nonumber \\
&\geq  \max_{{\QQ}\in \Q} \bigg( I( W;Y  |V )  -   I(  U ; V  ,W  )  \bigg)  \geq0.\label{eq:ICMatchingPenniesUp}
\end{align}
%\end{footnotesize}

The difference between the upper bound \eqref{eq:ICMatchingPenniesUp}  and the lower bound \eqref{eq:ICMatchingPenniesLow} is equal to $I(U;X|V,Y)$. 

\begin{proposition}\label{prop:InformationConstraintMP}
The lower \eqref{eq:ICMatchingPenniesLow} and upper \eqref{eq:ICMatchingPenniesUp} bounds on the information constraint \eqref{eq:1CwithChannel1} of Theorem \ref{theo:1CwithChannel}, depending on parameter $\gamma \in [0,1]$, are given by:
\begin{align}
&I( X;Y  |V )  -   I(  U ; V  ,X  )  \nonumber\\
&=H_b(\gamma) + (1-\gamma) \cdot \log_2(3) - 1 - H_b\bigg(\frac{2}{3} -  \frac{2\gamma}{3} \bigg)  \nonumber\\
&- H_b(\varepsilon) + H_b\bigg(\frac{2}{3} - \frac{2\gamma}{3}  +  \varepsilon \cdot \frac{4\gamma - 1}{3}  \bigg),\label{eq:LowerBound}
\end{align}
%A lower bound on \eqref{eq:1CwithChannel1} is :
\begin{align}
&I( X;Y  |  U,V )  -   I(  U ; V   ) \nonumber \\
&= H_b(\gamma) + (1-\gamma) \cdot \log_2(3) - 1 -  H_b(\varepsilon) +  \frac{2\gamma + 1}{3} \nonumber \\
&\times \Bigg( H_b\bigg( (1 - \varepsilon ) \cdot \frac{3\gamma}{2\gamma + 1}  +  \varepsilon \cdot  \frac{1 - \gamma}{2\gamma + 1}  \bigg) -  H_b\bigg( \frac{3\gamma}{2\gamma + 1} \bigg) \Bigg) . \label{eq:UpperBound}
\end{align}
\end{proposition}
%The proof of Proposition \ref{prop:InformationConstraintMP} is stated in App. \ref{sec:ProofPropMatchingPennies} and App. \ref{sec:ProofPropMatchingPennies2}.
The proof of Proposition \ref{prop:InformationConstraintMP} comes from the definition of the entropy.

\begin{remark}
Note that when $\varepsilon=0$, the information constraints \eqref{eq:LowerBound} and \eqref{eq:UpperBound} reduce to the one stated in \cite{GossnerHernandezNeyman06}  and  \cite{Cuff(ImplicitCoordination)11}:
\begin{eqnarray}
H( X   |V )  -   I(  U ; V  ,X  ) = H_b(\gamma) + (1-\gamma) \cdot \log_2(3) - 1 .
\end{eqnarray}
\end{remark}

 \begin{figure}[ht!]
\centering
\includegraphics[width=0.5\textwidth]{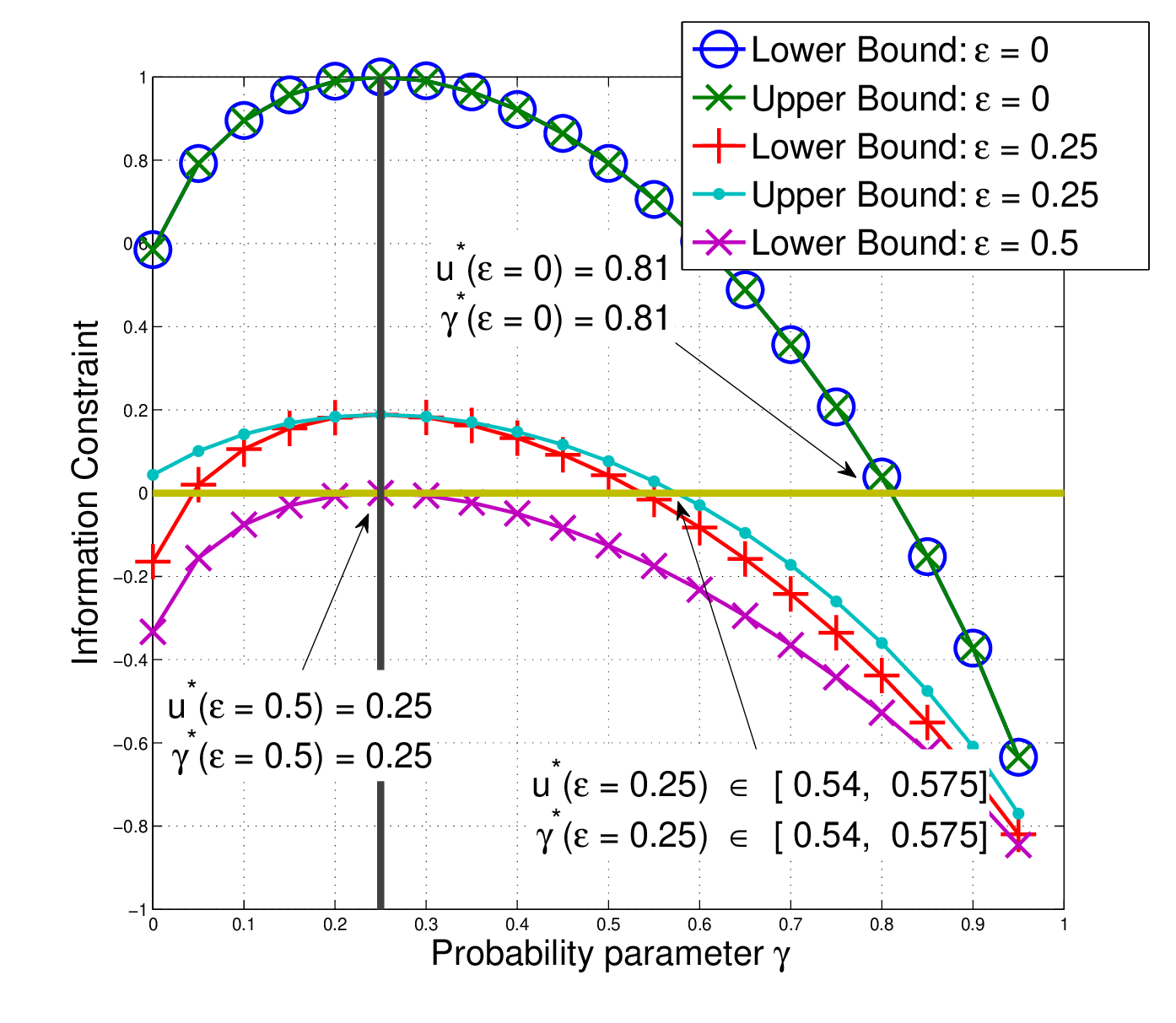}
\caption{Information constraints depending on the probability parameter $\gamma\in [0,1]$, for different values of the channel noise parameter $\varepsilon \in \{0,0.25,0.5\}$.}\label{fig:InformationConstraintPlot2014_03_04}
 \end{figure}

  \begin{figure}[ht!]
\centering
\includegraphics[width=0.5\textwidth]{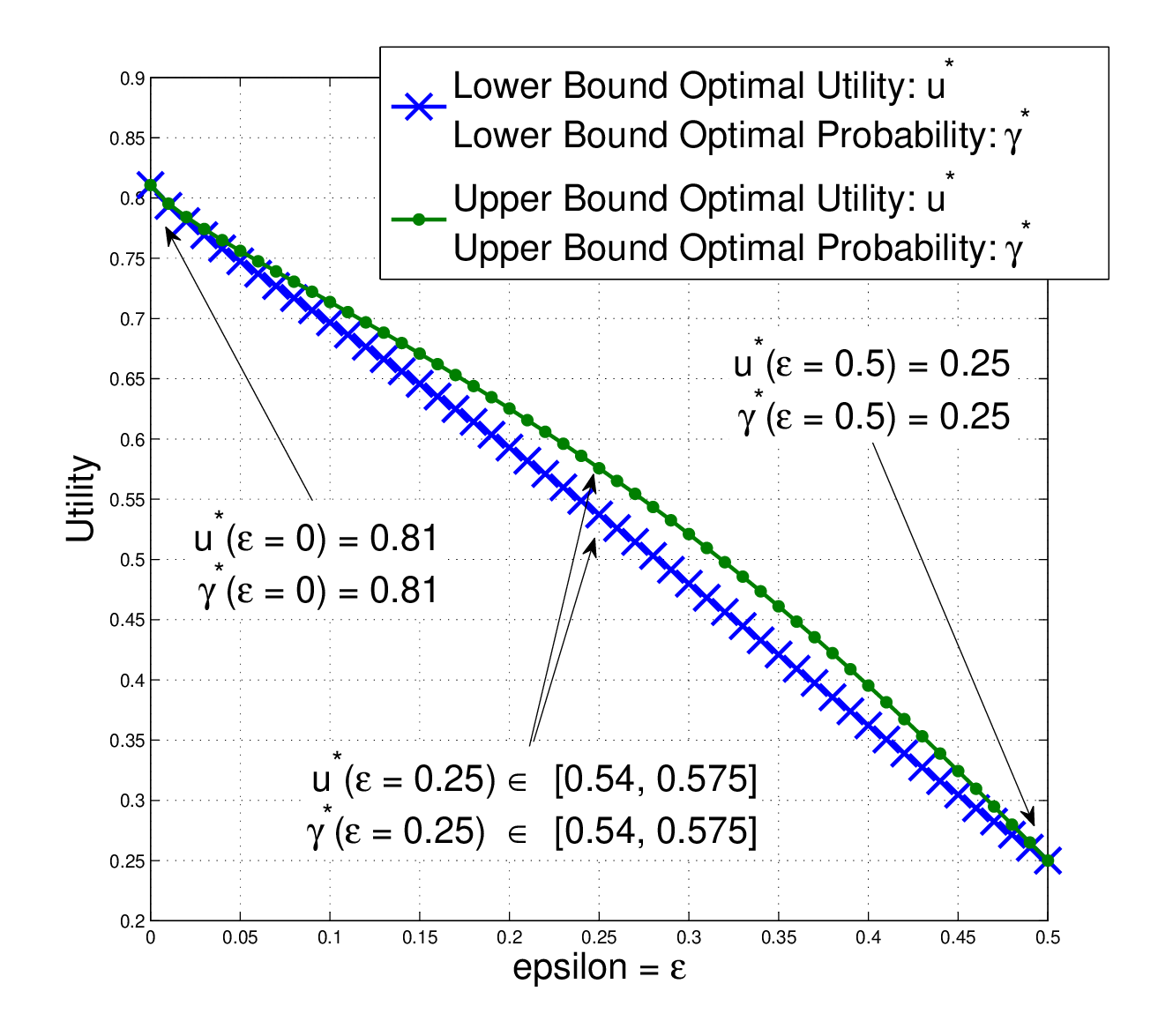}
\caption{Optimal utility/probability parameter $\gamma\in [0,1]$ depending on the channel parameter $\varepsilon \in [0,0.5]$.}
\label{fig:InformationConstraintUtility2014_03_04}
 \end{figure} 

Fig. \ref{fig:InformationConstraintPlot2014_03_04} and \ref{fig:InformationConstraintUtility2014_03_04} represent  the lower  and the upper bounds of equations \eqref{eq:LowerBound} and \eqref{eq:UpperBound}, depending on $\gamma\in [0,1]$, for different values of $\varepsilon$. As established by Theorem \ref{theo:ConvexProblem}, these information constraints are concave with respect to the  probability parameter $\gamma\in [0,1]$. The maximum of the information constraint is achieved by parameter $\gamma = 0.25$, that corresponds to the uniform probability distribution over the symbols $\mc{U}\times \mc{X}\times\mc{V}$. The maximum of the utility is achieved by parameter $\gamma^{\star}\in [0,1]$, that corresponds to a zero of the information constraint.

\begin{itemize}
\item[$\bullet$] If the channel is perfect \textit{i.e.}, $\varepsilon =0$, the optimal solution corresponds to the one stated in  \cite{GossnerHernandezNeyman06}  and \cite{LarrousseLasaulceBloch(IT)14}. The optimal utility and the optimal probability distribution $\QQ^{\star}$ are given by the parameter $\gamma^{\star} \simeq 0.81$, that is solution of the equation $H_b(\gamma) + (1-\gamma) \cdot \log_2(3)  =1$.
\item[$\bullet$] If the channel parameter is $\varepsilon =0.5$, then the channel outputs are statistically independent of the channel inputs and the optimal utility $0.25$ corresponds to the situation where the random variables $U$,  $X$ and $V$ are uniform and mutually independent. In that case, no information is transmitted.
\item[$\bullet$] If the channel parameter is $\varepsilon = 0.25$, the optimal utility belongs to the interval $ \max_{\QQ(x,v|u)\in \mc{A} }\E_{\QQ} \big[ \Phi(U ,X ,Y ,V )\big] \in [0.54,0.575]$. Even if the channel is noisy, the symbols of the encoder, the decoder and the source are perfectly coordinated more than half of the time. 
\end{itemize}

%%%%%%%%%%%%%%%%%%%%%%%%%%%%%%%%%%%%%%%%%%%%%%%%%%%%%%%%%%%%%%%%%%%%%%%%%%%%%%%%%%%%%%%%%%%%%%%%%%%%%%%%%%%%%%%%%%%%%%%%%%%%%%%%%%%%%%%%%%%%%%%%%%%%%%%%%%%%%%%%%%%%%%%%%%%%%%%%%%%%%%%%%%%%%%%%%%%%%%%%%%%%%%%%%%%%%%%%%%%%%%%%%%%%%%%%%%%%%%%%%%%%%%%%%%%%%%%%%%%%%%%%%%%%%%%%%%%%%%%%%%%%%%%%%%%%%%%%%%%%%%%%%%%%%%%%%%%%%%%%%%%%

\subsection{Source distortion and channel cost}\label{sec:SourceDistortionCost}

%\subsection{Characterization of the pairs of distortion and cost}\label{sec:CharacterizationDistortionCost}

We investigate the relationship between the result stated in Theorem \ref{theo:1CwithChannel} for empirical coordination and the classical results of rate-distortion and channel capacity, stated in \cite[pp. 47, 57 and 66]{ElGammalKim(book)11}. We characterize the achievable pairs of distortion-cost $(\textsf{D}^{\star},\textsf{C}^{\star})$ by using a distribution $\QQ(x,v | u)$ that satisfies $ \E_{\QQ}\big[d(U,V)\big] =\textsf{D}^{\star} $ and  $\E_{\QQ}\big[c(X)\big] = \textsf{C}^{\star}$ and that is achievable. We show that the joint source-channel coding result of Shannon \cite{shannon-bell-1948} is a particular case of Theorem \ref{theo:1CwithChannel}.

\begin{definition}
A distortion function evaluates the distortion between the source symbol $u \in \mc{U}$ and the output of the decoder $v \in \mc{V}$:
\begin{eqnarray}
d : \mc{U} \times  \mc{V} \mapsto \R.
\end{eqnarray}
A channel cost function evaluates the cost of the input symbol $x \in \mc{X}$ of the channel:
\begin{eqnarray}
c : \mc{X} \mapsto \R.
\end{eqnarray}
\end{definition}

\begin{definition}\label{def:DC}
The pair of distortion-cost $(\textsf{D}^{\star},\textsf{C}^{\star})$ is achievable if for all $ \varepsilon >0$,  there exists $\bar{n}$ such that for all $n\geq\bar{n}$, there exists a code $c\in \mc{C}(n)$ such that:
\begin{eqnarray}
\Bigg|  \E\bigg[ \frac{1}{n} \sum_{i=1}^n d(U_i ,V_i ) \bigg]  -  \textsf{D}^{\star} \Bigg| &\leq& \varepsilon ,\\
\Bigg|  \E\bigg[ \frac{1}{n} \sum_{i=1}^n c(X_i ) \bigg] -  \textsf{C}^{\star} \Bigg| &\leq&  \varepsilon.
\end{eqnarray}
\end{definition}

In this section, we consider exact distortion and cost rather than upper bounds on the distortion  $ \E\big[ \frac{1}{n} \sum_{i=1}^n d(U_i ,V_i ) \big]  \leq   \textsf{D}^{\star}  + \varepsilon$, and cost $ \E\big[ \frac{1}{n} \sum_{i=1}^n c(X_i ) \big]  \leq   \textsf{C}^{\star}  + \varepsilon$, as in  \cite[pp. 47 and 57]{ElGammalKim(book)11}. Although the conditions of exact distortion and cost, are more restrictive than upper bounds, the solution is a direct consequence of Theorem \ref{theo:1CwithChannel}.

\begin{corollary}[Trade-off distortion cost]\label{theo:DistortionCost}
The two following assertions are equivalent:
\begin{itemize}
\item[$\bullet$] The pair of distortion-cost $(\textsf{D}^{\star},\textsf{C}^{\star})$ satisfies:
\begin{eqnarray}
\max_{\PP(x), \atop  \E_{\QQ}[c(X)]  = \textsf{C}^{\star}} I( X;Y )  -  \min_{\QQ(v|u), \atop  \E_{\QQ}[d(U,V)]  = \textsf{D}^{\star}} I(  U ; V   )\geq 0.\label{eq:CostDistortion}
\end{eqnarray}
\item[$\bullet$] There exists an achievable probability distribution $\QQ(x,v|u) $ \textit{i.e.}, that satisfy the information constraint \eqref{eq:1CwithChannel1}, such that:
\begin{eqnarray}
\E_{\QQ}\Big[d(U,V)\Big] &=& \textsf{D}^{\star},\\
\E_{\QQ}\Big[c(X)\Big] &=& \textsf{C}^{\star}.
\end{eqnarray}
\end{itemize}
\end{corollary}

\begin{remark}
Unlike \cite[pp. 43, Remark 3.5]{ElGammalKim(book)11}, Corollary \ref{theo:DistortionCost} establishes that the pair of distortion-cost $(\textsf{D}^{\star},\textsf{C}^{\star})$ is achievable when equation \eqref{eq:CostDistortion} is also equal to zero. More details are provided in App. \ref{sec:ProofEqualityIC}.
%Corollaries \ref{coro:CoordTrans} and \ref{coro:Shannon}  establish that the triple of rate-distortion-cost $(\textsf{R}, \textsf{D}^{\star},\textsf{C}^{\star})$ is also achievable when $I(X;Y)=I(U;V)$. More details are provided in App. \ref{sec:ProofEqualityIC}.
\end{remark}

\begin{proof}[Corollary \ref{theo:DistortionCost}]
The equivalence stated in Corollary \ref{theo:DistortionCost} is reformulated using the equations  \eqref{eq:ProofTheoDC6b} and \eqref{eq:ProofTheoDC7b}.
\begin{align}
&\exists \; \QQ(x)\times \QQ(v|u) \; \text{ s.t. } \E_{\QQ}\big[c(X)\big]= \textsf{C}^{\star},\nonumber \\
& \E_{\QQ}\big[d(U,V)\big]= \textsf{D}^{\star},  \text{ and } I(X;Y) - I(U;V) \geq 0, \label{eq:ProofTheoDC6b}  \\
\Longleftrightarrow& \exists \;\QQ(x,v|u) \; \text{ s.t. } \E_{\QQ}\big[c(X)\big]= \textsf{C}^{\star},  \E_{\QQ}\big[d(U,V)\big]= \textsf{D}^{\star}, \nonumber\\
 &\text{ and }  \max_{\QQ(w|u,v,x),\atop  |\mc{W}| \leq  |\mc{U} \times   \mc{X} \times \mc{V}| +1 } \bigg( I( W;Y  |V )  -   I(  U ; V  ,W  )  \bigg) \geq 0 .\label{eq:ProofTheoDC7b}
\end{align}
If distribution $\QQ(x,v|u)$ is achievable, then the product of marginal distributions $ \QQ(x)\times \QQ(v|u)$ is also achievable, see Remark \ref{remark:CoordRestrictive}. Corollary \ref{coro:Shannon} proves that equations \eqref{eq:ProofTheoDC6b} and \eqref{eq:ProofTheoDC7b} are equivalent and this concludes the proof of Corollary \ref{theo:DistortionCost}.
\end{proof}

\begin{example}[Trade-off distortion-cost]\label{ex:DistortionCost}
We consider the communication problem with parameters $p\in [0,1]$ and $\varepsilon \in [0,1]$ represented by Fig. \ref{fig:JointSourceChannelProblem2x2b}. The distortion function is $d(u,v) = \UN(u \neq v)$ and cost function is $c(x) = \UN(x=0)$. The distribution of $X$ is binary $\QQ(X =0 ) = \alpha$, $\QQ(X =1 ) = 1 - \alpha$, with $\alpha \in [0,1]$ and the conditional probability distribution $\QQ(v|u)$ is binary symmetric $\QQ(V=0|U=0)= 1- \beta$, $\QQ(V=0|U=1)= \beta$, with $\beta \in [0,1]$. The expected distortion and the expected cost are given by:
\begin{eqnarray}
\E_{\QQ}\Big[d(U,V)\Big]  &=& p \cdot \beta  + (1 - p) \cdot \beta = \beta ,\\
\E_{\QQ}\Big[c(X)\Big] &=& \sum_{x} \QQ(x)\cdot \UN(x=0) = \alpha.
\end{eqnarray}
The maximum and the minimum of information constraint \eqref{eq:CostDistortion} disappear since the distributions $\QQ(v|u)$ and $\QQ(x)$ that achieve the target distortion-cost $(\textsf{D}^{\star},\textsf{C}^{\star}) = (\beta,\alpha)$ are unique. The information constraint \eqref{eq:CostDistortion} of Corollary \ref{theo:DistortionCost} is equal to:
\begin{eqnarray}
&&I( X;Y  )  -  I(  U;V  ) \\
&=& H_b\Big( \alpha\cdot \varepsilon + (1- \alpha) \cdot (1- \varepsilon) \Big)+ H_b\Big(\beta\Big)  \\
&-& H_b\Big(\varepsilon\Big) - H_b\Big( \beta\cdot p + (1- \beta) \cdot (1- p) \Big) .
\end{eqnarray}
\begin{figure}[ht!]
\centering
\includegraphics[width=0.4\textwidth]{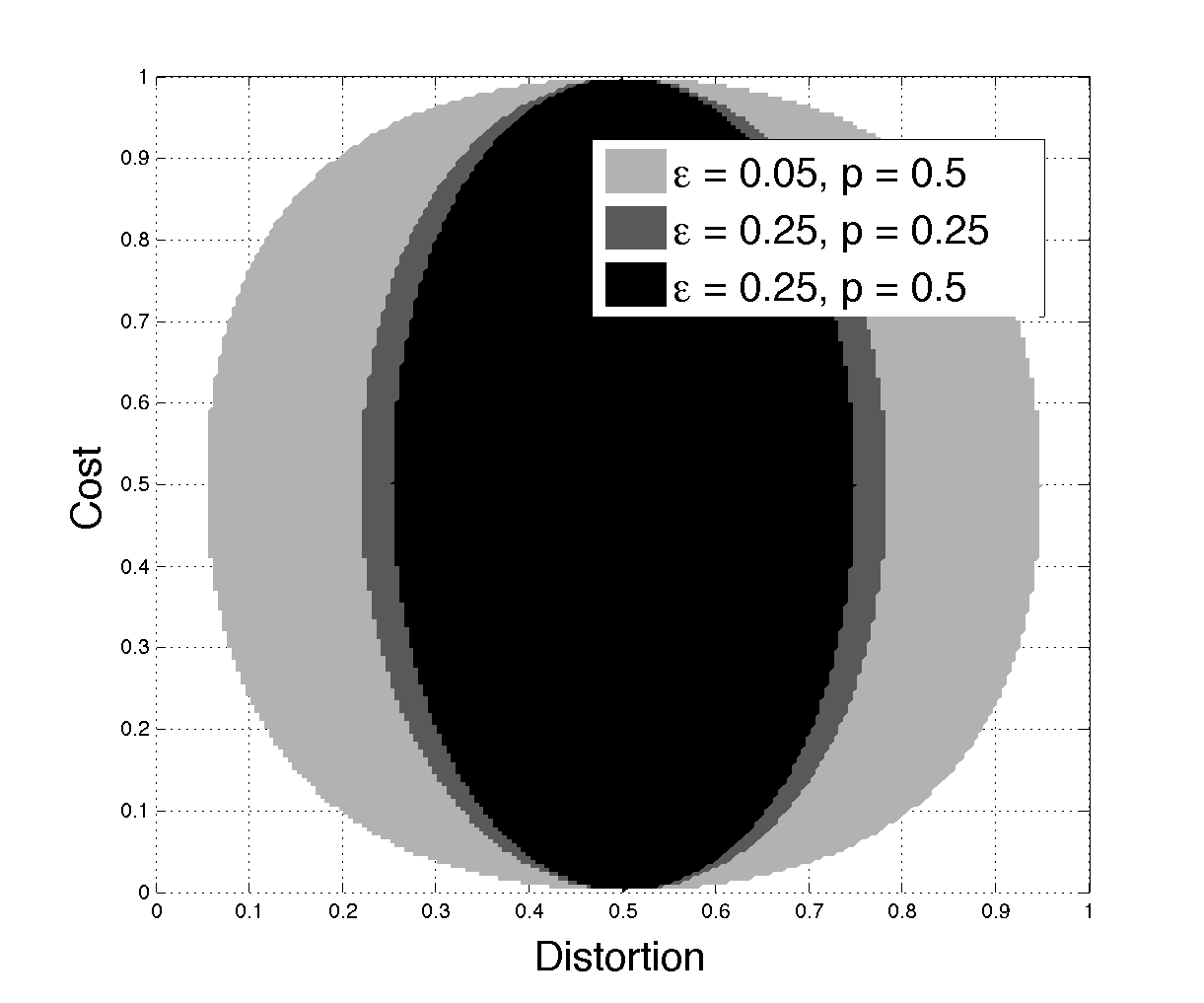}
\caption{The set of achievable Distortion-Cost $(\textsf{D}^{\star},\textsf{C}^{\star})$ depending on source and channel parameters $( \varepsilon,p) \in \{(0.05,0.5) ;  (0.25,0.25); (0.25,0.5)\}$.}
\label{fig:TradeOffDistortionCost2015_05_11}
 \end{figure} 
Fig. \ref{fig:TradeOffDistortionCost2015_05_11} represents three regions of achievable pairs of exact distortion-cost, for parameters $( \varepsilon,p) \in \{(0.05,0.5) ;  (0.25,0.25); (0.25,0.5)\}$. This illustrates the trade-off between minimal source distortion and minimal channel cost. The boundary of the three dark regions corresponds to the  pairs of distortion and cost $(\textsf{D}^{\star},\textsf{C}^{\star})$, that satisfy the equality $I( X;Y  ) = I(  U;V  )$, in equation \eqref{eq:CostDistortion}. 
\end{example}

%%%%%%%%%%%%%%%%%%%%%%%%%%%%%%%%%%%%%%%%%%%%%%%%%%%%%%%%%%%%%%%%%%%%%%%%%%%%%%%%%%%%%%%%%%%%%%%%%%%%%%%%%%%%%%%%%%%%%%%%%%%%%%%%%%%%%%%%%%%%%%%%%%%%%%%%%%%%%%%%%%%%%%%%%%%%%%%%%%%%%%%%%%%%%%%%%%%%%%%%%%%%%%%%%%%%%%%%%%%%%%%%%%%%%%%%%%%%%%%%%%%%%%%%%%%%%%%%%%%%%%%%%%%%%%%%%%%%%%%%

\section{Causal decoding}\label{sec:CwithChannel}

%\subsection{Characterization for causal decoding}\label{sec:CD}

In this section, we consider the problem of causal decoding instead of strictly-causal decoding.  At instant $i\in\{1,\ldots,n\}$, the decoder $\D$ observes the sequence of past and \textit{current} channel outputs $y^{i} = (y_1,\ldots,y_{i}) \in \mc{Y}^{i}$  and returns a symbol $v_i \in \mc{V}_i$. The main difference between causal and strictly causal decoding is that at each instant $i\in\{1,\ldots,n\}$, the symbol $V_i$ may also be correlated with the channel output $Y_i$.

\begin{definition}\label{def:CodeCD}
A  code $c\in\mc{C}(n)$ with causal decoder is a tuple of functions $c=(f,\{g_i\}_{i=1}^n)$ defined by:
\begin{eqnarray}
f &:& \mc{U}^n  \longrightarrow \mc{X}^n ,\label{eq:CausalCodeSource1}\\
g_i &:& \mc{Y}^{i}  \longrightarrow \mc{V},\qquad i \in \{ 1,\ldots,n\}.\label{eq:CausalCodeSource2}
\end{eqnarray}
\end{definition}
Similarly as in Definition \ref{def:AchievableDistribution}, a joint probability distribution $\QQ(u,x,y,v)$ is achievable with non-causal encoding and causal decoding, if there exists a sequence of causal code with  small error probability. 

\begin{figure}[!ht]
\begin{center}
\psset{xunit=0.9cm,yunit=0.9cm}
\begin{pspicture}(0,-0.35)(8.5,1.7)
\pscircle(0,0.5){0.45}
\psframe(2,0)(3,1)
\pscircle(5,0.5){0.45}
\psframe(7,0)(8,1)
\psline[linewidth=1pt]{->}(0.5,0.5)(2,0.5)
\psline[linewidth=1pt]{->}(3,0.5)(4.5,0.5)
\psline[linewidth=1pt]{->}(5.5,0.5)(7,0.5)
\psline[linewidth=1pt]{->}(8,0.5)(9,0.5)
%\psline[linewidth=1pt]{->}(0,0)(0,-0.5)(5,-0.5)(5,0)
%\psline[linewidth=1pt]{->}(2.5,-0.5)(2.5,0)
%\psline[linewidth=1pt]{->}(0,1)(0,1.5)(7.5,1.5)(7.5,1)
\rput[u](1,0.8){$U^n$}
\rput[u](3.75,0.8){$X^n$}
\rput[u](6.25,0.8){$Y^{i}$}
\rput[u](8.5,0.8){$V_i$}
\rput(0,0.5){$\PP_{\sf{u}}$}
\rput(2.5,0.5){$\C$}
\rput(5,0.5){$\mc{T}_{\sf{y|x}}$}
\rput(7.5,0.5){$\D$}
\end{pspicture}
\caption{The decoder is causal $V_i = g_i(Y^{i})$, for all $i\in\{1, \ldots ,n\}$ and the encoder is non-causal $X^n = f(U^n)$. Theorem \ref{theo:CwithChannel} characterizes the set of probability distributions $\QQ(u,x,y,v)$ that are achievable, using two auxiliary random variables $W_1$  and $W_2$.}\label{fig:CwithChannel}
\end{center}
\end{figure}
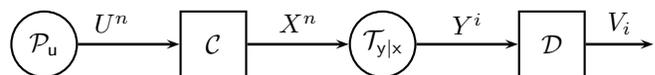

\begin{theorem}[Causal decoding] \label{theo:CwithChannel}
$\qquad$\\
A target joint probability distribution $\QQ(u,x,y,v) \in  \Delta(\mc{U} \times \mc{X} \times \mc{Y} \times \mc{V}  )$ is achievable if and only if the two following conditions are satisfied:\\
1) It decomposes as follows:
\begin{eqnarray}
\QQ(u,x,y,v) = \PP_{\sf{u}}(u)   \times \QQ(x | u) \times  \mc{T}(y | x )\times \QQ(v | u,x,y),\label{eq:CwithChannel0}
\end{eqnarray}
2) There exists two auxiliary random variables $W_1 \in \mc{W}_1$ and  $W_2 \in \mc{W}_2$ such that: 
\begin{eqnarray}
\max_{{\QQ}\in \Q_{\sf{c}}} \bigg( I( W_1;Y  |W_2 )  -  I( W_1, W_2 ; U   ) \bigg) \geq 0,
\label{eq:CwithChannel1}
\end{eqnarray}
where $ \Q_{\sf{c}}$ is the set of  joint probability distributions ${\QQ}(u,x,w_1,w_2,y,v)\in  \Delta(\mc{U} \times \mc{X}  \times \mc{W}_1 \times \mc{W}_2 \times \mc{Y} \times \mc{V}  )$ that decompose as follows:
\begin{eqnarray}
{\QQ}(u,x,w_1,w_2,y,v) = \PP_{\sf{u}}(u)   \times \QQ(x,w_1,w_2 | u) \nonumber \\
 \times  \mc{T}(y | x ) \times \QQ(v| y,w_2) ,\label{eq:CwithChannel5}
\end{eqnarray}
and for which the supports of $W_1$ and $W_2$ are bounded by $\max\big( |\mc{W}_1|,  |\mc{W}_2| \big)\leq  |\mc{U} \times \mc{X}  \times \mc{Y}\times \mc{V}  | +2 $. We denote by $\mc{A}_{\textsf{d}}$ the set of  probability distributions $\QQ(u,x,y,v)$ that are achievable with causal decoding.
\end{theorem}

The proof of Theorem \ref{theo:CwithChannel} is stated in App. \ref{sec:ProofDecompositionCD}-\ref{sec:CardinalityBoundCD}. App. \ref{sec:ProofDecompositionCD} presents the decomposition of the probability distribution. App. \ref{sec:ProofAchievabilityC} presents the achievability result for strictly positive information constraint. The case of null information constraint is stated in App.  \ref{sec:ProofCEqualityIC}. App. \ref{sec:ProofConverseC} presents converse result and App. \ref{sec:CardinalityBoundCD} presents the upper bound on the cardinality of the supports of the auxiliary random variables $W_1$ and $W_2$.

\textit{Proof ideas:} 

$\bullet$ The achievability proof follows by replacing the random variable $V$ with the auxiliary random variable $W_2$, in the block-Markov coding scheme of Theorem \ref{theo:1CwithChannel}. In the block $b+1 $, decoder returns the sequence $V^n_{b+1}$ drawn from the conditional probability distribution $\QQ_{\sf{v|yw_2}}^{\times n}$  depending on the current sequence of channel outputs $Y^n_{b+1}$ and the decoded sequence $W^n_{2,b+1}(m)$, corresponding to index $m\in \mc{M} $. Both Markov chains $Y  -\!\!\!\!\minuso\!\!\!\!-X    -\!\!\!\!\minuso\!\!\!\!-  ( U,W_1 ,W_2)$ and $V  -\!\!\!\!\minuso\!\!\!\!- (Y,W_2)    -\!\!\!\!\minuso\!\!\!\!-  ( U ,X,W_1)$ are satisfied. The causal decoding requires that the symbol $V$ depends on $(Y,W_2)$ but not on $W_1$. 

$\bullet$ The converse proof is obtained by identifying the auxiliary random variables $W_{1,i} =  U^n_{i+1}$ and $W_{2,i} = Y^{i-1}$, in the converse proof of Theorem  \ref{theo:1CwithChannel}. The Markov chain $Y_i  -\!\!\!\!\minuso\!\!\!\!-X_i   -\!\!\!\!\minuso\!\!\!\!-  ( U_i,W_{1,i} , W_{2,i})$ is satisfied since  $Y_i$ is not included in $(W_{1,i} , W_{2,i})$. The Markov chain $V_i  -\!\!\!\!\minuso\!\!\!\!- (Y_i ,W_{2,i})   -\!\!\!\!\minuso\!\!\!\!-  ( U_i,X_i,W_{1,i} )$ is satisfied since  $Y^{i-1}$ is  included in $W_{2,i}$ and the causal decoding function writes $V_i = g_i(Y_i ,Y^{i-1})$, for all $i\in \{1,\ldots,n\}$.

$\bullet$ For the causal decoding case, the probability distribution decomposes as $\PP_{\sf{u}}(u)   \times \QQ(x | u) \times  \mc{T}(y | x )\times \QQ(v | u,x,y)$, whereas for strictly causal decoding, the probability distribution decomposes as $\PP_{\sf{u}}(u)   \times \QQ(x | u)\times \QQ(v | u,x) \times  \mc{T}(y | x )$. The main difference is that for causal decoding, the symbol $V$ is affected by the randomness of the channel $Y$. In that case, the output  $Y$ plays also the role of an information source.

\begin{remark}[The role of the auxiliary random variables]
Theorem \ref{theo:CwithChannel} involves an additional auxiliary random variable $W_2$ that replace $V$ in Theorem \ref{theo:1CwithChannel} and that characterizes the tension between the correlation of $V$ with $Y$ and the correlation of $V$ with $(U,X)$. Intuitively, $W_1$ is related to the channel input $X$, as for Gel'fand Pinsker's coding \cite{gelfand-it-1980}, whereas $W_2$ is related to decoder's output $V$, as for Wyner Ziv's coding \cite{wyner-it-1976}.
\end{remark}

\begin{theorem}\label{theo:ConvexProblemCD}
The set $\mc{A}_{\textsf{d}}$ is convex and the information constraint \eqref{eq:CwithChannel1} is concave with respect to the distribution $\PP_{\sf{u}}(u)   \times \QQ(x | u) \times  \mc{T}(y | x )\times \QQ(v | u,x,y)$. 
\end{theorem}
The proof of Theorem \ref{theo:ConvexProblemCD} is stated in App. \ref{sec:ProofTheoConvexCD}. As for Theorem \ref{theo:ConvexProblem}, the set of achievable distribution using causal decoding is convex.

\begin{remark}[Non-causal encoding and decoding]\label{remark:NonCausalEncDec}
The case of non-causal encoder and decoder is open in general. In the coordination framework, the random variables $(U,Y)$ behave as correlated sources, for which the lossy source coding problem is open. Nevertheless, the optimal solutions have been characterized for three particular cases involving two-sided state information:\\
1) Perfect channel is solved by Theorem IV.1 in \cite{LeTreust(ISIT-TwoSided)15}, and it extends the result of Wyner-Ziv \cite{wyner-it-1976},\\
2) Lossless decoding is solved by Corollary III.3 in \cite{LeTreust(CorrelationITW)14}, and it  extends the result of Gel'fand-Pinsker \cite{gelfand-it-1980},\\
3) Independent source-channel is solved by Theorem IV.2 in \cite{LeTreust(ISIT-TwoSided)15}, and it  extends the result of Merhav-Shamai \cite{MerhavShamai03}.

The duality between channel capacity and rate distortion \cite{CoverChiang02}, can be seen directly on the information constraints of Theorem IV.1 in \cite{LeTreust(ISIT-TwoSided)15} for perfect channel, and of Corollary III.3 in \cite{LeTreust(CorrelationITW)14} for lossless decoding. The problem of empirical coordination has also strong relationship with the problem of ``state communication'' \cite{ChoudhuriKimMitra13} which is solved for Gaussian channels \cite{SutivongChiangCoverKim05}, but remains open for non-causal encoding and decoding. 
\end{remark}

%%%%%%%%%%%%%%%%%%%%%%%%%%%%%%%%%%%%%%%%%%%%%%%%%%%%%%%%%%%%%%%%%%%%%%%%%%%%%%%%%%%%%%%%%%%%%%%%%%%%%%%%%%%%%%%%%%%%%%%%%%%%%%%%%%%%%%%%%%%%%%%%%%%%%%%%%%%%%%%%%%%%%%%%%%%%%%%%%%%%%%%%%%%%%%%%%%%%%%%%%%%%%%%%
%
%

\section{Conclusion}\label{sec:Conclusion}

In this paper, we investigate the problem of empirical coordination for two point-to-point source-channel scenarios, in which the encoder is non-causal and the decoder is strictly causal or causal. 
Empirical coordination characterizes the possibilities of coordinated behavior in a network of autonomous devices. Coordination is measured in terms of  empirical distribution of symbols of source and channel. We characterize the set of achievable target  probability distributions over the symbols of source and channel, for strictly causal and causal decoding. Compression and transmission of information are special cases of empirical coordination and the corresponding information constraint is stronger. We compare our characterization with the previous results stated in the literature and we investigate the maximization problem of a utility function that is common to both encoder and decoder. These results will be extended to the case of distinct utility functions, by using the tools from the repeated game theory.

 %%%%%%%%%%%%%%%%%%%%%%%%%%%%%%%%%%%%%%%%%%
  %%%%%%%%%%%%%%%%%%%%%%%%%%%%%%%%%%%%%%%%%%
 %%%%%%%%%%%%%%%%%%%%%%%%%%%%%%%%%%%%%%%%%%
 %%%%%%%%%%%%%%%%%%%%%%%%%%%%%%%%%%%%%%%%%%

\appendices.

\section{Decomposition of the probability distribution for Theorem \ref{theo:1CwithChannel}}\label{sec:ProofDecomposition}

In order to prove the assertion $1)$ of Theorem \ref{theo:1CwithChannel}, we assume that the joint distribution $\QQ(u,x,y,v)$ is achievable and we introduce the mean probability distribution $ \overline{\PP}_n(u,x,y,v)$ defined by:
\begin{align}
 \overline{\PP}_n(u,x,y,v) &= \frac{1}{n} \cdot \sum_{j=1}^n  \PP(u_j, x_j, y_j ,v_j).\label{eq:ProofDecomposition13}
\end{align}
Lemma \ref{lemma:MarginalDistribution} states that for all $j\in \{1,\ldots,n\}$, the marginal distribution $\PP(u_j, x_j, y_j ,v_j)$ decomposes as: $\PP_{\sf{u}} (u_j) \times \PP(x_j,v_j |u_j) \times \mc{T}(y_j | x_j)$. Hence the mean  distribution $ \overline{\PP}_n(u,x,y,v)$ also decomposes as:
\begin{align}
 \overline{\PP}_n(u,x,y,v) &= \PP_{\sf{u}} (u) \times    \overline{\PP}_n(x,v|u) \times \mc{T}(y | x),
\end{align}
where for each symbol $u \in \mc{U}$ we have: $\overline{\PP}_n(x,v |u) =  \frac{1}{n} \cdot \sum_{j=1}^n  \PP(x_j, v_j | u_j=u )$. Since the joint distribution $\QQ(u,x,y,v)$ is achievable, there exists a code $c\in\mc{C}(n)$ such that the empirical distribution $Q^n(u,x,y,v)$ converges to $\QQ(u,x,y,v)$, with high probability. Convergence in probability implies the convergence in distribution, hence  $ \overline{\PP}_n(u,x,y,v) $ also converges to $\QQ(u,x,y,v)$. This proves that the probability distribution $\QQ(u,x,y,v) = \PP_{\sf{u}} (u) \times \QQ(x,v |u)  \times \mc{T}(y | x) $ satisfies the assertion $1)$ of Theorem \ref{theo:1CwithChannel}.

\begin{remark}[Implementable probability distribution]
The mean probability distribution $ \overline{\PP}_n(U,X,Y,V)$, defined in equation \eqref{eq:ProofDecomposition13}, corresponds to the definition of \textit{implementable probability distribution}, stated in \cite{GossnerHernandezNeyman06}, that is weaker than the definition of empirical coordination, see \cite[Proposition 5]{LarrousseLasaulceBloch(IT)14}. \end{remark}

\begin{lemma}[Marginal distribution]\label{lemma:MarginalDistribution}
Let $\PP\big( u^n, x^n, y^n, v^n\big)$, the joint distribution induced by the coding scheme $c\in \mc{C}(n)$. For all $j\in \{1,\ldots,n\}$, the marginal distribution satisfies:
\begin{eqnarray}
\PP(u_j,x_j,y_j,v_j) = \PP_{\sf{u}} (u_j) \times \PP(x_j,v_j |u_j) \times \mc{T}(y_j | x_j) .
\end{eqnarray}
\end{lemma}
\begin{proof}[Lemma \ref{lemma:MarginalDistribution}]
The notation $u^{-j}$ stands for the sequence $u^n$ where the symbol $u_j$ has been removed: $u^{-j}=\{u_1,\ldots,u_{j-1},u_{j+1},\ldots,u_n\} \in \mc{U}^{n-1}$.
\begin{align}
&\PP\big( u^n, x^n, y^n, v^n\big) \nonumber \\
&= \prod_{i=1}^n \PP_{\sf{u}} (u_i) \times \PP(x^n|u^n)  \times \prod_{i=1}^n \mc{T}(y_i | x_i) \times \prod_{i=1}^n \PP(v_i |y^{i-1})\label{eq:ProofDecomposition1C0} \\
&= \PP_{\sf{u}} (u_j)\times \PP(u^{-j},x^{n}|u_j)  \times \prod_{i=1}^n \mc{T}(y_i | x_i)  \times \prod_{i=1}^n \PP(v_i |y^{i-1})\label{eq:ProofDecomposition1C1} \\
&= \PP_{\sf{u}} (u_j)\times \PP(u^{-j},x^{n}|u_j) \times \mc{T}(y_j | x_j)\nonumber \\
& \times   \PP(y^{-j}|u^n,x^n)  \times \prod_{i=1}^n \PP(v_i |y^{i-1}) \label{eq:ProofDecomposition1C3} \\
 &= \PP_{\sf{u}} (u_j) \times \PP(u^{-j}, x^{n}, y^{-j} , v_j |u_j)  \nonumber \\
& \times   \mc{T}(y_j | x_j)  \times\prod_{i \neq j } \PP(v_i |y^{i-1}) \label{eq:ProofDecomposition1C5} \\
    &= \PP_{\sf{u}} (u_j) \times \PP(x_j,v_j |u_j) \times \mc{T}(y_j | x_j)   \nonumber \\
& \times 
    \PP(v^{-j} ,y^{-j} , x^{-j},u^{-j} |u_j,x_j,y_j,v_j) . \label{eq:ProofDecomposition1C7} 
\end{align}
Equation \eqref{eq:ProofDecomposition1C0} comes from the properties of the i.i.d. information source, the non-causal encoder, the memoryless channel and the strictly causal decoder.\\
Equation \eqref{eq:ProofDecomposition1C1} comes from the i.i.d. property of the information source.\\
Equation \eqref{eq:ProofDecomposition1C3}  comes from the memoryless property of the channel.\\
Equation  \eqref{eq:ProofDecomposition1C5} comes from the strictly causal decoding.\\
Equation  \eqref{eq:ProofDecomposition1C7} concludes Lemma \ref{lemma:MarginalDistribution} by taking the sum over $(v^{-j} ,y^{-j} , x^{-j},u^{-j})$.
\end{proof}

%%%%%%%%%%%%%%%%%%%%%%%%%%%%%%%%%

%%%%%%%%%%%%%%%%%%%%%%%%%%%%%%%%%%%%%%%%%%%%%%%%%%%%%%%%%%%%%%%%%%%%%%%%%%%%%%%%%%%%%%%%%%%%%%%%%%%%%%%%%%%%%%%%%%%%%%%%%%%%%%%%%%%%%%%%%%%%%%%%%%%%%%%%%%%%%%%%%%%%%%%%%%%%%%%%%%%%%%%%%%%%%%%%%%%%%%%%%%%%%%%%%%

 %%%%%%%%%%%%%%%%%%%%%%%%%%%%%%%%%%%%%%%%%%
  %%%%%%%%%%%%%%%%%%%%%%%%%%%%%%%%%%%%%%%%%%
 %%%%%%%%%%%%%%%%%%%%%%%%%%%%%%%%%%%%%%%%%%
 %%%%%%%%%%%%%%%%%%%%%%%%%%%%%%%%%%%%%%%%%%

\section{Achievability proof of Theorem \ref{theo:1CwithChannel}}\label{sec:ProofAchievability}

In order to prove assertion $2)$ of Theorem \ref{theo:1CwithChannel}, we consider a joint probability distribution $\QQ(u,x,w,y,v) \in \Q$ that achieves the maximum in  \eqref{eq:1CwithChannel1}. In this section, we assume that the information constraint \eqref{eq:1CwithChannel1} is satisfied with strict inequality \eqref{eq:SCDstrict}. The case of equality in \eqref{eq:1CwithChannel1} will be discussed in App. \ref{sec:ProofEqualityIC}.
\begin{eqnarray}
&& I( W;Y |V) - I(U;V,W) \\
 &=&  I( W;Y ,V) - I(U,V;W) - I(U;V) >0.\label{eq:SCDstrict}
\end{eqnarray}
There exists a small parameter $\delta>0$, a rate parameter $\textsf{R}\geq 0 $, corresponding to the source coding and a rate parameter $ \textsf{R}_{\sf{L}}\geq 0 $, corresponding to the binning parameter, such that:
\begin{eqnarray}
\textsf{R}  & =&  I( V ; U )  + \delta , \label{eq:AchievabilitySCD1} \\
\textsf{R}_{\sf{L}}  & =&   I( W ; U ,V )  + \delta ,  \label{eq:AchievabilitySCD2}  \\
\textsf{R} +   \textsf{R}_{\sf{L}}   & \leq&   I( W;Y ,V)  -  \delta  .  \label{eq:AchievabilitySCD3}  
\end{eqnarray}

We consider a block-Markov random code $c\in \mc{C}(n\cdot B)$, defined over $B\in \N$ blocks of length $n\in \N$. The total length of the code is denoted by $\tilde{n} = n\cdot B\in \N$.  We denote by $Q^{\widetilde{n} }$, the empirical distribution of the sequences of symbols $(U^{\widetilde{n} },X^{\widetilde{n} },Y^{\widetilde{n} },V^{\widetilde{n} })$ of length $\tilde{n}  \in \N$ and $\QQ_{\sf{uxyv}}$, the target joint probability distribution. The notations $(U_{b_1}^n,U_{b_2}^n,U_{b_3}^n, \ldots, U_{B-1}^n, U_{B}^n)$ stands for the sequences of symbols corresponding to the blocks $b \in \{1,2,3,\ldots,B-1,B\}$, with length $n\in \N$. The parameter $\varepsilon>0$ is involved in both the definition of the typical sequences and the upper bound for the error probability. We prove that for all $\varepsilon>0$, there exists an $\bar{n}\in \N$, such that for all $\tilde{n} = n\cdot B \geq \bar{n}$, there exists a strictly causal code $c \in \mc{C}(n\cdot B)$ that satisfies:
\begin{eqnarray}
\PP_c\bigg(\Big|\Big|Q^{\widetilde{n} } - \QQ_{\sf{uxyv}} \Big|\Big|_{\sf{tv}} \geq \varepsilon\bigg)\leq \varepsilon.
\end{eqnarray}
\begin{figure}[!ht]
\begin{center}
\psset{xunit=0.5cm,yunit=0.5cm}
\begin{pspicture}(-0.5,-3)(14,10)
\psline[linewidth=2pt](0,-2)(14,-2)
\psline[linewidth=2pt](0,0)(14,0)
\psline[linewidth=2pt](0,2)(14,2)
\psline[linewidth=2pt](0,4)(14,4)
\psline[linewidth=2pt](0,6)(14,6)
\psline[linewidth=2pt](0,8)(14,8)
\psline[linewidth=0.5pt,linestyle=dashed ](2,-3)(2,9)
\psline[linewidth=0.5pt,linestyle=dashed ](4,-3)(4,9)
\psline[linewidth=0.5pt,linestyle=dashed ](6,-3)(6,9)
\psline[linewidth=0.5pt,linestyle=dashed ](8,-3)(8,9)
\psline[linewidth=0.5pt,linestyle=dashed ](10,-3)(10,9)
\psline[linewidth=0.5pt,linestyle=dashed ](13,-3)(13,9)
\rput[u](-0.5,8){$U^n$}
\rput[u](-0.5,6){$V^n$}
\rput[u](-0.5,4){$W^n$}
\rput[u](-0.5,2){$X^n$}
\rput[u](-0.5,0){$Y^n$}
\rput[u](-0.5,-2){$\hat{V}^n$}
\rput[u](3,8.8){$b-1$}
\rput[u](5,8.8){$b$}
\rput[u](7,8.8){$b+1$}
\rput[u](9,8.8){$b+2$}
\pspolygon[fillstyle=vlines*,linecolor=red,hatchcolor = red](6,7.8)(6,8.2)(8,8.2)(8,7.8)
\pspolygon[fillstyle=vlines*,linecolor=red,hatchcolor = red](6,5.8)(6,6.2)(8,6.2)(8,5.8)
\pspolygon[fillstyle=vlines*,linecolor=red,hatchcolor = red](4,3.8)(4,4.2)(6,4.2)(6,3.8)
\pspolygon[fillstyle=vlines*,linecolor=red,hatchcolor = red](4,1.8)(4,2.2)(6,2.2)(6,1.8)
\pspolygon[fillstyle=vlines*,linecolor=red,hatchcolor = red](4,-0.2)(4,0.2)(6,0.2)(6,-0.2)
\pspolygon[fillstyle=vlines*,linecolor=red,hatchcolor = red](6,-2.2)(6,-1.8)(8,-1.8)(8,-2.2)
\psline[linewidth=1.5pt,linecolor=red  ]{->}(7,7.7)(7,6.3)
\psline[linewidth=1.5pt,linecolor=red  ]{->}(7,5.7)(5,4.3)
\psline[linewidth=1.5pt,linecolor=red ]{->}(5,3.7)(5,2.2)
\psline[linewidth=1.5pt,linecolor=red ]{->}(5,1.7)(5,0.2)
\psline[linewidth=1.5pt,linecolor=red ]{->}(5,-0.3)(7,-1.8)
\pspolygon[fillstyle=vlines*,linecolor=blue,hatchcolor = blue](8,7.8)(8,8.2)(10,8.2)(10,7.8)
\pspolygon[fillstyle=vlines*,linecolor=blue,hatchcolor = blue](8,5.8)(8,6.2)(10,6.2)(10,5.8)
\pspolygon[fillstyle=vlines*,linecolor=blue,hatchcolor = blue](6,3.8)(6,4.2)(8,4.2)(8,3.8)
\pspolygon[fillstyle=vlines*,linecolor=blue,hatchcolor = blue](6,1.8)(6,2.2)(8,2.2)(8,1.8)
\pspolygon[fillstyle=vlines*,linecolor=blue,hatchcolor = blue](6,-0.2)(6,0.2)(8,0.2)(8,-0.2)
\pspolygon[fillstyle=vlines*,linecolor=blue,hatchcolor = blue](8,-2.2)(8,-1.8)(10,-1.8)(10,-2.2)
\psline[linewidth=1.5pt,linecolor=blue  ]{->}(9,7.7)(9,6.3)
\psline[linewidth=1.5pt,linecolor=blue  ]{->}(9,5.7)(7,4.3)
\psline[linewidth=1.5pt,linecolor=blue ]{->}(7,3.7)(7,2.2)
\psline[linewidth=1.5pt,linecolor=blue ]{->}(7,1.7)(7,0.2)
\psline[linewidth=1.5pt,linecolor=blue ]{->}(7,-0.3)(9,-1.8)
\pspolygon[fillstyle=vlines*,linecolor=vert,hatchcolor = vert](4,7.8)(4,8.2)(6,8.2)(6,7.8)
\pspolygon[fillstyle=vlines*,linecolor=vert,hatchcolor = vert](4,5.8)(4,6.2)(6,6.2)(6,5.8)
\pspolygon[fillstyle=vlines*,linecolor=vert,hatchcolor = vert](2,3.8)(2,4.2)(4,4.2)(4,3.8)
\pspolygon[fillstyle=vlines*,linecolor=vert,hatchcolor = vert](2,1.8)(2,2.2)(4,2.2)(4,1.8)
\pspolygon[fillstyle=vlines*,linecolor=vert,hatchcolor = vert](2,-0.2)(2,0.2)(4,0.2)(4,-0.2)
\pspolygon[fillstyle=vlines*,linecolor=vert,hatchcolor = vert](4,-2.2)(4,-1.8)(6,-1.8)(6,-2.2)
\psline[linewidth=1.5pt,linecolor=vert  ]{->}(5,7.7)(5,6.3)
\psline[linewidth=1.5pt,linecolor=vert  ]{->}(5,5.7)(3,4.3)
\psline[linewidth=1.5pt,linecolor=vert ]{->}(3,3.7)(3,2.2)
\psline[linewidth=1.5pt,linecolor=vert ]{->}(3,1.7)(3,0.2)
\psline[linewidth=1.5pt,linecolor=vert ]{->}(3,-0.3)(5,-1.8)
\end{pspicture}
%\caption{Encoder $\C$ observes the sequence of  symbols $U^n_{i+1}$ of block $i+1\in \{2,\ldots,B\}$ and finds the sequence of  symbols $W^n_{i+1}(m) \in  $ uses the symbols $X^n_i$ of block $i \in  \{1,\ldots,B-1\}$ in order to indicates a correlated sequence of symbols $V^n_{i+1}$ to decoder $\D$. The sequence $X^n_i$ must be correlated with the sequences of symbols $(U^n_i ,V^n_i)$ over block $i  \in  \{1,\ldots,B-1\} $.}
\end{center}
\label{fig:Ach1CwithChannel}
\end{figure}
\begin{itemize}
\item[$\bullet$] \textit{Random codebook.} We generate $| \mc{M}   |= 2^{n   \sf{R}   } $ sequences $V^n(m)$,  drawn from the i.i.d. probability distribution $\QQ_{\sf{v}}^{\times n} $ with index  $m\in \mc{M} $. We generate $| \mc{M}    \times \mc{M}_{\sf{L}}  |= 2^{n (  \sf{R}  +  \sf{R}_{\sf{L}} ) } $ sequences $W^n(m,l)$,  drawn from the i.i.d. probability distribution $\QQ_{\sf{w}}^{\times n} $, independently of $V^n(m)$, with indices  $(m,l) \in \mc{M}  \times \mc{M}_{\sf{L}} $. \\
\item[$\bullet$] \textit{Initialization of the encoder.} The encoder finds the index $m\in \mc{M}$ such that the sequences $\big(U^n_{b_2},V^n(m)\big) \in A_{\varepsilon}^{{\star}{n}}(\QQ)$ of the second block $b_2$, are jointly typical. We denote by $V^n_{b_2} = V^n(m)$, the sequence corresponding to the block $b_2$. During the first block $b_1$, the encoder sends the index $m\in \mc{M}$ using Shannon's channel coding theorem with the codeword $X^n(m)= X_{b_1}^n  \in \mc{X}^n$. More details are provided in Remark \ref{remark:ChannelCode}. The sequences $\big(U_{b_1}^n , X_{b_1}^n \big) \notin A_{\varepsilon}^{{\star}{n}}(\QQ)$ are not jointly typical, in general. The encoder finds the index $m'\in \mc{M}$ such that the sequences  $\big(U^n_{b_3},V^n(m')\big) \in A_{\varepsilon}^{{\star}{n}}(\QQ)$ of the third block $b_3$, are jointly typical. It finds the index $l' \in \mc{M}_{\sf{L}}$ such that the sequences $\big(U^n_{b_2},V^n_{b_2},W^n(m',l') \big) \in A_{\varepsilon}^{{\star}{n}}(\QQ)$ of the second block $b_2$, are jointly typical. We denote by $V^n_{b_3} = V^n(m')$ and $W^n_{b_2} =W^n(m',l') $, the sequences corresponding to the blocks $b_3$ and $b_2$. During the second block, the encoder sends the sequence $X^n_{b_2}$,  drawn from the conditional probability distribution $\QQ_{\sf{x|uvw}}^{\times n}$  depending on the sequences $\big(U^n_{b_2},V^n_{b_2},W^n_{b_2}\big)$.\\
\begin{remark}\label{remark:ChannelCode}
In the first block $b_1$, the index $m\in \mc{M}$ is transmitted using Shannon's channel coding theorem, stated in \cite[pp. 39]{ElGammalKim(book)11}. Equation \eqref{eq:RemarkAchiev3} guarantees that this transmission is reliable. 
\begin{align}
0 &\leq I( W;Y,V) -   I( W ; U ,V ) -  I(U ;V  )  - 3\delta \label{eq:RemarkAchiev1}\\
&= I( W;Y|V) -   I( W ; U |V ) -  I(U ;V  )  - 3\delta \label{eq:RemarkAchiev1b}\\
&\leq I( W;Y|V) -  I(U ;V )  - 3\delta \label{eq:RemarkAchiev1bc}\\
&\leq I(X;Y) - I( U ;V )- 3  \delta \label{eq:RemarkAchiev2}\\
&=I(X;Y) - \sf{R} - 2  \delta  \label{eq:RemarkAchiev3} . 
\end{align}
Equation \eqref{eq:RemarkAchiev1} comes from the equations \eqref{eq:AchievabilitySCD1}, \eqref{eq:AchievabilitySCD2},  \eqref{eq:AchievabilitySCD3}, defined with the parameter $\delta>0$.\\
Equations \eqref{eq:RemarkAchiev1b} and \eqref{eq:RemarkAchiev1bc} come from the properties of the mutual information.\\
Equation \eqref{eq:RemarkAchiev2} comes from the Markov chain property of the channel $ Y -\!\!\!\!\minuso\!\!\!\!-   X  -\!\!\!\!\minuso\!\!\!\!-  (U,V,W)$. \\
Equation \eqref{eq:RemarkAchiev3} comes from the rate parameter $\sf{R}$, defined by equation \eqref{eq:AchievabilitySCD1}.
\end{remark}
\item[$\bullet$] \textit{Initialization of the decoder.} During the first block $b_1$, the decoder returns an arbitrary sequence of symbols $V_{b_1}^n \in \mc{V}^n$. The sequences $\big(U_{b_1}^n , X_{b_1}^n , Y_{b_1}^n , V_{b_1}^n \big) \notin A_{\varepsilon}^{{\star}{n}}(\QQ)$ are not jointly typical, in general. At the end of the first block, the decoder observes the sequence of channel outputs $Y_{b_1}^n \in \mc{Y}^n$ and finds the index $m\in \mc{M} $ such that the sequences $\big(X^n(m) , Y_{b_1}^n \big) \in A_{\varepsilon}^{{\star}{n}}(\QQ)$ are  jointly typical. During the second block $b_2$, it returns the sequence $V^n_{b_2}=V^n(m)$  that corresponds to the index $m\in \mc{M} $. At the end of the second block, the decoder observes $Y_{b_2}^n \in \mc{Y}^n$, recalls $V_{b_2}^n \in \mc{V}^n$ and finds the pair of indices $(m',l') \in \mc{M}   \times \mc{M}_{\sf{L}} $ such that  $\big( Y^n_{b_2} ,V^n_{b_2} ,W^n(m',l') \big) \in A_{\varepsilon}^{{\star}{n}}(\QQ)$ are jointly typical. In the third block $b_3$, the decoder returns the sequence $V^n_{b_3}= V^n(m')$,  that corresponds to the index $m'\in \mc{M} $.\\
\item[$\bullet$] \textit{Encoding function.} At the beginning of the block $b\in \{2,\ldots B-1\}$, the encoder observes the sequence of symbols of source $U^n_{b+1} \in  \mc{U}^n$ of the next block $b+1$. It finds an index $m\in \mc{M}$ such that the sequences  $\big(U^n_{b+1},V^n(m)\big) \in A_{\varepsilon}^{{\star}{n}} (\QQ)$ are jointly typical. The encoder observes the jointly typical sequences of symbols $(U^n_{b},V^n_{b}) \in  \mc{U}^n \times  \mc{V}^n$, of the current block $b$. It finds the index $l\in \mc{M}_{\sf{L}}$ such that the sequences  $\big(U^n_{b},V^n_{b},W^n(m,l) \big) \in A_{\varepsilon}^{{\star}{n}}(\QQ)$ are jointly typical. We denote by $V^n_{b+1} = V^n(m)$ and $W^n_b = W^n(m,l)$, the sequences corresponding to the blocks $b+1$ and $b$. The encoder sends the sequence $X^n_b$,  drawn from the conditional probability distribution $\QQ_{\sf{x|uvw}}^{\times n}$  depending on the sequences $\big(U^n_{b},V^n_{b},W^n_b\big)$ of the block $b$.\\
\item[$\bullet$] \textit{Decoding function.} At the end of the block $b\in \{2,\ldots B-1\}$, the decoder observes the sequence $Y^n_b$ and recalls the sequence $V^n_b$, it returned during the block $b $. It finds the indices $(m,l) \in \mc{M}   \times \mc{M}_{\sf{L}} $ such that  $\big( Y^n_b ,V^n_b ,W^n(m,l) \big) \in A_{\varepsilon}^{{\star}{n}}(\QQ)$ are jointly typical. In the next block $b+1 $, the decoder returns the sequence $V^n_{b+1} = V^n(m) $,  that corresponds to the index $m\in \mc{M} $.\\
\item[$\bullet$] \textit{Last blocks for the encoder.} The encoder finds the index $m\in \mc{M}$ such that the sequences  $\big(U^n_{B},V^n(m)\big) \in A_{\varepsilon}^{{\star}{n}}(\QQ)$ of the last block $B$, are jointly typical. It finds the index $l\in \mc{M}_{\sf{L}}$ such that the sequences $\big(U^n_{B-1},V^n_{B-1},W^n(m,l) \big) \in A_{\varepsilon}^{{\star}{n}}(\QQ)$  of the block $B-1$, are jointly typical. We denote by $V^n_{B} = V^n(m)$ and $W^n_{B-1}  = W^n(m,l) $, the sequences corresponding to the blocks $B$ and $B-1$. During the block $B-1$, the encoder sends the sequence $X^n_{B-1}$,  drawn from the conditional probability distribution $\QQ_{\sf{x|uvw}}^{\times n}$  depending on the sequences $\big(U^n_{B-1},V^n_{B-1},W^n_{B-1}\big)$. During the last block $B$, the encoder sends a sequence $X^n_{B}$  drawn from the conditional probability distribution $\QQ_{\sf{x|uv}}^{\times n}$  depending on the sequences $\big(U^n_{B},V^n_{B}\big)$.  \\
\item[$\bullet$] \textit{Last blocks for the decoder.} At the end of the block $B-1$, the decoder observes the sequence of channel outputs $Y_{B-1}^n \in \mc{Y}^n$, recalls $V_{B-1}^n \in \mc{V}^n$ and finds the pair of indices $(m,l) \in \mc{M}   \times \mc{M}_{\sf{L}} $ such that  $\big( Y^n_{B-1} ,V^n_{B-1} ,W^n(m,l) \big) \in A_{\varepsilon}^{{\star}{n}}(\QQ)$ are jointly typical. In the last block $B$, the decoder returns the sequence $V^n_{B} =V^n(m)$,  that corresponds to the index $m\in \mc{M} $. \\
\item[$\bullet$] \textit{Typical sequences.} If no error occurs in the coding process, the sequences of symbols $\big( U^n_{b} ,W^n_{b} ,X^n_{b} ,Y^n_{b} ,V^n_{b}  \big) \in A_{\varepsilon}^{{\star}{n}}(\QQ)$ are jointly typical, for each block $b \in \{2, \ldots, B-1\}$. The sequences $\big( U^n_{B} ,X^n_{B} ,Y^n_{B} ,V^n_{B}  \big) \in A_{\varepsilon}^{{\star}{n}}(\QQ)$ of the last block $B$ are also jointly typical but the sequences  $\big( U^n_{b_1} ,X^n_{b_1} ,Y^n_{b_1} ,V^n_{b_1} \big) \notin A_{\varepsilon}^{{\star}{n}}(\QQ)$ of the first block $b_1$, are not jointly typical in general.\\ 
\end{itemize}

\begin{figure}[!ht]
\begin{center}
\psset{xunit=0.7cm,yunit=0.7cm}
\begin{pspicture}(0.8,-1)(12,4.5)
\pscircle[linecolor = blue](0.5,0.5){0.35}
\psframe[linecolor = blue](2,0)(3,1)
\psframe(4,0)(5,1)
\pscircle(6.5,0.5){0.35}
\psframe(8,0)(9,1)
\pscircle(4.5,2.5){0.35}
\psframe[linecolor = blue](10,0)(11,1)
\psline[linewidth=1pt,linecolor = blue]{->}(1,0.5)(2,0.5)
\psline[linewidth=1pt]{->}(3,0.5)(4,0.5)
\psline[linewidth=1pt]{->}(5,0.5)(6,0.5)
\psline[linewidth=1pt]{->}(7,0.5)(8,0.5)
\psline[linewidth=1pt]{->}(9,0.5)(10,0.5)
\psline[linewidth=1pt]{->}(4.5,2)(4.5,1)
\psline[linewidth=1pt]{->}(5,2.5)(8.5,2.5)(8.5,1)
\psline[linewidth=1pt,linecolor = blue]{->}(11,0.5)(12,0.5)
\rput(0.5,0.5){\textcolor[rgb]{0.00,0.00,1.0}{$\PP_{\sf{u}}$}}
\rput(2.5,0.5){\textcolor[rgb]{0.00,0.00,1.0}{$\C$}}
\rput(4.5,0.5){$\C$}
\rput(6.5,0.5){$\mc{T}$}
\rput(8.5,0.5){$\D$}
\rput(4.5,2.5){\textcolor[rgb]{0.00,0.00,0.0}{$\PP_{\sf{uv}}$}}
\rput(10.5,0.5){\textcolor[rgb]{0.00,0.00,1.0}{$\D$}}
\rput[u](1.5,0.8){\textcolor[rgb]{0.00,0.00,1.0}{$U^n_{b+1}$}}
\rput[u](12,0.8){\textcolor[rgb]{0.00,0.00,1.0}{$V^n_{b+1}(m)$}}
\psline[linewidth=0.5pt, linestyle = dashed](3.5,-1)(3.5,4.5)
\psline[linewidth=0.5pt, linestyle = dashed](9.5,-1)(9.5,4.5)
\rput[u](5.5,0.8){$X^n_{b}$}
\rput[u](7.5,0.8){$Y^n_{b}$}
\rput[u](5.4,1.65){$(U^n_{b},V^n_{b})$}
\rput[u](8.1,1.65){$V^n_{b}$}
\rput[u](3.5,0.8){\textcolor[rgb]{0.70,0.00,0.0}{${m}$}}
\rput[u](9.5,0.8){\textcolor[rgb]{0.70,0.00,0.0}{${{m}}$}}
\rput[l](0,4){Lossy source code}
\rput[l](0,3.5){block $b+1$}
\rput[l](5.2,4){Channel code}
\rput[l](5.2,3.5){block $b$}
\rput[l](9,4){Lossy source code}
\rput[l](9.5,3.5){block $b+1$}
\rput[l](3.8,-0.7){Two-sided state information}
\end{pspicture}
\caption{This coding scheme is inspired by the block-Markov code represented by Fig. 6, in \cite{LetreustZaidiLasaulce(Allerton)11}. It is the concatenation of a lossy source coding, for the block $b+1$ and a channel coding with two-sided state information, for the block $b$. The encoder determines  the index $m \in \mc{M}$, such that the sequences $\big(U^n_{b+1} , V^n(m) \big) \in A_{\varepsilon}^{{\star}{n}}(\QQ)$ of the block $b+1$, are jointly typical.  During the previous block $b$, the index $m \in \mc{M}$ is sent using the channel coding of Cover and Chiang \cite{CoverChiang02}, where the encoder observes $(U^n_{b} , V^n_{b}) $ and the decoder observes  $V^n_{b} $.  The pair of sequences $(U^n_{b} , V^n_{b}) $ does not influence the statistics of the  channel $\mc{T}(y|x)$.  %As mentioned by Theorem 14 in \cite{LarrousseLasaulceBloch(IT)14}, the result still holds if we consider a state-dependent memoryless channel $\mc{T}_{\sf{y|uvx}}$ where $(U,V)$ are considered as channel states.
}\label{fig:BlockMarkovSCD}
\end{center}
\end{figure}
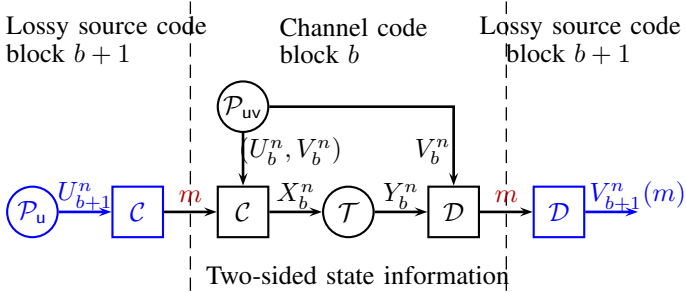

\vspace{-0.5cm}

\begin{remark}[State-dependent channel]\label{remark:StateDependent}
As mentioned by Theorem 14 in \cite{LarrousseLasaulceBloch(IT)14}, this coding scheme also applies to state-dependent memoryless channel $\mc{T}_{\sf{y|uvx}}$ where the symbols $(U,V)$ are the channel states.
\end{remark}

\begin{remark}[Strictly causal decoding with delay $k>1$]\label{remark:FixedDelayofDecoding}
This achievability proof still holds when considering a strictly causal decoding function $V_i = g_i(Y^{i-k})$ with a larger delay $k>1$. In that case, the index $m\in \mc{M}$ corresponding to the future block $b+k$, will be encoded on the current block $b$. The sequences are not jointly typical over the $k>1$ first blocks. 
\end{remark}

\textit{Expected error probability for each block $b \in \{2,\ldots ,B\}$.} We introduce the parameter $\varepsilon_1>0$, in order to provide an upper bound on the expected error probability by block. For all $\varepsilon_1>0$ there exists an $\bar{n} \in \N$ such that for all $n\geq\bar{n}$, the expected probability of the following error events are bounded by $\varepsilon_1$:
\begin{align}
&\E_c\bigg[ \PP\bigg(U^n_{b}     \notin A_{\varepsilon}^{{\star}{n}} (\QQ) \bigg)\bigg]  \leq \varepsilon_1, \label{eq:AchievProba1} \\
&\E_c\bigg[ \PP\bigg( \forall  m \in \mc{M}  , 
\big(U^n_{b}, V^n(m) \big) \notin A_{\varepsilon}^{{\star}{n}}(\QQ) \bigg)\bigg]  \leq \varepsilon_1, \label{eq:AchievProba2} 
\end{align}
\begin{align}
&\E_c\bigg[ \PP\bigg( \forall  l \in \mc{M}_{\sf{L}}   ,\qquad\qquad\qquad\qquad\qquad\qquad\qquad \nonumber \\
&\quad \quad\big(U^n_{b-1}, V^n_{b-1} ,W^n(m,l) \big) \notin A_{\varepsilon}^{{\star}{n}}(\QQ) \bigg)\bigg]  \leq \varepsilon_1, \label{eq:AchievProba3} \\
&\E_c\bigg[ \PP\bigg(  \exists (m', l')\neq  (m ,l)  ,\text{ s.t. } \qquad \nonumber \\
&\quad \quad\big(Y^n_{b-1} , V^n_{b-1}  ,W^n(m',l') \big) \in A_{\varepsilon}^{{\star}{n}}( \QQ)\bigg)\bigg]   \leq \varepsilon_1.\label{eq:AchievProba4}
\end{align}
Equation \eqref{eq:AchievProba1} comes from the properties of typical sequences, stated in \cite[pp. 27]{ElGammalKim(book)11}.\\
Equation \eqref{eq:AchievProba2} comes from equation \eqref{eq:AchievabilitySCD1} and the covering lemma, stated in \cite[pp. 208]{ElGammalKim(book)11}.\\
Equation \eqref{eq:AchievProba3} comes from equation \eqref{eq:AchievabilitySCD2} and the covering lemma, stated in \cite[pp. 208]{ElGammalKim(book)11}.\\
Equation \eqref{eq:AchievProba4} comes from equation \eqref{eq:AchievabilitySCD3} and the packing lemma, stated in \cite[pp. 46]{ElGammalKim(book)11}.\\

We denote by $\textsf{E}_{b} = \big\{(U_b^n, W_b^n , X_b^n, Y_b^n, V_b^n) \in A_{\varepsilon}^{{\star}{n}}(\QQ) \big\} $, the random event corresponding to the jointly typical  sequences for the block $b\in \{2,\ldots ,B-1\}$ and $\textsf{E}_{b}^c$ the complementary event. Equations  \eqref{eq:AchievProba1},  \eqref{eq:AchievProba2},  \eqref{eq:AchievProba3},  \eqref{eq:AchievProba4} imply that for all block $b \in \{2,\ldots ,B-1\}$, for all $n\geq \bar{n}$, there exists a code $c^{\star} \in \mc{C}(n)$ such that:
\begin{align}
\E_c\bigg[ \PP\bigg(\textsf{E}_{b}^c \bigg| \bigcap_{b' \in \{2,\ldots ,b-1\}} \Big\{ \textsf{E}_{b'}  \Big\} \bigg) \bigg] \leq  4\cdot \varepsilon_1. \label{eq:ProbaIntersectionBlock0}
\end{align}  
If the sequences are jointly typical for all the previous block $b' \in \{2,\ldots ,b-1\}$, then the sequences of the current block $b$ are not jointly typical $\textsf{E}_{b}^c$, with probability less than $4 \cdot \varepsilon_1$. Hence, the sequences are jointly typical for all block $b\in \{2,\ldots ,B-1\}$, with probability more than $(1 - 4 \cdot \varepsilon_1)^{B-2} $.
\begin{align}
 &\E_c\bigg[ \PP\bigg(\forall b \in \{2,\ldots ,B-1\}, \;\; \textsf{E}_{b} \bigg) \bigg] \nonumber \\
&= \E_c\bigg[ \PP\bigg(\bigcap_{b \in \{2,\ldots ,B-1\}} \textsf{E}_{b} \bigg) \bigg] \label{eq:ProbaIntersectionBlock1} \\
&= \prod_{b \in \{2,\ldots ,B-1\}} \E_c\bigg[ \PP\bigg(\textsf{E}_{b} \bigg| \bigcap_{b' \in \{2,\ldots ,b-1\}} \Big\{\textsf{E}_{b'}\Big\}  \bigg) \bigg] \label{eq:ProbaIntersectionBlock2} \\
&\geq \prod_{b \in \{2,\ldots ,B-1\}} (1 - 4 \cdot\varepsilon_1)  =  \bigg(1 - 4 \cdot \varepsilon_1\bigg)^{B-2} \label{eq:ProbaIntersectionBlock3}.
\end{align}
Equations \eqref{eq:ProbaIntersectionBlock1} and  \eqref{eq:ProbaIntersectionBlock2} come from the definition of the intersection of probability events.\\
Equation \eqref{eq:ProbaIntersectionBlock3} comes from equation \eqref{eq:ProbaIntersectionBlock0}, that is valid for all block $b \in \{2,\ldots ,B-1\}$.\\

%sequences $(U_b^n, W_b^n , X_b^n, Y_b^n, V_b^n) \in A_{\varepsilon}^{{\star}{n}}(\QQ)$ are jointly typical,  with probability more than $1 - 4\varepsilon$. For all $b \in \{2,\ldots ,B\}$, we have $\E_c\big[ \PP\big(\textsf{E}_{b} | \cap_{b' \in \{2,\ldots ,b-1\}} \textsf{E}_{b'}\big) \big] \geq 1 -  4\varepsilon$. 

\textit{Expected error probability of the block-Markov code.} 

We denote by $Q_{b_1}$, the empirical distribution of the sequences $(U^n_{b_1}, X^n_{b_1}, Y^n_{b_1}, V^n_{b_1})$ of the first block $b_1$ and by $Q^{\widetilde{n} }$, the empirical distribution of the sequences $(U^{\widetilde{n} }, X^{\widetilde{n} }, Y^{\widetilde{n} }, V^{\widetilde{n} })$ of the block-Markov code of length $\tilde{n} = n \cdot B \in \N$. We denote by $(U^{\widehat{n}}, X^{\widehat{n}}, Y^{\widehat{n}}, V^{\widehat{n}})$ the truncated sequences of length $\hat{n} = n \cdot (B-2)$, corresponding to the blocks $\{2, \ldots,B-1\}$, where the first and last blocks  have been removed. We denote by $Q^{\widehat{n} }  \in \Delta(\mc{U}\times \mc{X}\times \mc{Y}\times \mc{V} )$, the empirical distribution of these truncated sequences $(U^{\widehat{n}}, X^{\widehat{n}}, Y^{\widehat{n}}, V^{\widehat{n}})$. We introduce the parameter $\varepsilon_2>0$,  for the definition of the typical sequences and we assume it satisfies $2 \cdot \varepsilon_2 \leq \varepsilon$. If the number of blocks is sufficiently large \textit{i.e.}, $ \frac{4}{\varepsilon_2} \leq B $, we show that $Q^{\widehat{n} }$ is close to $Q^{\widetilde{n} }$. 
\begin{align}
&\Big|\Big|Q^{\widetilde{n} }  - Q^{\widehat{n} }\Big|\Big|_{\sf{tv}}\nonumber \\
&= \Big|\Big|\frac{1}{B} \cdot \Big( (B-2) \cdot Q^{\widehat{n} } +  Q_{b_1}+  Q_{B}\Big)  - Q^{\widehat{n} }\Big|\Big|_{\sf{tv}} \label{eq:ProbaBlock1}\\
&= \frac{1}{B} \cdot \Bigg( \Big|\Big|   Q_{b_1} - Q^{\widehat{n} } \Big|\Big|_{\sf{tv}} +  \Big|\Big|   Q_{B} - Q^{\widehat{n} } \Big|\Big|_{\sf{tv}} \Bigg)\leq \frac{4}{B} \leq\varepsilon_2. \label{eq:ProbaBlock3}
\end{align}
Equation \eqref{eq:ProbaBlock1} comes from the definition of the empirical distribution $Q^{\widetilde{n} } = \frac{1}{B} \cdot \big( (B-2) \cdot Q^{\widehat{n} } +  Q_{b_1} +  Q_B \big)$ that is a convex combination of distributions $Q_{b_1}$, $Q_B$ and $Q^{\widehat{n} }$.\\
Equation \eqref{eq:ProbaBlock3} comes from the upper bound on the total variation distance \cite[eq. (7)]{Verdu(ITA)14} and the large number of blocks: $ \frac{4}{\varepsilon_2} \leq B$. \\

The parameter $\varepsilon_2>0$ for the typical sequences satisfies $2 \cdot \varepsilon_2 \leq \varepsilon$. We provide an upper bound on the expected error probability of the block-Markov code.
\begin{align}
&\E_c\bigg[ \PP_{\sf{e}}({c}) \bigg] \nonumber \\
&\leq  \E_c\bigg[ \PP\bigg(\Big|\Big|Q^{\widetilde{n} } - \QQ  \Big|\Big|_{\sf{tv}}\geq 2 \cdot \varepsilon_2 \bigg) \bigg] \label{eq:ProbaConcatenation1} \\
&=  \E_c\bigg[ \PP\bigg(\Big|\Big|Q^{\widetilde{n} }  - Q^{\widehat{n} } + Q^{\widehat{n} } - \QQ \Big|\Big|_{\sf{tv}}\geq 2 \cdot \varepsilon_2 \bigg) \bigg]\label{eq:ProbaConcatenation2} \\
&\leq  \E_c\bigg[ \PP\bigg(\Big|\Big|Q^{\widetilde{n} }  - Q^{\widehat{n} }\Big|\Big| + \Big|\Big|Q^{\widehat{n} } - \QQ \Big|\Big|_{\sf{tv}}\geq 2 \cdot \varepsilon_2 \bigg) \bigg] \label{eq:ProbaConcatenation3} \\
&\leq  \E_c\bigg[ \PP\bigg(  \Big|\Big|Q^{\widehat{n} } - \QQ \Big|\Big|_{\sf{tv}}\geq 2 \cdot \varepsilon_2 -  \frac{4}{B}  \bigg) \bigg] \label{eq:ProbaConcatenation4} \\
&\leq  \E_c\bigg[ \PP\bigg(  \Big|\Big|Q^{\widehat{n} } - \QQ \Big|\Big|_{\sf{tv}}\geq \varepsilon_2 \bigg) \bigg] \label{eq:ProbaConcatenation5} \\
&= 1 -   \E_c\bigg[ \PP\bigg(  \Big|\Big|Q^{\widehat{n} } - \QQ \Big|\Big|_{\sf{tv}}< \varepsilon_2 \bigg) \bigg] \label{eq:ProbaConcatenation6} \\
&\leq 1 -   \E_c\bigg[ \PP\bigg( \forall b \in \{2,\ldots ,B-1\}, \nonumber \\
& \qquad(U_b^n,  X_b^n, Y_b^n, V_b^n) \in A_{\varepsilon_2}^{{\star}{n}}(\QQ) \bigg) \bigg] \label{eq:ProbaConcatenation7} \\
&\leq  1 - \bigg(1 - 4 \cdot {\varepsilon}_1   \bigg)^{B-2} .\label{eq:ProbaConcatenation8} 
\end{align}
Equation \eqref{eq:ProbaConcatenation1} comes from the choice of the parameters satisfying $2 \cdot \varepsilon_2  \leq  \varepsilon$. By definition of the typical sequences stated in \cite[pp. 25]{ElGammalKim(book)11}, the inequality $\varepsilon'  \leq  \varepsilon$, implies  that $\PP\big(U^n \notin A_{\varepsilon}^{{\star}{n}}(\QQ)\big) \leq \PP\big(U^n \notin A_{ \varepsilon'  }^{{\star}{n}}(\QQ)\big)$.\\
Equation \eqref{eq:ProbaConcatenation3} comes from the triangle inequality.\\
Equations \eqref{eq:ProbaConcatenation4} and \eqref{eq:ProbaConcatenation5} come from the equation \eqref{eq:ProbaBlock3} that requires a large number of blocks $B$ \textit{i.e.}, $ \frac{4}{\varepsilon_2} \leq B$.\\
Equation \eqref{eq:ProbaConcatenation7} comes from Lemma \ref{lemma:concatenation} and the fact that if $A \Rightarrow B$, then we have $\PP(A) \leq \PP(B)$. If the sequences $(U_b^n, X_b^n, Y_b^n, V_b^n)  \in A_{\varepsilon_2}^{{\star}{n}}(\QQ)$ are jointly typical for all block $b \in \{2,\ldots,B-1\}$, then the concatenated sequences are also jointly typical for the same parameter $\varepsilon_2>0$.\\
Equation \eqref{eq:ProbaConcatenation8} comes from equation \eqref{eq:ProbaIntersectionBlock3}.\\

Equation \eqref{eq:ProbaConcatenation8} implies that for all parameter ${\varepsilon}>0$, there exists a parameter $\varepsilon_2>0$ for the typical sequences, there exists a large number of blocks $B$,  there exists a small error probability by block $\varepsilon_1>0$, there exists a large $\bar{n}$ such that for all $\tilde{n}\geq \bar{n} \cdot B$, there exists a  code $c^{\star} \in \mc{C}(\tilde{n})$ with an error probability below $\varepsilon>0$, with:
\begin{eqnarray}
\varepsilon \geq \max\Big(1 - (1 - 4 \cdot {\varepsilon}_1   )^{B-2} , 2 \cdot \varepsilon_2 \Big) . 
\end{eqnarray}
This concludes the achievability proof of Theorem \ref{theo:1CwithChannel}.

\begin{remark}\label{remark:parameters}
The parameter $\varepsilon>0$ determines the appropriate parameters $\varepsilon_2$, $B$, $\varepsilon_1$, $n$, as explained below:\\

$\bullet$ The parameter ${\varepsilon_2}>0$ for the typical sequences depends on the parameter $\varepsilon>0$ and should satisfy equation \eqref{eq:ProbaBlock0}:
\begin{eqnarray}
2 \cdot \varepsilon_2  &\leq  \varepsilon.  \label{eq:ProbaBlock0}
\end{eqnarray}

$\bullet$ The number of blocks $B$ depends on the parameter ${\varepsilon_2}>0$ for the typical sequences and should satisfy equation \eqref{eq:ProbaBlock3}:
\begin{eqnarray}
  \frac{4}{B}
&\leq \varepsilon_2.  \label{eq:ProbaBlock3b}
\end{eqnarray}

$\bullet$  The parameter $\varepsilon_1$ of the error probability by block, depends on the number of blocks $B$, of the parameter $\varepsilon$ and should satisfy  equation \eqref{eq:ProbaConcatenation8}:
 \begin{eqnarray}
1 - \bigg(1 - 4 \cdot {\varepsilon}_1   \bigg)^{B-2} & \leq& \varepsilon.\label{eq:ProbaConcatenation8b} 
\end{eqnarray}

$\bullet$  The length $n\in\N$ of each block depends on the parameter $\varepsilon_1$ of the error probability by block and the parameter $\varepsilon_2$ of the typical sequences. For each $b \in \{2,\ldots,B-1\}$, the parameter $n$ should satisfy the equation \eqref{eq:ProbaIntersectionBlock0}:
\begin{eqnarray}
\E_c\bigg[ \PP\bigg(\textsf{E}_{b}^c \bigg| \bigcap_{b' \in \{2,\ldots ,b-1\}} \Big\{ \textsf{E}_{b'}  \Big\} \bigg) \bigg] \leq  4\cdot \varepsilon_1. \label{eq:ProbaIntersectionBlock0b}
\end{eqnarray}  \\
\end{remark}

\begin{lemma}\label{lemma:concatenation}
If the sequences $U_b^n \in A_{\varepsilon}^{{\star}{n}}(\QQ)$ are jointly typical for all block $b \in \{1,\ldots,B\}$, then the concatenated sequence $U^{\widetilde{n} } = (U^n_{b_1} , \ldots, U^n_{B} )\in A_{\varepsilon}^{{\star}{n}}(\QQ)$, is also jointly typical for the same parameter $\varepsilon>0$.
\end{lemma}

\begin{proof}
The proof of Lemma \ref{lemma:concatenation} is based on the triangle inequality. We consider a target probability distribution $\QQ\in\Delta(\mc{U})$. We denote by $Q^{\widetilde{n} }$ the empirical distribution of the concatenated sequence $U^{\widetilde{n} }= (U^n_{b_1} , \ldots, U^n_{B} )$ and $(Q_{b_1},Q_{b_2}, \ldots,Q_{B})$ the empirical distributions of the sequences $(U^n_{b_1} ,U^n_{b_2} , \ldots, U^n_{B} )$ over the blocks $b \in \{1,\ldots,B\}$. 
\begin{align}
&\Big|\Big|Q^{\widetilde{n} }  - \QQ\Big|\Big|_{\sf{tv}} \nonumber \\
&=  \Big|\Big|\frac{1}{B} \cdot (Q_{b_1} + Q_{b_2} + \ldots + Q_{B})  - \QQ\Big|\Big|_{\sf{tv}}\label{eq:lemmaTriangle1} \\
%&=  \Big|\Big|\frac{1}{B} \cdot (Q_{b_1} + Q_{b_2} + \ldots + Q_{B}  - B \cdot  \QQ)\Big|\Big|_{\sf{tv}} \\
&=  \frac{1}{B} \cdot \Big|\Big|Q_{b_1} + Q_{b_2} + \ldots + Q_{B}  - B \cdot  \QQ\Big|\Big|_{\sf{tv}} \label{eq:lemmaTriangle2}\\
&\leq  \frac{1}{B} \cdot \bigg( \Big|\Big|Q_{b_1} - \QQ\Big|\Big|_{\sf{tv}} + \Big|\Big|Q_{b_2} - \QQ\Big|\Big|_{\sf{tv}} + \ldots + \Big|\Big|Q_{B} - \QQ\Big|\Big|_{\sf{tv}} \bigg) \label{eq:lemmaTriangle3} \\
&\leq  \frac{1}{B} \cdot \bigg( B \cdot \varepsilon \bigg) = \varepsilon. \label{eq:lemmaTriangle4}
\end{align}
Equation \eqref{eq:lemmaTriangle1} comes from the definition of the empirical distribution: $Q^{\widetilde{n} } = \frac{1}{B} \cdot \sum_{b \in \{1,\ldots,B\}} Q_b $ as a convex combination of the empirical distributions $Q_b$, with $b \in \{1,\ldots,B\}$.\\
Equations \eqref{eq:lemmaTriangle2} and \eqref{eq:lemmaTriangle3} come from the properties of the total variation distance and the  triangle inequality.\\
%Equation \eqref{eq:lemmaTriangle3} comes from the triangle inequality.\\
Equation \eqref{eq:lemmaTriangle4} comes from the hypothesis of typical sequences $U_b^n$, for all block $b \in \{1,\ldots,B\}$.\\
\end{proof}

%%%%%%%%%%%%%%%%%%%%%%%%%%%%%%%%%%%%%%%%%%%%%%%%%%%%%%%%%%%%%%%%%%%%%%%%%%%%%%%%%%%%%%%%%%%%%%%%%%%%%%%%%%%%%%%%%%%%%%%%%%%%

\section{Equality in the information constraint \eqref{eq:1CwithChannel1}}\label{sec:ProofEqualityIC}

We consider a target distribution $ \QQ(u,x,y,v)= \PP_{\sf{u}}(u)   \times \QQ(x,v | u) \times  \mc{T}(y | x )$ such that the maximum in the information constraint \eqref{eq:1CwithChannel1} is equal to zero:
\begin{align}
\max_{{\QQ}\in \Q} \bigg( I( W;Y  |V )  -   I(  U ; V  ,W  )  \bigg)=0. \label{eq:EqualityIC}
\end{align}

\subsection{First case: channel capacity is strictly positive}
The channel is not trivial and it is possible to send some reliable information: 
\begin{align}
\max_{\PP(x)} I(X;Y)  >0. \label{eq:1Capacity}
\end{align}
We denote by $\QQ^{\star}(u,x,y,v)=\PP_{\sf{u}}(u) \times \PP^{\star}(x) \times \mc{T}(y|x) \times \QQ(v)$ the distribution where $U$, $V$ and $(X,Y)$ are mutually independent and where $\PP^{\star}(x)$ achieves the maximum in \eqref{eq:1Capacity}. We denote by $\max_{{\QQ^{\star}}\in \Q^{\star}} \bigg( I_{\QQ^{\star}}( W;Y  |V )  -   I_{\QQ^{\star}}( U ; V  ,W  )  \bigg) $ the information constraint corresponding to $\QQ^{\star}(u,x,y,v)$ and we show that it is strictly positive:
\begin{align}
&\max_{{\QQ^{\star}}\in \Q^{\star}} \bigg( I_{\QQ^{\star}}( W;Y  |V )  -   I_{\QQ^{\star}}( U ; V  ,W  )  \bigg) \nonumber \\
&= I(X;Y) - I(U;V) \label{eq:MutualIndep1}\\
& = I(X;Y)>0.   \label{eq:MutualIndep2}
\end{align}
Equation \eqref{eq:MutualIndep1} comes from Corollary \ref{coro:Shannon}, since $(U,V)$ are independent of $(X,Y)$.\\
Equation \eqref{eq:MutualIndep2} comes from the independence between $U$ and $V$ and the last inequality comes from the hypothesis of strictly positive channel capacity \eqref{eq:1Capacity}.

We construct the sequence $\{\QQ^k(u,x,y,v)\}_{k\in\N}$ of convex combinations between the target distribution $\QQ(u,x,y,v)$ and the distribution $\QQ^{\star}(u,x,y,v)$ where for all $(u,x,y,v)$:
\begin{align}
\QQ^k(u,x,y,v)=\frac{1}{k} \cdot\bigg( (k-1) \cdot \QQ(u,x,y,v) + \QQ^{\star}(u,x,y,v)\bigg).\nonumber
%,\nonumber \\\qquad\qquad\qquad\forall \; (u,x,y,v).
\end{align}
We denote by $\max_{{\QQ^k}\in \Q^k} \bigg( I_{\QQ^k}( W;Y  |V )  -   I_{\QQ^k}( U ; V  ,W  )  \bigg) $ the information constraint corresponding to the distribution $\QQ^k(u,x,y,v)$.\\
$\bullet$ The information constraint   \eqref{eq:EqualityIC} corresponding to $\QQ(u,x,y,v)$ is equal to zero.\\
$\bullet$ The information constraint  \eqref{eq:MutualIndep2} corresponding to $\QQ^{\star}(u,x,y,v)$ is strictly positive.\\
By Theorem \ref{theo:ConvexProblem}, the information constraint is concave with respect to the distribution. Hence, the information constraint corresponding to the distribution $\QQ^k(u,x,y,v)$ is strictly positive for all $k>1$:
\begin{align}
\max_{{\QQ^k}\in \Q^k} \bigg(  I_{\QQ^k}( W;Y  |V )  -    I_{\QQ^k}(  U ; V  ,W  )  \bigg)>0. \label{eq:ConvexCom}
\end{align}
 App. \ref{sec:ProofAchievability} concludes that the distribution $\QQ^k(u,x,y,v)$ is achievable, for all $k>1$. Moreover, the distribution $\QQ^k(u,x,y,v)$ converges to the target distribution $\QQ(u,x,y,v)$, as $k$ goes to $+\infty$. This proves that the target distribution $\QQ(u,x,y,v) $ is achievable.

\subsection{Second case: channel capacity is equal to zero}

We assume that the channel capacity is equal to zero: $\max_{\PP(x)} I(X;Y) =0$, and we show that the set of achievable distributions $\mc{A}$ defined in \eqref{eq:AchievableSet}, boils down to the set of distributions $\big\{  \QQ(v) \times  \QQ(x | u,v)\big\}$.
\begin{align}
\mc{A} &= \bigg\{  \QQ(x,v | u), \text{ s.t. } \nonumber \\
&\max_{\QQ(w|u,v,x),\atop |\mc{W}| \leq  |\mc{U} \times \mc{X} \times \mc{V}| +1} \bigg( I( W;Y  |V )  -   I(  U ; V  ,W  )  \bigg) \geq 0\bigg\} \label{eq:Boundary00} \\
&= \bigg\{   \QQ(x,v | u), \qquad \qquad\text{ s.t. }\qquad I(U;V) =0\bigg\} \label{eq:Boundary01} \\
&= \bigg\{   \QQ(v) \times  \QQ(x | u,v) \bigg\} \label{eq:Boundary02}.
\end{align}
Equation \eqref{eq:Boundary00} comes from the definition of $\mc{A}$, stated in \eqref{eq:AchievableSet}.\\
Equation \eqref{eq:Boundary01} comes from the hypothesis of channel capacity equal to zero $\max_{\PP(x)} I(X;Y) =0$ and Lemma \ref{lemma:ZeroCapacityIC}.\\
Equation \eqref{eq:Boundary02} comes from Lemma \ref{lemma:decomposition2}.

\textit{Coding Scheme:} We consider a target distribution $ \PP_{\sf{u}}(u)   \times \QQ(v) \times \QQ(x | u,v)\times  \mc{T}(y)$ that belongs to the set defined by \eqref{eq:Boundary02}.\\
$\bullet$ The sequence $V^n$ is drawn with the i.i.d. probability $\QQ(v)$ and known in advance by both encoder and decoder,\\
$\bullet$ The encoder observes the sequence of source $U^n$ and generates the sequence $X^n$ with the i.i.d. probability distribution   $\QQ(x|u,v)$, depending on  the pair $(U^n,V^n)$. \\
$\bullet$ This coding scheme proves that any probability distribution that belongs to the set defined by \eqref{eq:Boundary02} is achievable:
\begin{align} 
 \PP_{\sf{u}}(u)   \times \QQ(v) \times \QQ(x | u,v) \times  \mc{T}(y ).
\end{align}

This proves that if $\max_{\PP(x)} I(X;Y) =0$, then any distribution $  \PP_{\sf{u}}(u)   \times\QQ(x,v|u) \times  \mc{T}(y|x ) \in \mc{A}$  that belongs to the set defined by \eqref{eq:Boundary00} is achievable.

\begin{remark}[Correlation between $U$ and $V$]
The distribution $\QQ^1(u,x,y,v)$ stated in \eqref{eq:trivialcode1} is achievable with a trivial coding scheme,  whereas the distribution $\QQ^2(u,x,y,v)$ stated in \eqref{eq:trivialcode2} is achievable using the coding scheme presented in  App. \ref{sec:ProofAchievability}.
\begin{align}
\QQ^1(u,x,y,v)=\PP_{\sf{u}}(u)   \times \QQ(v) \times \QQ(x| u,v) \times  \mc{T}(y | x ),\label{eq:trivialcode1}\\
\QQ^2(u,x,y,v)= \PP_{\sf{u}}(u)   \times \QQ(v|u) \times \QQ(x| u,v) \times  \mc{T}(y | x ).\label{eq:trivialcode2}
\end{align}
The price of the correlation between the symbols $V$ and $U$ is captured by the information constraint \eqref{eq:1CwithChannel1}.
\end{remark}

%such that $\max_{\QQ(w|u,v,x),\atop |\mc{W}| \leq  |\mc{U} \times \mc{X} \times \mc{V}| +1} \bigg( I( W;Y  |V )  -   I(  U ; V  ,W  )  \bigg) \geq 0$, is achievable.}

\begin{lemma}\label{lemma:ZeroCapacityIC}
We consider both sets of probability distributions:
\begin{align}
\mc{A} &= \bigg\{  \QQ(x,v | u), \text{ s.t. }\nonumber \\
&\max_{\QQ(w|u,v,x),\atop |\mc{W}| \leq  |\mc{U} \times \mc{X} \times \mc{V}| +1} \bigg( I( W;Y  |V )  -   I(  U ; V  ,W  )  \bigg) \geq 0\bigg\} \\
\mc{B}&= \bigg\{   \QQ(x,v | u), \qquad \qquad\text{ s.t. }\qquad I(U;V) =0\bigg\}.
\end{align}
If the channel capacity is equal to zero: $\max_{\PP(x)} I(X;Y) =0$, then both sets of probability distributions are equal $\mc{A} = \mc{B}$.
\end{lemma}

\begin{proof}[Lemma \ref{lemma:ZeroCapacityIC}]
\textit{First inclusion $\mc{A} \subset \mc{B}$.} We consider a distribution $\QQ(x,v | u)$ that belongs to $\mc{A}$ and we denote by $W$ the auxiliary random variable that achieves the maximum in the information constraint. Since the channel capacity is equal to zero, we have:
\begin{align}
 0 &= \max_{\PP(x)} I(X;Y)   \geq I(W;Y,V) \nonumber \\
 &\geq I(W;U,V) + I(U;V) \geq I(U;V)\geq0. \label{eq:Capacity05}
\end{align}
This implies that $I(U;V)=0$, hence the distribution $\QQ(x,v | u)$ belongs to the set  $\big\{   \QQ(x,v | u),  \;\; \text{ s.t. }  \;\;  I(U;V) =0\big\}$. This proves the first inclusion $\mc{A} \subset \mc{B}$.\\
\textit{Second inclusion $\mc{A} \supset \mc{B}$.} We consider a distribution $\QQ(x,v | u)$ that belongs to $\mc{B}$, hence for which $I(U;V)=0$. We introduce a deterministic auxiliary variable $\widetilde{W}$, for which $I(\widetilde{W};Y,V) = I(\widetilde{W};U,V)=0$. Hence the information constraint of the distribution $\QQ(x,v | u)$ satisfies:
\begin{align}
&\max_{\QQ(w|u,v,x),\atop |\mc{W}| \leq  |\mc{U} \times \mc{X} \times \mc{V}| +1} \bigg( I( W;Y  |V )  -   I(  U ; V  ,W  )  \bigg)\nonumber \\
 &\geq I(\widetilde{W};Y,V) - I(\widetilde{W};U,V) - I(U;V) = 0.
\end{align}
The information constraint corresponding to $\QQ(x,v | u)$ is positive, hence $\QQ(x,v | u)$ belongs to $\mc{A}$. This shows the second inclusion $\mc{A} \supset \mc{B}$.
\end{proof}

\begin{lemma}\label{lemma:decomposition2}
%We denote by $\QQ(u)$, $\QQ(v)$ and $\QQ(x|u,v)$ the marginals of $ \QQ(u,v,x)$. Hence, b
Both sets of probability distributions are equal:
\begin{align} 
&\bigg\{  \PP_{\sf{u}}(u)   \times \QQ(x,v | u),  \;\; \text{ s.t. } \;\; I(U;V) =0\bigg\} \nonumber \\
&= \bigg\{ \PP_{\sf{u}}(u)   \times  \QQ(v) \times  \QQ(x | u,v) \bigg\}.
\end{align}
\end{lemma}

\begin{proof}[Lemma \ref{lemma:decomposition2}]
We consider that the set $\mc{U}$, $\mc{V}$, $\mc{X}$ are fixed.
\begin{align}
&\bigg\{ \QQ(u,v,x), \quad \text{ s.t. } \quad I(U;V) = 0\bigg\}\nonumber \\
&=\bigg\{ \QQ(u,v) \times \QQ(x|u,v), \; \text{ s.t. } \; I(U;V) = 0\bigg\}\label{eq:LemmaDecomposition1}\\
&=\bigg\{ \QQ(u)   \times \QQ(v) \times \QQ(x|u,v)\bigg\}.\label{eq:LemmaDecomposition2}
\end{align}
Equation \eqref{eq:LemmaDecomposition1} comes from the decomposition of the probability distribution $ \QQ(u,v,x)= \QQ(u,v) \times \QQ(x|u,v)$.\\
Equation \eqref{eq:LemmaDecomposition2} comes from the equivalence $I(U;V) = 0 \Longleftrightarrow \QQ(u,v) = \QQ(u)   \times \QQ(v)$. \\
When the marginal distribution $\PP_{\sf{u}}(u)  $ is fixed, we have:
\begin{align} 
& \bigg\{  \PP_{\sf{u}}(u)   \times \QQ(x,v | u),  \;\; \text{ s.t. } \;\; I(U;V) =0\bigg\} \nonumber \\
&= \bigg\{ \PP_{\sf{u}}(u)   \times  \QQ(v) \times  \QQ(x | u,v) \bigg\}.
\end{align}
\end{proof}

%%%%%%%%%%%%%%%%%%%%%%%%%%%%%%%%%%%%%%%%%%%
%%%%%%%%%%%%%%%%%%%%%%%%%%%%%%%%%%%%%%%%%%%

\section{Converse proof of Theorem \ref{theo:1CwithChannel}.}\label{sec:ProofConverse}

We suppose that the target joint probability distribution $\PP_{\sf{u}}(u)   \times \QQ(x, v | u) \times  \mc{T}(y | x )$ is achievable with a strictly causal code. For all $\varepsilon>0$, there exists a minimal length $\bar{n}\in \N$, such that for all $n \geq \bar{n}$, there exists a code $c\in\mc{C}(n)$, such that the probability of error satisfies $\PP_{\sf{e}}(c) = \PP\big(||{Q}^n - \QQ  ||_{\sf{tv}} > \varepsilon \big) \leq \varepsilon$. The parameter $\varepsilon>0$ is involved in both the definition of the typical sequences and the upper bound of the error probability. We introduce the random  event of error $E \in \{0,1\}$ defined as follows:
\begin{eqnarray}
E = \Bigg\{
\begin{array}{lll}
0 &\text{ if }& \big|\big|{Q}^n - \QQ  \big|\big|_{\sf{tv}} \leq \varepsilon\quad\\
%&& \Longleftrightarrow \quad (U^n,X^n,Y^n,V^n) \in A_{\varepsilon}^{{\star}{n}}(\QQ) ,\\
1 &\text{ if }& \big|\big|{Q}^n - \QQ  \big|\big|_{\sf{tv}} > \varepsilon \quad \\
%&&\Longleftrightarrow \quad (U^n,X^n,Y^n,V^n) \notin A_{\varepsilon}^{{\star}{n}}(\QQ).
\end{array}
\Bigg.\label{eq:DefErrorEvent}
\end{eqnarray}
%Consider a sequence of code $c(n) \in \mc{C}$ that achieves the probability distribution $\QQ(u,x,y,v)$ \textit{i.e.}, for which the probability of error  $\PP_{\sf{e}}(c) = \PP(E=1)$ goes to zero. 
%For all $\varepsilon>0$, there exists a minimal length $\bar{n}\in \N$ such that for all $n \geq \bar{n}$ there exists a code $c\in\mc{C}(n)$ such that the probability of error  $\PP_{\sf{e}}(c) = \PP(E=1) \leq \varepsilon$. 
%The probability of error satisfies $\PP_{\sf{e}}(c) = \PP\big( E = 1 \big) \leq \varepsilon$. 
We have the following equations:
\begin{align}
0&= \sum_{i=1}^n I(U^n_{i+1} ; Y_i  | Y^{i-1})
 - \sum_{i=1}^n I(   Y^{i-1}  ; U_i  | U^n_{i+1} ) \label{eq:Conv1SCwithChannel1} \\
&=  \sum_{i=1}^n I( U^n_{i+1} ; Y_i  | Y^{i-1}) 
- \sum_{i=1}^n I(    Y^{i-1} , U^n_{i+1} ; U_i  ) \label{eq:Conv1SCwithChannel2} \\
&=  \sum_{i=1}^n I( U^n_{i+1} ; Y_i   | Y^{i-1}, V_i) 
- \sum_{i=1}^n I(    Y^{i-1} , U^n_{i+1},  V_i; U_i  ) \label{eq:Conv1SCwithChannel2b} \\
&\leq  \sum_{i=1}^n I(  Y^{i-1},  U^n_{i+1} ; Y_i  | V_i) 
- \sum_{i=1}^n I(    Y^{i-1} , U^n_{i+1} , V_i; U_i )  \label{eq:Conv1SCwithChannel2c} \\
&=  \sum_{i=1}^n I(   W_{i} ; Y_i     | V_i) - \sum_{i=1}^n I(  W_{i} , V_i   ; U_i  ) \label{eq:Conv1SCwithChannel3} .
\end{align}
Equation \eqref{eq:Conv1SCwithChannel1} comes from Csisz\'{a}r Sum Identity stated in \cite[pp. 25]{ElGammalKim(book)11}.\\
Equation \eqref{eq:Conv1SCwithChannel2} comes from the i.i.d. property of the information source $U$, that implies $ I(U_i   ; U^n_{i+1} ) = 0$ for all $i\in\{1 , \ldots,n\}$.\\
Equation \eqref{eq:Conv1SCwithChannel2b} comes from the strictly causal decoding property that implies the following Markov chain:
\begin{align}
V_i -\!\!\!\!\minuso\!\!\!\!- Y^{i-1} -\!\!\!\!\minuso\!\!\!\!-  (U_i , Y_i,U^n_{i+1} ).
\end{align}
Equation \eqref{eq:Conv1SCwithChannel2c} comes from the properties of the mutual information.\\
Equation \eqref{eq:Conv1SCwithChannel3} comes from the introduction of the auxiliary random variable $W_{i} = (Y^{i-1}  , U^n_{i+1})$ that satisfies the Markov Chain of the set of probability distributions $\Q$.
\begin{align}
&Y_i -\!\!\!\!\minuso\!\!\!\!- X_i  -\!\!\!\!\minuso\!\!\!\!-  (U_i , W_{i} ,V_i )  .
\end{align}
This Markov chain comes from the memoryless property of the channel and the fact that $Y_i$ does not belong to $W_{i}= (Y^{i-1}  , U^n_{i+1}) $.\\
\begin{align}
0&\leq  \sum_{i=1}^n I(   W_{i} ; Y_i      | V_i) - \sum_{i=1}^n I(  W_{i} , V_i   ; U_i )  \\
&=  n \cdot \bigg( I(  W_{T}  ; Y_T   | V_T, T) -  I(  W_{T} , V_T  ; U_T  | T) \bigg)\label{eq:Conv1SCwithChannel4} \\
&\leq  n \cdot \bigg( I(W_{T}  ,T ; Y_T  | V_T ) -  I(  W_{T} ,T , V_T    ; U_T ) \bigg)\label{eq:Conv1SCwithChannel5} \\
&\leq n \cdot \max_{\QQ \in \Q} \bigg( I(W  ; Y_T  | V_T ) -  I(  W , V_T   ; U_T) \bigg) \label{eq:Conv1SCwithChannel6} \\
&\leq  n \cdot \max_{\QQ \in \Q} \bigg(   I(W  ; Y_T | V_T,E =0 )\nonumber \\
& -  I(  W , V_T   ; U_T | E=0 ) + \varepsilon  \bigg) 
\label{eq:Conv1SCwithChannel7} \\
&\leq  n \cdot \max_{\QQ \in \Q} \bigg( I(W  ; Y | V  ) -  I(  W , V   ; U ) + 2\varepsilon \bigg). \label{eq:Conv1SCwithChannel9} 
\end{align}
Equation \eqref{eq:Conv1SCwithChannel4} comes from the introduction of the uniform random variable $T$ over the indices $\{1,\ldots,n\}$ and the introduction of the corresponding mean random variables $U_T$,  $W_{T}$,   $X_T$, $Y_T$, $V_T$.\\
Equation \eqref{eq:Conv1SCwithChannel5} comes from the i.i.d. property of the information source that implies $I(T ; U_T) =0$.\\
Equation \eqref{eq:Conv1SCwithChannel6} comes from identifying $W$ with  $( W_{T},T)$ and taking the  maximum over the probability distributions that belong to $\Q$. This is possible since the random variable $W =( W_{T},T)$ satisfies the Markov chain of the set of probability distributions $\Q$.\\
Equation \eqref{eq:Conv1SCwithChannel7} comes from the empirical coordination framework, as stated in Lemma \ref{lemma:ErrorEventCoordinationSCD}. 
By hypothesis, the sequences are not jointly typical with small error probability $\PP_{\sf{e}}(c) = \PP(E=1)$. Lemma \ref{lemma:ErrorEventCoordinationSCD} adapts the proof of Fano's inequality to the empirical coordination requirement.\\
%There exists a minimal length $\bar{n}\in \N$ such that for all $n\geq \bar{n}$, there exists a code $c\in\mc{C}(n)$ such that the sequences are not jointly typical with small error probability $\PP_{\sf{e}}(c) = \PP(E=1) \leq \varepsilon$.\\
Equation \eqref{eq:Conv1SCwithChannel9} comes from Lemma \ref{lemma:3}. The probability distribution induced by the coding scheme $ \PP\big((U_T,X_T,Y_T ,V_T)= (u,x,y,v) \big| E=0\big)$ is closed to the target probability distribution $\PP_{\sf{u}}(u)   \times \QQ(x,v | u) \times  \mc{T}(y | x )$. The continuity of the entropy function stated in \cite[pp. 33]{CsiszarKorner(Book)11} implies equation \eqref{eq:Conv1SCwithChannel9}.

If the joint probability distribution $\PP_{\sf{u}}(u)   \times \QQ(x,v | u) \times  \mc{T}(y | x )$ is achievable with a strictly causal code, then the following equation is satisfied for all $\varepsilon>0$:
\begin{align}
0&\leq& \max_{\QQ \in \Q} \bigg( I(W  ; Y | V  ) -  I(  W , V   ; U ) + 2\varepsilon \bigg).
\end{align}
This concludes the converse proof of Theorem \ref{theo:1CwithChannel}.

\begin{remark}[Stochastic encoder and decoder]
This converse result still holds when considering stochastic encoder and decoder instead of deterministic encoder and decoder.
\end{remark}

%\begin{remark}[Two-way state-dependent channel with two-sided state information]
\begin{remark}[Channel feedback observed by the encoder]
\label{remark:FeedbackConverse}
The converse proof of Theorem \ref{theo:1CwithChannel} is based on the following assumptions:\\
$\bullet$ The information source $U$ is i.i.d distributed with $\PP_{\sf{u}}(u)$.\\
$\bullet$ The decoding function is strictly causal $V_i = g_i(Y^{i-1})$, for all $i\in\{1,\ldots,n\}$.\\
$\bullet$ The auxiliary random variables $W_i = (Y^{i-1} , U^n_{i+1} )$  satisfies the Markov chain $Y_i  -\!\!\!\!\minuso\!\!\!\!-X_i    -\!\!\!\!\minuso\!\!\!\!-  (U_i,V_i,W_i)$ of the channel, for all $i\in \{1 , \ldots,n\}$. \\
$\bullet$ The sequences of random variables $(U^n,X^n,Y^n,V^n)$ are jointly typical for the target probability distribution $\PP_{\sf{u}}(u)   \times \QQ(x,v | u) \times  \mc{T}(y | x )$, with high probability.

As mentioned in \cite{LarrousseLasaulceWigger(ITW)15}, each step of the converse holds when the encoder $X_i = f_i(U^n,Y_1^{i-1})$ observes the channel feedback $Y_1^{i-1}$ with $i\in\{1,\ldots,n\}$, drawn from the memoryless channel $\mc{T}(y_1,y | x )$. The encoder ignores the channel feedback since it arrives too late to be exploited by the strictly causal decoder. 
\end{remark}

\begin{lemma}[Fano's inequality for coordination]\label{lemma:ErrorEventCoordinationSCD}
We fix a probability distribution $\QQ \in \Q$ and we suppose that the error probability $\PP(E=1)$ (see definition in \eqref{eq:DefErrorEvent}) is small enough such that $\PP(E=1) \cdot \log_2 |\mc{Y}| + 2 \cdot h_b\Big(\PP(E=1)\Big) \leq \varepsilon$. Then equation \eqref{eq:ErrorCoordinationSCD} is satisfied.
\begin{eqnarray}
&& I(W  ; Y_T  |  V ) -  I( W ,  V_T   ; U_T  ) \nonumber\\
&\leq&   I(W  ; Y_T |  V , E =0 ) -  I( W ,  V_T  ; U_T | E=0 ) + \varepsilon.\nonumber \\&&\label{eq:ErrorCoordinationSCD}
\end{eqnarray}
%The arguments for proving equation \eqref{eq:ErrorCoordinationSCD} are used to prove the equation \eqref{eq:ErrorCoordinationSCD}, that is involved in the converse proof of Theorem \ref{theo:CwithChannel}.
%\begin{eqnarray}
%&& I(W_1  ; Y_T  |  W_2 ) -  I(  W_2   ; U_T  | W_1  ) \nonumber \\
%&\leq&   I(W_1  ; Y_T |  W_2 , E =0 ) -  I(  W_2   ; U_T | W_1 ,E=0 ) + \varepsilon. \label{eq:ErrorCoordinationSCD}
%\end{eqnarray}
\end{lemma}

\begin{proof}
The proof of Lemma \ref{lemma:ErrorEventCoordinationSCD} comes from the properties of the mutual information.
\begin{align*}
& I(W  ; Y_T  | V_T ) -  I(  W , V_T   ; U_T  ) \\
&= I(W  ; Y_T  | V_T, E  ) -  I(  W,  V_T   ; U_T | E )\\
&+ I(E ; Y_T   | V_T ) - I(E ; Y_T  | W, V_T ) \nonumber \\
&-  I(  E   ; U_T ) +  I(  E ; U_T  | W, V_T)\\
&\leq  I(W  ; Y_T  | V_T, E  ) -  I(  W,  V_T   ; U_T | E )+ 2 \cdot H(E)\\
&= \PP(E=0) \cdot \Big( I(W  ; Y_T  |  V_T, E =0 )\nonumber \\
&-   I(  W,  V_T   ; U_T | E  = 0)  \Big)  + 2 \cdot H(E) + \PP(E=1)\nonumber \\
&\times \Big( I(W  ; Y_T  | V_T,  E =1 ) -  I(  W,  V_T   ; U_T | E  = 1) \Big)  \\
&\leq   I(W   ; Y_T  | V_T,  E =0 ) -  I(  W,  V_T   ; U_T | E  = 0)   \nonumber \\
&+ \PP(E=1) \cdot \log_2 |\mc{Y} | + 2 \cdot h_b\Big(\PP(E=1)\Big)\\
&\leq   I(W  ; Y_T  |  V_T,  E =0 ) -  I(  W,  V_T   ; U_T | E  = 0)  + \varepsilon.\nonumber 
\end{align*}
This concludes the proof of Lemma \ref{lemma:ErrorEventCoordinationSCD}.
\end{proof}

\begin{lemma}\label{lemma:3}
Probability distribution $ \PP\big((U_T,X_T,Y_T ,V_T)= (u,x,y,v) \big| E=0\big)$ is closed to the target probability distribution $\QQ(u,x,y,v)$:
\begin{eqnarray}
 \bigg|  \PP\Big((U_T,X_T,Y_T ,V_T)= (u,x,y,v) \Big| E=0\Big) \nonumber \\
 -  \QQ(u,x,y,v) \bigg|  \leq \varepsilon.
\end{eqnarray} 
\end{lemma}

\begin{proof}[Lemma \ref{lemma:3}]
We fix a symbol $u \in \mc{U}$ and we evaluate the probability $\PP(U_T = u | E=0)$. We show it is closed to the desired probability $\PP_{\sf{u}}(u)$:
\begin{align}
&\PP(U_T = u | E=0) \nonumber \\
&=\sum_{u^n \in A_{\varepsilon}^{{\star}{n}} } \sum_{i = 1}^n  \PP\big(U^n = u^n , T = i , U_T = u  \big| E=0\big)  \label{eq:Lemma3_Eq1} \\
&=\sum_{u^n \in A_{\varepsilon}^{{\star}{n}} } \sum_{i = 1}^n   \PP\big(U^n = u^n   \big| E=0\big)  \nonumber \\
&\times  \PP\big( T = i   \big| U^n = u^n , E=0\big) \nonumber \\
&\times   \PP\big(U_T = u  \big| U^n = u^n , T = i ,  E=0\big)  \label{eq:Lemma3_Eq2} \\
&=\sum_{u^n \in A_{\varepsilon}^{{\star}{n}} } \sum_{i = 1}^n   \PP\big(U^n = u^n   \big| E=0\big)\nonumber \\
&\times   \PP\big( T = i \big) \cdot \PP\big(U_T = u  \big|U^n = u^n , T = i ,  E=0\big)  \label{eq:Lemma3_Eq3} \\
&=\sum_{u^n \in A_{\varepsilon}^{{\star}{n}} } \sum_{i = 1}^n   \PP\big(U^n = u^n   \big| E=0\big)\cdot   \PP\big( T = i \big) \cdot \UN(u_i = u)\label{eq:Lemma3_Eq4} \\
&=\sum_{u^n \in A_{\varepsilon}^{{\star}{n}} }   \PP\big(U^n = u^n   \big| E=0\big)\cdot  \sum_{i = 1}^n  \frac{1}{n} \cdot \UN(u_i = u)\label{eq:Lemma3_Eq5} \\
&=\sum_{u^n \in A_{\varepsilon}^{{\star}{n}} }   \PP\big(U^n = u^n   \big| E=0\big)\cdot    \frac{\textsf{N}(u|u^n)}{n} .\label{eq:Lemma3_Eq6}
\end{align}
Equation \eqref{eq:Lemma3_Eq3} comes from the independence of event $\{T = i\} $ with events $\{U^n = u^n\}$ and $\{E=0\}$.\\
Equation \eqref{eq:Lemma3_Eq6} comes from the definition of the number of occurrences $\textsf{N}(u|u^n) = \sum_{i = 1}^n   \UN(u_T = u )$.\\

Since the sequences $u^n \in A_{\varepsilon}^{{\star}{n}}$ are typical, we have the following equation:
\begin{eqnarray}
 \PP_{\sf{u}}(u) - \varepsilon  \leq&  \frac{\textsf{N}(u|u^n)}{n} & \leq  \PP_{\sf{u}}(u) + \varepsilon .
\end{eqnarray}
This provides an upper bound and a lower bound on  $\PP(U_T = u | E=0)$:
\begin{align}
& \PP_{\sf{u}}(u) - \varepsilon \nonumber \\
 &= \sum_{u^n \in A_{\varepsilon}^{{\star}{n}} }   \PP\big(U^n = u^n   \big| E=0\big)\cdot  \Big( \PP_{\sf{u}}(u) - \varepsilon\Big)\\ 
 &\leq   \PP(U_T = u | E=0)\\
&\leq   \sum_{u^n \in A_{\varepsilon}^{{\star}{n}} }   \PP\big(U^n = u^n   \big| E=0\big)\cdot  \Big( \PP_{\sf{u}}(u) + \varepsilon\Big)\\ 
&=  \PP_{\sf{u}}(u) + \varepsilon, \\
&\Longleftrightarrow \bigg|  \PP(U_T = u | E=0) -  \PP_{\sf{u}}(u) \bigg|  \leq \varepsilon.
\end{align}
Using the same arguments, we prove that $ \PP\Big((U_T,X_T,Y_T ,V_T)= (u,x,y,v) \Big| E=0\Big)$ is closed to the target probability distribution $\QQ(u,x,y,v)$:
\begin{eqnarray}
 \bigg|  \PP\Big((U_T,X_T,Y_T ,V_T)= (u,x,y,v) \Big| E=0\Big)  \nonumber \\
 -  \QQ(u,x,y,v) \bigg|  \leq \varepsilon.
\end{eqnarray} 

This concludes the proof of Lemma \ref{lemma:3}.

\end{proof}

 %%%%%%%%%%%%%%%%%%%%%%%%%%%%%%%%%%%%%%%%%%
  %%%%%%%%%%%%%%%%%%%%%%%%%%%%%%%%%%%%%%%%%%
 %%%%%%%%%%%%%%%%%%%%%%%%%%%%%%%%%%%%%%%%%%
 %%%%%%%%%%%%%%%%%%%%%%%%%%%%%%%%%%%%%%%%%%
 %%%%%%%%%%%%%%%%%%%%%%%%%%%%%%%%%%%%%%%%%%
  %%%%%%%%%%%%%%%%%%%%%%%%%%%%%%%%%%%%%%%%%%
 %%%%%%%%%%%%%%%%%%%%%%%%%%%%%%%%%%%%%%%%%%
 %%%%%%%%%%%%%%%%%%%%%%%%%%%%%%%%%%%%%%%%%%

\section{Bound on the cardinality of $|\mc{W}|$ for Theorem \ref{theo:1CwithChannel}}\label{sec:CardinalityBound}

This section is similar to the App. C, in \cite[pp. 631]{ElGammalKim(book)11}. Lemma \ref{lemma:CardinalityBound1C} relies on the support Lemma and the Lemma of Fenchel-Eggleston-Carathéodory, stated in \cite[pp. 623]{ElGammalKim(book)11}.

%This result is inspired by App. C in \cite{ElGammalKim(book)11} page 631.
%\begin{lemma}[Fenchel-Eggleston-Carathéodory Lemma]\label{lemma:Fenchel}
%Let $\mc{C}\subset \R^d$ a connected compact set.\\
%For all $x \in \conv \C$, there exists a convex combination of $d\in \N$ points such that $x = \sum_{i =1}^d \alpha_i x_i$.
%%\begin{eqnarray}
%%x&=& \sum_{i =1}^d \alpha_i x_i
%%\end{eqnarray}
%\end{lemma}
%\begin{lemma}[Support Lemma]\label{lemma:Support}
%Let $\mc{U}$, $\mc{X}$ and $\mc{V}$ be finite sets and $\mc{W}$ be an arbitrary set. Let $\Delta(\mc{U} \times\mc{X} \times  \mc{V})$ the set of probability on $\mc{U} \times \mc{X} \times \mc{V}$. Let $g_i :  \Delta(\mc{U} \times \mc{X} \times \mc{V}) \mapsto \R$ with $i = \big\{1 , \ldots,  |\mc{U} \times \mc{X} \times\mc{V}| \big\}$ be continuous functions.\\
%Then for every continuous random variable $W \sim \PP_{\sf{w}}$ defined on $\mc{W}$, there exists a random variable $W' \sim \PP_{\sf{w}'}$  defined on a  set $\mc{W}'$ with finite cardinality $|\mc{W}'|\leq  |\mc{U} \times \mc{X} \times\mc{V}| $ such that:
%\begin{eqnarray*}
%\int_{\mc{W}} g_i(\PP_{\sf{uxv|w}})d \PP_{\sf{w}}(w) &=& \sum_{w'\in\mc{W}'} g_i(\PP_{\sf{uxv|w}'}) \PP_{\sf{w}'}(w'),\\
%&&\qquad \forall i\in \{1, \ldots , |\mc{U} \times \mc{X} \times\mc{V}|\}.
%\end{eqnarray*}
%\end{lemma}
%Lemma \ref{lemma:Support} is based on Lemma \ref{lemma:Fenchel}.

\begin{lemma}[Cardinality bound for Theorem \ref{theo:1CwithChannel}]\label{lemma:CardinalityBound1C}
The cardinality of the support of the auxiliary random variable $W$ of the Theorem \ref{theo:1CwithChannel}, is bounded by $|\mc{W}| \leq   |\mc{U}  \times \mc{X} \times \mc{V}| +1$.
\end{lemma}

\begin{proof}[Lemma \ref{lemma:CardinalityBound1C}] 
We fix a symbol $w\in \mc{W}$ and we consider the following continuous functions from $\Delta(\mc{U} \times\mc{X} \times  \mc{V})$ into $\R$:
\begin{eqnarray*}
h_i\Big(\PP(u,x,v|w)\Big) = \qquad\qquad\qquad\qquad\qquad\qquad\qquad\\
\begin{cases}
\PP(u,x,v|w)  ,\; \text{ for } i=\big\{ 1,\ldots,  |\mc{U} \times \mc{X} \times\mc{V}| -1\big\},\\
H(Y|V,W=w) ,\quad \text{ for } i=  |\mc{U} \times \mc{X} \times\mc{V}|,\\
H(U|V,W=w)  ,\quad \text{ for } i=  |\mc{U} \times \mc{X} \times\mc{V}| +1.
\end{cases}
\end{eqnarray*}
The conditional entropies $H(Y|V,W=w)$, $H(U|V,W=w)$ are evaluated with respect to the probability distribution $\PP(u,x,v|w) \times \mc{T}(y|x)$, with a fixed $w\in \mc{W}$. The support Lemma, stated in \cite[pp. 631]{ElGammalKim(book)11}, implies that there exists an auxiliary random variable $W' \sim \PP(w')$  defined on a  set $\mc{W}'$ with bounded cardinality $|\mc{W}'|\leq  |\mc{U} \times \mc{X} \times\mc{V}| +1 $, such that:
\begin{eqnarray*}
H(Y|V,W)  &=& \int_{\mc{W}} H(Y|V,W = w) d \PP_{\sf{w}}(w)  \\
&=&   \sum_{w'\in\mc{W}'} H(Y|V,W' = w')  \PP_{\sf{w}'}(w')\\
&=& H(Y|V,W'),\\
H(U|V,W)  &=& \int_{\mc{W}} H(U|V,W = w) d \PP_{\sf{w}}(w)  \\
&=&   \sum_{w'\in\mc{W}'} H(U|V,W' = w')  \PP_{\sf{w}'}(w') \\
&=& H(U|V,W'),\\
\PP(u,x,y,v) &=& \mc{T}(y|x) \times \int_{\mc{W}} \PP(u,v,x | w) d \PP_{\sf{w}}(w) \\
&=& \sum_{w'\in\mc{W}'} \PP(u,v,x | w') \PP_{\sf{w}'}(w') , 
\end{eqnarray*}
for all $ (u,x,v)$ indexed by $i=\big\{ 1,\ldots,  |\mc{U} \times \mc{X}  \times\mc{V} |  +1\big\}$.
Hence, the probability distribution $\PP(u,x,y,v)$ and the conditional entropies $H(Y|V,W)$, $H(U|V,W) $ are  preserved. The information constraint writes:
\begin{eqnarray*}
&&I( W;Y  |V )  -   I(  U ; V  ,W  ) \\
 &= & H(Y  |V )  - H( Y  |V, W )  -   H(  U  )  + H(  U |V  ,W  ) \\
 &= & H(Y  |V )  - H( Y  |V, W' )  -   H(  U  )  + H(  U |V  ,W'  ) \\
&=& I( W';Y  |V )  -   I(  U ; V  ,W'  ) ,
\end{eqnarray*}
with $ |\mc{W}'| \leq  |\mc{U} \times \mc{X} \times \mc{V}| +1$. This concludes the proof of Lemma \ref{lemma:CardinalityBound1C}, for the cardinality bound on the support of the auxiliary random variable $W$ of Theorem \ref{theo:1CwithChannel}.
\end{proof}

 %%%%%%%%%%%%%%%%%%%%%%%%%%%%%%%%%%%%%%%%%%
  %%%%%%%%%%%%%%%%%%%%%%%%%%%%%%%%%%%%%%%%%%
 %%%%%%%%%%%%%%%%%%%%%%%%%%%%%%%%%%%%%%%%%%
 %%%%%%%%%%%%%%%%%%%%%%%%%%%%%%%%%%%%%%%%%%

%%%%%%%%%%%%%%%%%%%%%%%%%%%%%%%%%%%%%%%%%%
%%%%%%%%%%%%%%%%%%%%%%%%%%%%%%%%%%%%%%%%%%%%%%%%%%%%%%%%%%%%%%%%%%%%%%%%%%%%%%%%%%%%%%%%%%%%%%%%%%%%%%%%%%%%%%%%%%%%%%%%%%%%%%%%%%%%%%%%%%%%%%%%%%%%%%%%%%%%%%%%%%%%%%%%

\section{Sketch of proof of Corollary \ref{coro:CoordTrans}}\label{sec:ProofCorollaryTrans}

\subsection{Achievability proof of Corollary \ref{coro:CoordTrans}.} 

 We consider a target information rate $\textsf{R} \geq 0$ and a joint probability distribution $\QQ(u,x,w,y,v) \in \Q$ that achieves the maximum in equation \eqref{eq:CoordTrans}. We split the index $m$ into a pair of indices $(m_1,m_2)$. The information message is encoded using the first index $m_1$, with the rate parameter $ \textsf{R}_1 $. The second index $m_2$ of rate $ \textsf{R}_2 $, has the same role as the index $m$ in the proof of Theorem  \ref{theo:1CwithChannel}. We consider a block-Markov random code $c\in \mc{C}(n\cdot B)$ and we prove that the pair of rate and probability distribution $(\textsf{R} , \QQ_{\sf{uxyv}})$ is achievable. There exists a $\delta>0$  and rate parameters $ \textsf{R}_1  \geq 0$,  $ \textsf{R}_2  \geq 0$, $  \textsf{R}_{\sf{L}}  \geq 0$ such that:
\begin{eqnarray}
\textsf{R}_1  & \geq&  \textsf{R} -  \delta , \label{eq:AchievabilityTransSCD0} \\
\textsf{R}_2  & =&  I( V ; U )  + \delta , \label{eq:AchievabilityTransSCD1} \\
\textsf{R}_{\sf{L}}  & =&   I( W ; U ,V )  + \delta ,  \label{eq:AchievabilityTransSCD2}  \\
\textsf{R}_1 + \textsf{R}_2 +   \textsf{R}_{\sf{L}}   & \leq&   I( W;Y ,V)  -  \delta  .  \label{eq:AchievabilityTransSCD3}  
\end{eqnarray}

\begin{itemize}
\item[$\bullet$] \textit{Random codebook.} We generate $| \mc{M}_2   |= 2^{n   \sf{R}_2   } $ sequences $V^n(m_2)$,  drawn from the i.i.d. probability distribution $\QQ_{\sf{v}}^{\times n} $ with index  $m_2\in \mc{M}_2 $. We generate $| \mc{M}_1    \times\mc{M}_2    \times \mc{M}_{\sf{L}}  |= 2^{n ( \sf{R}_1 +  \sf{R}_2  +  \sf{R}_{\sf{L}} ) } $ sequences $W^n(m_1,m_2,l)$,  drawn from the i.i.d. probability distribution $\QQ_{\sf{w}}^{\times n} $ with indices  $(m_1,m_2,l) \in \mc{M}_1  \times \mc{M}_2  \times \mc{M}_{\sf{L}} $, independently of $V^n(m_2)$. \\
\item[$\bullet$] \textit{Encoding function.} At the beginning of the block $b\in \{2,\ldots B-1\}$, the encoder observes the sequence of symbols of source $U^n_{b+1} \in  \mc{U}^n$ of the next block $b+1$. It finds an index $m_2\in \mc{M}_2$ such that the sequences  $\big(U^n_{b+1},V^n(m_2)\big) \in A_{\varepsilon}^{{\star}{n}} (\QQ)$ are jointly typical. The encoder observes the information message $m_{1}\in \mc{M}_1 $ and the jointly typical sequences of symbols $(U^n_{b},V^n_{b}) \in  \mc{U}^n \times  \mc{V}^n$ of the current block $b$. It finds the index $l\in \mc{M}_{\sf{L}}$ such that the sequences  $\big(U^n_{b},V^n_{b},W^n_b(m_1,m_2,l) \big) \in A_{\varepsilon}^{{\star}{n}}(\QQ)$ are jointly typical. We denote by $V^n_{b+1}=V^n(m_2)$ and $W^n_b=W^n(m_1,m_2,l)$, the sequences corresponding to the blocks $b+1$ and $b$. The encoder sends the sequence $X^n_b$,  drawn from the conditional probability distribution $\QQ_{\sf{x|uvw}}^{\times n}$  depending on the sequences $\big(U^n_{b},V^n_{b},W^n_b \big)$ of the block $b$.\\
\item[$\bullet$] \textit{Decoding function.} At the end of the block $b\in \{2,\ldots B-1\}$, the decoder observes the sequence $Y^n_b$ and recalls the sequence $V^n_b$, it returned during the block $b $. It finds the indices $(m_1,m_2,l) \in \mc{M}_1   \times\mc{M}_2   \times \mc{M}_{\sf{L}} $ such that  $\big( Y^n_b ,V^n_b ,W^n(m_1,m_2,l) \big) \in A_{\varepsilon}^{{\star}{n}}(\QQ)$ are jointly typical. In the next block $b+1 $, the decoder returns the sequence $V^n_{b+1} = V^n(m_2)$  that corresponds to the index $m_2\in \mc{M}_2 $. The decoder returns the  message $m_1\in \mc{M}_1 $ corresponding to the transmission during the block $b$. \\
\item[$\bullet$] \textit{Rate of the transmitted information.} If no error occurs during the block-Markov coding process, the decoder returns $B$ messages, corresponding to $B \cdot n \cdot \textsf{R}_1$ information bits. Since the length of the code is $\tilde{n} = n \cdot B$, the corresponding information rate is $\textsf{R}_1$, that is close \eqref{eq:AchievabilityTransSCD0} to the target rate $\textsf{R}$.\\
\end{itemize}

For all $\varepsilon_1>0$ there exists an $\bar{n} \in \N$ such that for all $n\geq\bar{n}$, the expected probability of the following error events are bounded by $\varepsilon_1$:
\begin{align}
&\E_c\bigg[ \PP\bigg(U^n_{b}     \notin A_{\varepsilon}^{{\star}{n}} (\QQ) \bigg)\bigg]  \leq \varepsilon_1, \label{eq:AchievProbaTrans1} \\
&\E_c\bigg[ \PP\bigg( \forall  m_2 \in \mc{M}_2  ,\quad 
\big(U^n_{b}, V^n(m_2) \big) \notin A_{\varepsilon}^{{\star}{n}}(\QQ) \bigg)\bigg]  \leq \varepsilon_1,  \nonumber\\
&\label{eq:AchievProbaTrans2} \\
&\E_c\bigg[ \PP\bigg( \forall  l \in \mc{M}_{\sf{L}}   ,\quad \nonumber\\
&\big(U^n_{b-1}, V^n_{b-1} ,W^n(m_1,m_2,l) \big) \notin A_{\varepsilon}^{{\star}{n}}(\QQ) \bigg)\bigg]  \leq \varepsilon_1, \label{eq:AchievProbaTrans3} \\
&\E_c\bigg[ \PP\bigg(  \exists (m_1',m_2', l')\neq  (m_1,m_2 ,l)  ,\text{ s.t. }  \nonumber\\
&
\big(Y^n_{b-1} , V^n_{b-1}  ,W^n(m_1',m_2',l') \big) \in A_{\varepsilon}^{{\star}{n}}( \QQ)\bigg)\bigg]   \leq \varepsilon_1.\label{eq:AchievProbaTrans4}
\end{align}
Equation \eqref{eq:AchievProbaTrans1} comes from the properties of typical sequences, stated  in \cite[pp. 27]{ElGammalKim(book)11}.\\
Equation \eqref{eq:AchievProbaTrans2} comes from equation \eqref{eq:AchievabilityTransSCD1} and the covering lemma, stated  in \cite[pp. 208]{ElGammalKim(book)11}.\\
Equation \eqref{eq:AchievProbaTrans3} comes from equation \eqref{eq:AchievabilityTransSCD2} and the covering lemma, stated in \cite[pp. 208]{ElGammalKim(book)11}.\\
Equation \eqref{eq:AchievProbaTrans4} comes from equation \eqref{eq:AchievabilityTransSCD3} and the packing lemma, stated in \cite[pp. 46]{ElGammalKim(book)11}.\\

The expected error probability of this block-Markov code is upper bounded, using the same arguments as for the achievability proof of Theorem \ref{theo:1CwithChannel}, stated in App. \ref{sec:ProofAchievability}.

\subsection{Converse proof of Corollary \ref{coro:CoordTrans}.} 

We suppose that the target rate $\textsf{R}\geq0$ and the target joint probability distribution $\PP_{\sf{u}}(u)   \times \QQ(x, v | u) \times  \mc{T}(y | x )$ are achievable with a strictly causal code. 
\begin{align}
&\log_2|\mc{M}|  \nonumber\\
&= H(M) = I( M ;Y^n ) + H(M|Y^n) \label{eq:ConvTransSCD1}  \\
&\leq I( M ;Y^n ) + n \cdot \varepsilon \label{eq:ConvTransSCD2}  \\
&= \sum_{i=1}^n I( U^n_{i+1} ,M ; Y_i | Y^{i-1} ) -   \sum_{i=1}^n I( U^n_{i+1}  ; Y_i  | M, Y^{i-1} )  + n  \varepsilon \label{eq:ConvTransSCD4}  \\
&= \sum_{i=1}^n I( U^n_{i+1} ,M ; Y_i | Y^{i-1} )   -   \sum_{i=1}^n I( Y^{i-1}  ; U_i  | M, U^n_{i+1}  )  + n  \varepsilon \label{eq:ConvTransSCD5}  \\
&= \sum_{i=1}^n I( U^n_{i+1} ,M ; Y_i | Y^{i-1} )   -   \sum_{i=1}^n I(U^n_{i+1} ,M, Y^{i-1}  ; U_i   )  + n  \varepsilon \label{eq:ConvTransSCD6}  \\
&= \sum_{i=1}^n I( U^n_{i+1} ,M ; Y_i | Y^{i-1} ,V_i)  \nonumber\\
& -  \sum_{i=1}^n I(U^n_{i+1} ,M, Y^{i-1} ,V_i ; U_i   )  + n \cdot \varepsilon \label{eq:ConvTransSCD7}  \\
&\leq \sum_{i=1}^n I( U^n_{i+1} ,M, Y^{i-1}  ; Y_i | V_i)  \nonumber\\
& -   \sum_{i=1}^n I(U^n_{i+1} ,M, Y^{i-1} ,V_i ; U_i   )  + n \cdot \varepsilon \label{eq:ConvTransSCD8}  \\
&= \sum_{i=1}^n I( W_i  ; Y_i | V_i)  -   \sum_{i=1}^n I(W_i,V_i ; U_i   )  + n \cdot \varepsilon \label{eq:ConvTransSCD9} .
\end{align}
Equation \eqref{eq:ConvTransSCD1} comes from the uniform distribution of the information message $m \in \mc{M}$ that implies $\log_2|\mc{M}| = H(M) $.\\
Equation \eqref{eq:ConvTransSCD2} comes from Fano's inequality, stated in \cite[pp. 19]{ElGammalKim(book)11}.\\ 
Equation \eqref{eq:ConvTransSCD4} comes from the properties of the mutual information.\\ 
Equation \eqref{eq:ConvTransSCD5} comes from Csisz\'{a}r Sum Identity, stated in \cite[pp. 25]{ElGammalKim(book)11}.\\
Equation \eqref{eq:ConvTransSCD6} comes from the i.i.d. property of the information source  that implies $ I(U_i   ; U^n_{i+1},M ) = 0$ for all $i\in\{1 , \ldots,n\}$.\\
Equation \eqref{eq:ConvTransSCD7} comes from the strictly causal decoding property that implies the following Markov chain:
\begin{eqnarray}
V_i -\!\!\!\!\minuso\!\!\!\!- Y^{i-1} -\!\!\!\!\minuso\!\!\!\!-  (U_i , Y_i,U^n_{i+1},M ).
\end{eqnarray}
Equation \eqref{eq:ConvTransSCD8} comes from the properties of the mutual information.\\
Equation \eqref{eq:ConvTransSCD9} comes from the introduction of the auxiliary random variable $W_{i} = ( U^n_{i+1} , M, Y^{i-1} )$ that satisfies the Markov Chain of the set of probability distributions $\Q$, for all $i\in\{1,\ldots,n\}$.
\begin{eqnarray}
&&Y_i -\!\!\!\!\minuso\!\!\!\!- X_i  -\!\!\!\!\minuso\!\!\!\!-  (U_i , W_{i} ,V_i )  .
\end{eqnarray}

We follows the arguments of the converse proof of Theorem \ref{theo:1CwithChannel}, stated in App. \ref{sec:ProofConverse}. The same conclusion holds for the auxiliary random variable $W = (W_T,T)$. If the pair of target rate $\textsf{R}>0$ and target joint probability distribution $\PP_{\sf{u}}(u)   \times \QQ(x,v | u) \times  \mc{T}(y | x )$ are achievable with a strictly causal code, then the following equation is satisfied for all $\varepsilon>0$:
\begin{eqnarray}
\textsf{R} &\leq& \max_{\QQ \in \Q} \bigg( I(W  ; Y | V  ) -  I(  W , V   ; U ) + \varepsilon \bigg).
\end{eqnarray}
This concludes the converse proof of Corolarry \ref{theo:1CwithChannel}.

%%%%%%%%%%%%%%%%%%%%%%%%%%%%%%%%%%%%%%%%%%
%%%%%%%%%%%%%%%%%%%%%%%%%%%%%%%%%%%%%%%%%%%%%%%%%%%%%%%%%%%%%%%%%%%%%%%%%%%%%%%%%%%%%%%%%%%%%%%%%%%%%%%%%%%%%%%%%%%%%%%%%%%%%%%%%%%%%%%%%%%%%%%%%%%%%%%%%%%%%%%%%%%%%%%%

\section{Proof of Theorem \ref{theo:AchievableUtility}}\label{sec:ProofCoroUtility}

%\begin{proof}[Theorem \ref{theo:AchievableUtility}]
\subsection{Achievability proof of Theorem \ref{theo:AchievableUtility}}

We consider a utility $\phi^{\star} \in  \textsf{U} $ and the corresponding probability distribution $\QQ^{\star}(x,v|u) \in \mc{A}$, satisfying $\E_{ \QQ^{\star}} \big[\Phi(U,X,Y,V)\big]  = \phi^{\star}$. Theorem \ref{theo:1CwithChannel} guarantees that the conditional probability distribution $\QQ^{\star}(x,v|u) $ is achievable. Hence, there exists a sequence of code $c\in\mc{C}(n)$, whose empirical distributions $Q^n(u,x,y,v)$ converge in probability to the target joint probability distribution $\QQ^{\star}(u,x,y,v) =  \PP_{\sf{u}}(u)   \times \QQ^{\star}(x,v | u) \times  \mc{T}(y | x )$. We denote by $\Phi^n(c)$, the expected utility corresponding to the code $c\in\mc{C}(n)$:
\begin{align}
&\Phi^n(c) = \E \Bigg[ \frac{1}{n} \cdot  \sum_{i=1}^n \Phi(U_i ,X_i ,Y_i ,V_i )\Bigg] \nonumber\\
&=\E \Bigg[ \frac{1}{n} \cdot  \sum_{u,x,y,v} \textsf{N}(u,x,y,v |U^n,X^n,Y^n,V^n)\times   \Phi(u,x,y,v)\Bigg] \label{eq:ProofCoroAchievableUtilityFunction1}  \\
&=\E \Bigg[  \;\; \E_{Q^n} \bigg[  \Phi(U,X,Y,V)  \bigg]\Bigg] .\label{eq:ProofCoroAchievableUtilityFunction2}  
\end{align}
Equations \eqref{eq:ProofCoroAchievableUtilityFunction1} and \eqref{eq:ProofCoroAchievableUtilityFunction2} come from the definition of the empirical distribution $Q^n$ of the random sequences $(U^n,X^n,Y^n,V^n)$, stated in Definition \ref{def:Code}.

The convergence in probability of $Q^n(u,x,y,v)$ toward $\QQ^{\star}(u,x,y,v) =  \PP_{\sf{u}}(u)   \times \QQ^{\star}(x,v | u) \times  \mc{T}(y | x )$ implies that $\Phi^n(c) $ converges to $ \phi^{\star} \in  \textsf{U}$, that is achievable.

\subsection{Converse proof of Theorem \ref{theo:AchievableUtility}}

For the converse proof of Theorem \ref{theo:AchievableUtility}, we consider that the utility $\phi^{\star} $ is achievable. By definition, there exists a sequence of code $c \in \mc{C}(n)$ whose $n$-stage utility $\E\big[ \frac{1}{n} \cdot  \sum_{i=1}^n \Phi(U_i ,X_i ,Y_i ,V_i )\big]$ converges to $\phi^{\star} $. We define the expected empirical distribution of symbols $ \overline{Q}^n = \E\Big[Q^n\Big]$, induced by a code $c\in\C(n)$, and we show that it corresponds to the probability distribution of the mean random variables $(U_T,X_T,Y_T,V_T)$ introduced in the converse proof of Theorem \ref{theo:1CwithChannel}, stated in App. \ref{sec:ProofConverse}:
\begin{align}
 &\overline{Q}^n(u,x,y,v)  = \E\Big[Q^n (u,x,y,v) \Big]\\
&=\sum_{(u^n,x^n,y^n,v^n)}\PP(u^n,x^n,y^n,v^n) \nonumber\\
&\times \frac{1}{n} \cdot \textsf{N}(u,x,y,v|u^n,x^n,y^n,v^n)\\
&=\sum_{(u^n,x^n,y^n,v^n)}\PP\Big(u^n,x^n,y^n,v^n\Big) \nonumber\\
&\times\sum_{i=1}^n \PP\Big(T=i\Big)  \cdot \UN_{\big\{(u_i,u_i,y_i,v_i) = (u,x,y,v)  \big\}} \\
&=\PP\Big((U_T,X_T,Y_T,V_T)=(u,x,y,v)\Big). \label{eq:ExpectedEmpiricalDistribution}
\end{align}
Equations \eqref{eq:Conv1SCwithChannel1}-\eqref{eq:Conv1SCwithChannel5} of the converse proof of Theorem \ref{theo:1CwithChannel}, in App. \ref{sec:ProofConverse}, guarantee that for all $n\geq 1$, the expected empirical distribution $  \overline{Q}^n$ satisfies the information constraint \eqref{eq:Conv1SCwithChannel6bb}:
\begin{align}
0&\leq&   \max_{\QQ \in \Q} \bigg( I(W  ; Y_T  | V_T ) -  I(  W , V_T   ; U_T) \bigg) .\label{eq:Conv1SCwithChannel6bb} 
\end{align}
This proves that the $n$-stage utility  $\E\big[ \frac{1}{n} \cdot  \sum_{i=1}^n \Phi(U_i ,X_i ,Y_i ,V_i )\big]$ belongs to $\textsf{U}$, for all $n\in \N^*$. Moreover, the set $\textsf{U}$ is closed and bounded subset of $\R$, hence it is a compact set. Hence a sequence of utility $\E\big[ \frac{1}{n} \cdot  \sum_{i=1}^n \Phi(U_i ,X_i ,Y_i ,V_i )\big] \in \textsf{U}$ converges to a point of $\textsf{U}$. This proves that the achievable utility $\phi^{\star} $ belongs to the set $\textsf{U}$.

%%%%%%%%%%%%%%%%%%%%%%%%%%%%%%%%%%%%%%%%%%
%%%%%%%%%%%%%%%%%%%%%%%%%%%%%%%%%%%%%%%%%%%%%%%%%%%%%%%%%%%%%%%%%%%%%%%%%%%%%%%%%%%%%%%%%%%%%%%%%%%%%%%%%%%%%%%%%%%%%%%%%%%%%%%%%%%%%%%%%%%%%%%%%%%%%%%%%%%%%%%%%%%%%%%%

\section{Proof of Theorem \ref{theo:ConvexProblem}}\label{sec:ProofTheoConvex}

We prove that the set $\mc{A}$ is convex in order to show that the optimization problem stated in equation \eqref{eq:optimizationProblem} is a convex problem. To do so, we investigate the  mapping $\Delta$ defined by equation \eqref{eq:MappingDelta} and we prove that it is concave with respect to the conditional probability distribution $\QQ(x,v | u)$
\begin{align}
\Delta : \QQ(x,v | u) \mapsto  \max_{ \QQ(w | u, x,v) } \bigg( I( W;Y  |V )  -   I(  U ; V  ,W  )  \bigg).\label{eq:MappingDelta} %  = \Delta\Big(\QQ(x,v | u)\Big)
\end{align}
For all $\lambda\in[0,1]$, for all conditional probability distributions $\QQ^1(x,v | u)$, $\QQ^2(x,v | u)$, we prove that the mapping $\Delta$ satisfies equation \eqref{eq:ConcaveMappingDelta}.
\begin{align}
&&\lambda \cdot \Delta\Big(\QQ^1(x,v | u)\Big) +  (1 - \lambda) \cdot \Delta\Big(\QQ^2(x,v | u)\Big)\nonumber  \\
&\leq& \Delta\Big(\lambda \cdot\QQ^1(x,v | u) + (1 - \lambda) \cdot\QQ^2(x,v | u) \Big).\label{eq:ConcaveMappingDelta}
\end{align}
We denote by $\QQ^{\star1}(w |u,x,v )$ and $\QQ^{\star2}(w |u,x,v )$ the conditional probability distributions that achieve the maximum in equations \eqref{eq:ConcaveProba1} and \eqref{eq:ConcaveProba2}, defined with respect to $\QQ^1(x,v | u)$ and $\QQ^2(x,v | u)$:
\begin{align}
&\max_{ \QQ^1(w | u, x,v) } \bigg( I_{\QQ^1}( W;Y  |V )  -   I_{\QQ^1}(  U ; V  ,W  )  \bigg),\label{eq:ConcaveProba1} \\% \text{ defined w.r.t. probability }\QQ^1(x,v | u),\\
&\max_{ \QQ^2(w | u, x,v) } \bigg( I_{\QQ^2}( W;Y  |V )  -   I_{\QQ^2}(  U ; V  ,W  )  \bigg). \label{eq:ConcaveProba2}%\text{ defined w.r.t. probability }\QQ^2(x,v | u),\label{eq:ConcaveProba2}
\end{align}
We define an auxiliary random variable $Z\in\{1,2\}$, independent of $U$ such that $\PP(Z = 1) = \lambda$ and $\PP(Z = 2) = 1 - \lambda$ and:
\begin{align}
\QQ^{\star} (x,v,w | u,z=1) =  \QQ^1(x,v | u) \cdot \QQ^{\star1}(w |u,x,v ), \label{eq:ConcaveProba3}\\
\QQ^{\star} (x,v,w | u,z=2) =  \QQ^2(x,v | u) \cdot \QQ^{\star2}(w |u,x,v ). \label{eq:ConcaveProba4}
\end{align}
We denote by $\QQ^{\star} (x,v | u)$, the convex combination of $\QQ^1(x,v | u)$ and $\QQ^2(x,v | u)$, defined by:
\begin{align}
&\QQ^{\star} (x,v | u) = \sum_{w,z}\QQ^{\star} (x,v,w,z | u) \\
&=  \sum_{w,z}  \PP(z) \cdot   \QQ^{\star} (x,v,w | u,z)  \\ 
&=   \sum_{w}  \lambda \cdot   \QQ^{\star} (x,v,w | u,z=1) \\
&+  \sum_{w}  (1 - \lambda) \cdot   \QQ^{\star} (x,v,w | u,z=2)   \\ 
&=     \sum_{w}\lambda \cdot  \QQ^1(x,v | u) \cdot \QQ^{\star1}(w |u,x,v ) \\
&+   \sum_{w}(1 - \lambda) \cdot   \QQ^2(x,v | u) \cdot\QQ^{\star2}(w |u,x,v )   \\ 
&=     \lambda \cdot  \QQ^1(x,v | u) \cdot \sum_{w}\QQ^{\star1}(w |u,x,v ) \\
&+   (1 - \lambda) \cdot   \QQ^2(x,v | u) \cdot \sum_{w} \QQ^{\star2}(w |u,x,v )   \\ 
&=     \lambda \cdot  \QQ^1(x,v | u)  +   (1 - \lambda) \cdot   \QQ^2(x,v | u)   .
\end{align}
We have the following equations:
\begin{align}
&\lambda \cdot  \Delta\Big(\QQ^1(x,v | u)\Big) +  (1 - \lambda) \cdot \Delta\Big(\QQ^2(x,v | u)\Big) \nonumber\\
&= \lambda \cdot  \bigg( I_{\QQ^{\star1}}( W;Y  |V )  -   I_{\QQ^{\star1}}(  U ; V  ,W  )  \bigg)  \nonumber \\
&+ (1 - \lambda) \cdot  \bigg( I_{\QQ^{\star2}}( W;Y  |V )  -   I_{\QQ^{\star2}}(  U ; V  ,W  )  \bigg) \label{eq:Concavity3} \\
&= \lambda \cdot  \bigg( I_{\QQ^{\star}}( W;Y  |V,  Z = 1 )  -   I_{\QQ^{\star}}(  U ; V  ,W  |   Z = 1)  \bigg)    \nonumber\\
&+ (1 - \lambda) \cdot  \bigg( I_{\QQ^{\star}}( W;Y  |V,  Z = 2 )-   I_{\QQ^{\star}}(  U ; V  ,W|   Z = 2  )  \bigg) \label{eq:Concavity4} \\
&= I_{\QQ^{\star}}( W;Y  |V,  Z  )  -   I_{\QQ^{\star}}(  U ; V  ,W  |   Z )  \label{eq:Concavity5} \\
&\leq  I_{\QQ^{\star}}( W,  Z;Y  |V  )  -   I_{\QQ^{\star}}(  U ; V  ,W  |   Z )  \label{eq:Concavity6} \\
&=  I_{\QQ^{\star}}( W,  Z;Y  |V  )  -   I_{\QQ^{\star}}(  U ; V  ,W ,   Z )  \label{eq:Concavity7} \\
&=  I_{\QQ^{\star}}( W' ;Y  |V  )  -   I_{\QQ^{\star}}(  U ; V  ,W' )  \label{eq:Concavity8} \\
&=  \max_{ \QQ^{\star}(w'' | u, x,v) } \bigg( I_{\QQ^{\star}}( W'' ;Y  |V  )  -   I_{\QQ^{\star}}(  U ; V  ,W'' )   \bigg)  \label{eq:Concavity9} \\
&= \Delta\Big(\QQ^{\star}(x,v | u)\Big)  \label{eq:Concavity10} \\
&= \Delta\Big( \lambda \cdot  \QQ^1(x,v | u)  +   (1 - \lambda) \cdot   \QQ^2(x,v | u)  \Big). \label{eq:Concavity11} 
\end{align}
Equation  \eqref{eq:Concavity3} comes from the definition of the mapping $\Delta\big(\QQ(x,v | u)\big)$ and the conditional probability distributions  $\QQ^{\star1}(w |u,x,v )$ and $\QQ^{\star2}(w |u,x,v )$, stated in equations \eqref{eq:ConcaveProba1} and \eqref{eq:ConcaveProba2}.\\
Equation  \eqref{eq:Concavity4} comes from the introduction of the auxiliary random variable $Z$ and the definition of conditional probability distribution $\QQ^{\star} (x,v,w | u,z) $, stated in equations \eqref{eq:ConcaveProba3} and \eqref{eq:ConcaveProba4}.\\ 
Equation  \eqref{eq:Concavity5} comes from the definition of the mutual information.\\ 
Equation  \eqref{eq:Concavity6} comes from the property of the mutual information.\\ 
Equation  \eqref{eq:Concavity7} comes from the fact that random variables $U$ and $Z$ are independent, hence $I_{\QQ^{\star}}(  U ;  Z )=0$.\\ 
Equation  \eqref{eq:Concavity8} comes from the  introduction of the auxiliary random variable $W' = (W, Z) $. This auxiliary random variable $W' = (W, Z) $ satisfies the Markov chain $Y  -\!\!\!\!\minuso\!\!\!\!-X    -\!\!\!\!\minuso\!\!\!\!-  (W' ,U,V)$.\\
Equation  \eqref{eq:Concavity9} comes from taking the maximum over the set of conditional probability distributions $\QQ^{\star}(w'' | u, x,v)$, that satisfy the Markov chain $Y  -\!\!\!\!\minuso\!\!\!\!-X    -\!\!\!\!\minuso\!\!\!\!-  (W'' ,U,V)$.\\
Equation \eqref{eq:Concavity10} follows from the definition of the mapping $\Delta\big(\QQ(x,v | u)\big)$ in equation  \eqref{eq:MappingDelta}.\\
Equation \eqref{eq:Concavity11} follows from the definition of convex combination $\QQ^{\star} (x,v | u)  = \lambda \cdot  \QQ^1(x,v | u)  +   (1 - \lambda) \cdot   \QQ^2(x,v | u) $.

The same arguments are valid for any convex combination. This proves that the mapping $\Delta$ is concave with respect to the set of conditional probability distributions $\QQ (x,v | u)$. The concavity property implies $  \Delta\Big( \lambda \cdot  \QQ^1(x,v | u)  +   (1 - \lambda) \cdot   \QQ^2(x,v | u)  \Big) \geq \lambda \cdot  \Delta\Big(\QQ^1(x,v | u)\Big) +  (1 - \lambda) \cdot \Delta\Big(\QQ^2(x,v | u)\Big) \geq0$, hence any convex combination $\QQ^{\star} (x,v | u) = \lambda \cdot  \QQ^1(x,v | u)  +   (1 - \lambda) \cdot   \QQ^2(x,v | u)$ also belongs to the set $\mc{A}$. This proves that $\mc{A}$ is a convex set and the optimization problem stated in equation \eqref{eq:optimizationProblem} of Theorem \ref{theo:ConvexProblem}, is a convex optimization problem.

 %%%%%%%%%%%%%%%%%%%%%%%%%%%%%%%%%%%%%%%%%%
  %%%%%%%%%%%%%%%%%%%%%%%%%%%%%%%%%%%%%%%%%%
 %%%%%%%%%%%%%%%%%%%%%%%%%%%%%%%%%%%%%%%%%%
 %%%%%%%%%%%%%%%%%%%%%%%%%%%%%%%%%%%%%%%%%% %%%%%%%%%%%%%%%%%%%%%%%%%%%%%%%%%%%%%%%%%%
  %%%%%%%%%%%%%%%%%%%%%%%%%%%%%%%%%%%%%%%%%%
 %%%%%%%%%%%%%%%%%%%%%%%%%%%%%%%%%%%%%%%%%%
 %%%%%%%%%%%%%%%%%%%%%%%%%%%%%%%%%%%%%%%%%%

\section{Decomposition of the probability distribution for Theorem \ref{theo:CwithChannel}}\label{sec:ProofDecompositionCD}

In order to prove the assertion $1)$ of Theorem \ref{theo:CwithChannel}, we assume that the joint distribution $\QQ(u,x,y,v)$ is achievable and we introduce the mean probability distribution $ \overline{\PP}_n(u,x,y,v)$ defined by:
\begin{eqnarray}
 \overline{\PP}_n(u,x,y,v) &=& \frac{1}{n} \cdot \sum_{j=1}^n  \PP(u_j, x_j, y_j ,v_j).\label{eq:ProofDecomposition15}
\end{eqnarray}
Lemma \ref{lemma:MarginalDistributionC} states that for all $j\in \{1,\ldots,n\}$, the marginal distribution $\PP(u_j, x_j, y_j ,v_j)$ decomposes as: $\PP_{\sf{u}} (u_j) \times \PP(x_j|u_j) \times \mc{T}(y_j | x_j) \times \PP(v_j |u_j,x_j,y_j)$. We only consider the three random variables $(U,X,Y)$ and we prove that the mean  distribution $ \overline{\PP}_n(u,x,y)$ satisfies the Markov chain $Y  -\!\!\!\!\minuso\!\!\!\!-X    -\!\!\!\!\minuso\!\!\!\!-  U$:
\begin{eqnarray}
 \overline{\PP}_n(u,x,y,v) &=& \PP_{\sf{u}} (u)    \times \overline{\PP}_n(x|u) \times \mc{T}(y | x),
\end{eqnarray}
where for each symbols $(u,x,y,v)$ we have: $\overline{\PP}_n(x|u)=  \frac{1}{n} \cdot \sum_{j=1}^n  \PP(x_j=x | u_j=u )$. Since the joint distribution $\QQ(u,x,y,v)$ is achievable, there exists a code $c\in\mc{C}(n)$ with causal decoding such that the empirical distribution $Q^n(u,x,y,v)$ converges to $\QQ(u,x,y,v)$, with high probability. Convergence in probability implies the convergence in distribution, hence  $ \overline{\PP}_n(u,x,y,v) $ also converges to $\QQ(u,x,y,v)$. This proves that the probability distribution $\QQ(u,x,y,v) = \PP_{\sf{u}} (u) \times \QQ(x |u)  \times \mc{T}(y | x) \times \QQ(v |u,x,y) $ satisfies the assertion $1)$ of Theorem \ref{theo:CwithChannel}.

\begin{lemma}[Marginal distribution]\label{lemma:MarginalDistributionC}
Let $\PP\big( u^n, x^n, y^n, v^n\big)$, the joint distribution induced by the code $c\in \mc{C}(n)$ with causal decoding. For all $j\in \{1,\ldots,n\}$, the marginal distribution satisfies:
\begin{align}
\PP(u_j,x_j,y_j,v_j)& = \PP_{\sf{u}} (u_j) \times \PP(x_j |u_j)  \nonumber \\
&\times \mc{T}(y_j | x_j) \times \PP(v_j |u_j , x_j , y_j) .
\end{align}
\end{lemma}
\begin{proof}[Lemma \ref{lemma:MarginalDistributionC}]
The notation $u^{-j}$ stands for the sequence $u^n$ where the symbol $u_j$ has been removed: $u^{-j}=\{u_1,\ldots,u_{j-1},u_{j+1},\ldots,u_n\} \in \mc{U}^{n-1}$.
\begin{align}
&\PP\big( u^n, x^n, y^n, v^n\big) \nonumber \\
&= \prod_{i=1}^n \PP_{\sf{u}} (u_i) \times \PP(x^n|u^n) \times \prod_{i=1}^n \mc{T}(y_i | x_i) \times \prod_{i=1}^n \PP(v_i |y^{i}) \nonumber \\ &\label{eq:ProofDecomposition0} \\
&= \PP_{\sf{u}} (u_j)\times \PP(u^{-j},x^{n}|u_j)  \times \prod_{i=1}^n \mc{T}(y_i | x_i) \times \prod_{i=1}^n \PP(v_i |y^{i}) \nonumber \\ & \label{eq:ProofDecomposition1} \\
&= \PP_{\sf{u}} (u_j)\times \PP(x_j|u_j) \times \PP(u^{-j},x^{-j}|u_j,x_j)\nonumber \\
&\times \mc{T}(y_j | x_j) \times \PP(y^{-j}|u^n,x^n)  \times \prod_{i=1}^n \PP(v_i |y^{i}) \label{eq:ProofDecomposition3} \\
&= \PP_{\sf{u}} (u_j)\times \PP(x_j|u_j)  \times \mc{T}(y_j | x_j) \nonumber \\
&\times \PP(u^{-j},x^{-j},y^{-j}|u_j,x_j,y_j) \times \PP(v^n|u^n,x^n,y^n)\label{eq:ProofDecomposition5} \\
&= \PP_{\sf{u}} (u_j)\times \PP(x_j|u_j)  \times \mc{T}(y_j | x_j)\times \PP(v_j|u_j,x_j,y_j) \nonumber \\
&\times \PP(u^{-j},x^{-j},y^{-j},v^{-j}|u_j,x_j,y_j,v_j). \label{eq:ProofDecomposition7} 
\end{align}
Equation \eqref{eq:ProofDecomposition0} comes from the properties of the i.i.d. information source, the non-causal encoder, the memoryless channel and the strictly causal decoder.\\
Equation \eqref{eq:ProofDecomposition1} comes from the i.i.d. property of the information source.\\
Equation \eqref{eq:ProofDecomposition3}  comes from the memoryless property of the channel.\\
Equation  \eqref{eq:ProofDecomposition5} comes from the causal decoding.\\
Equation  \eqref{eq:ProofDecomposition7} concludes Lemma \ref{lemma:MarginalDistributionC} by taking the sum over the sequences $(v^{-j} ,y^{-j} , x^{-j},u^{-j})$.
\end{proof}

 %%%%%%%%%%%%%%%%%%%%%%%%%%%%%%%%%%%%%%%%%%
  %%%%%%%%%%%%%%%%%%%%%%%%%%%%%%%%%%%%%%%%%%
 %%%%%%%%%%%%%%%%%%%%%%%%%%%%%%%%%%%%%%%%%%
 %%%%%%%%%%%%%%%%%%%%%%%%%%%%%%%%%%%%%%%%%% %%%%%%%%%%%%%%%%%%%%%%%%%%%%%%%%%%%%%%%%%%
  %%%%%%%%%%%%%%%%%%%%%%%%%%%%%%%%%%%%%%%%%%
 %%%%%%%%%%%%%%%%%%%%%%%%%%%%%%%%%%%%%%%%%%
 %%%%%%%%%%%%%%%%%%%%%%%%%%%%%%%%%%%%%%%%%%

\section{Achievability  of Theorem \ref{theo:CwithChannel}}\label{sec:ProofAchievabilityC}

In order to prove assertion $2)$ of Theorem \ref{theo:CwithChannel}, we consider a joint probability distribution $\QQ(u,x,w_1,w_2,y,v) \in \Q_{\sf{c}}$ that achieves the maximum in equation \eqref{eq:CwithChannel1}. 
In this section, we assume that the information constraint \eqref{eq:CwithChannel1} is satisfied with strict inequality \eqref{eq:CDstrict}. The case of equality in \eqref{eq:CwithChannel1} will be discussed in App. \ref{sec:ProofCEqualityIC}.
\begin{align}
 &I( W_1;Y |W_2) - I(U;W_2,W_1) \nonumber\\
 &=  I( W_1;Y ,W_2) - I(U,W_2;W_1) - I(U;W_2) >0.\label{eq:CDstrict}
\end{align}
There exists a small parameter $\delta>0$, a rate parameter $\textsf{R}\geq 0 $, corresponding to the source coding and a rate parameter $ \textsf{R}_{\sf{L}}\geq 0 $, corresponding to the binning parameter, such that:
\begin{eqnarray}
\textsf{R}  & =&  I( W_2 ; U )  + \delta , \label{eq:AchievabilityCD1} \\
\textsf{R}_{\sf{L}}  & =&   I( W_1 ; U ,W_2 )  + \delta ,  \label{eq:AchievabilityCD2}  \\
\textsf{R} +   \textsf{R}_{\sf{L}}   & \leq&   I( W_1;Y ,W_2)  -  \delta  .  \label{eq:AchievabilityCD3}  
\end{eqnarray}

Similarly to the proof of Theorem \ref{theo:1CwithChannel}, in Sec. \ref{sec:ProofAchievability}, we define a block-Markov random code $c\in \mc{C}(n\cdot B)$, over $B\in \N$ blocks of length $n\in \N$ and we prove that the empirical distribution converges in probability to the target distribution $\QQ(u,x,w_1,w_2,y,v) \in \Q_{\sf{c}}$.

\begin{itemize}
\item[$\bullet$] \textit{Random codebook.} We generate $| \mc{M}   |= 2^{n   \sf{R}   } $ sequences $W_2^n(m)$,  drawn from the i.i.d. probability distribution $\QQ_{\sf{w}_2}^{\times n} $ with index  $m\in \mc{M} $. We generate $| \mc{M}   \times \mc{M}_{\sf{L}}   |= 2^{n (  \sf{R}  +  \sf{R}_{\sf{L}} ) } $ sequences $W_1^n(m,l)$,  drawn from the i.i.d. probability distribution $\QQ_{\sf{w}_1}^{\times n} $,  independently of $W_2^n(m)$, with indices  $(m,l) \in \mc{M}  \times \mc{M}_{\sf{L}} $. \\
\item[$\bullet$] \textit{Encoding function.} At the beginning of the block $b\in \{2,\ldots B-1\}$, the encoder observes the sequence of symbols of source $U^n_{b+1} \in  \mc{U}^n$ of the next block $b+1$. It finds an index $m\in \mc{M}$ such that the sequences  $\big(U^n_{b+1},W^n_{2}(m)\big) \in A_{\varepsilon}^{{\star}{n}} (\QQ)$ are jointly typical. We denote by $W^n_{2,b+1} = W^n_{2}(m)$, the sequence corresponding to the block $b+1$. The encoder observes the jointly typical sequences of symbols $(U^n_{b},W^n_{2,b}) \in  \mc{U}^n \times  \mc{W}_2^n$ of the current block $b$. It finds the index $l\in \mc{M}_{\sf{L}}$ such that the sequences  $\big(U^n_{b},W^n_{2,b},W^n_{1}(m,l) \big) \in A_{\varepsilon}^{{\star}{n}}(\QQ)$ are jointly typical. We denote by $W^n_{1,b}=W^n_{1}(m,l)$, the sequence corresponding to the block $b$. The encoder sends the sequence $X^n_b$,  drawn from the conditional probability distribution $\QQ_{\sf{x|uw_1w_2}}^{\times n}$  depending on the sequences $\big(U^n_{b},W^n_{2,b},W^n_{1,b}\big)$ of the block $b$.\\
\item[$\bullet$] \textit{Decoding function.} At the end of the block $b\in \{2,\ldots B-1\}$, the decoder recalls the sequence $Y^n_b$ and $W^n_{2,b}$. It finds the indices $(m,l) \in \mc{M}   \times \mc{M}_{\sf{L}} $ such that  $\big( Y^n_b ,W^n_{2,b} ,W^n_{1}(m,l) \big) \in A_{\varepsilon}^{{\star}{n}}(\QQ)$ are jointly typical. It deduces the sequence $W_{2,b+1}^n = W_{2}^n(m)$ for the next block $b+1$, that corresponds to the index $m \in \mc{M}$. In the block $b+1 $, the decoder returns the sequence $V^n_{b+1}$, drawn from the conditional probability distribution $\QQ_{\sf{v|yw_2}}^{\times n}$  depending on the sequences $\big(Y^n_{b+1} ,W^n_{2,b+1}\big)$. The sequence $W^n_{1,b+1}$ is not involved in the draw of the output $V^n_{b+1}$ of the decoder, because it is decoded with the delay of one block.\\
\item[$\bullet$] \textit{Typical sequences.} If no error occurs in the coding process, the sequences of symbols $\big( U^n_{b} ,X^n_{b} ,W^n_{1,b} ,W^n_{2,b} ,Y^n_{b} ,V^n_{b}  \big)\in A_{\varepsilon}^{{\star}{n}}(\QQ)$ are jointly typical for each block $b \in \{2, \ldots, B-1\}$. The sequences $\big( U^n_{B} ,W^n_{2,B} ,X^n_{B} ,Y^n_{B} ,V^n_{B}  \big) \in A_{\varepsilon}^{{\star}{n}}(\QQ)$ of the last block $B$ are also jointly typical but the sequences  $\big( U^n_{b_1} ,X^n_{b_1} ,Y^n_{b_1} ,V^n_{b_1} \big) \notin A_{\varepsilon}^{{\star}{n}}(\QQ)$ of the first block $b_1$, are not jointly typical in general.\\ 
\end{itemize}

\begin{figure}[!ht]
\begin{center}
\psset{xunit=0.7cm,yunit=0.7cm}
\begin{pspicture}(0.5,-1)(12,4)
\pscircle[linecolor = blue](0.5,0.5){0.35}
\psframe[linecolor = blue](2,0)(3,1)
\psframe(4,0)(5,1)
\pscircle(6.5,0.5){0.35}
\psframe(8,0)(9,1)
\pscircle(4.5,2.5){0.35}
\psframe[linecolor = blue](10,0)(11,1)
\psline[linewidth=1pt,linecolor = blue]{->}(1,0.5)(2,0.5)
\psline[linewidth=1pt]{->}(3,0.5)(4,0.5)
\psline[linewidth=1pt]{->}(5,0.5)(6,0.5)
\psline[linewidth=1pt]{->}(7,0.5)(8,0.5)
\psline[linewidth=1pt]{->}(9,0.5)(10,0.5)
\psline[linewidth=1pt]{->}(4.5,2)(4.5,1)
\psline[linewidth=1pt]{->}(5,2.5)(8.5,2.5)(8.5,1)
\psline[linewidth=1pt,linecolor = blue]{->}(11,0.5)(12,0.5)
\rput(0.5,0.5){\textcolor[rgb]{0.00,0.00,1.0}{$\PP_{\sf{u}}$}}
\rput(2.5,0.5){\textcolor[rgb]{0.00,0.00,1.0}{$\C$}}
\rput(4.5,0.5){$\C$}
\rput(6.5,0.5){$\mc{T}$}
\rput(8.5,0.5){$\D$}
\rput(4.5,2.5){\textcolor[rgb]{0.00,0.00,0.0}{$\PP_{\sf{uw_2}}$}}
\rput(10.5,0.5){\textcolor[rgb]{0.00,0.00,1.0}{$\D$}}
\rput[u](1.5,0.8){\textcolor[rgb]{0.00,0.00,1.0}{$U^n_{b+1}$}}
\rput[u](11.5,1.3){\textcolor[rgb]{0.00,0.00,1.0}{$W^n_{2,b+1}(m)$}}
\psline[linewidth=0.5pt, linestyle = dashed](3.5,-1)(3.5,4.5)
\psline[linewidth=0.5pt, linestyle = dashed](9.5,-1)(9.5,4.5)
\rput[u](5.5,0.8){$X^n_{b}$}
\rput[u](7.5,0.8){$Y^n_{b}$}
\rput[u](5.7,1.65){$(U^n_{b},W^n_{2,b})$}
\rput[u](8,1.65){$W^n_{2,b}$}
\rput[u](3.5,0.8){\textcolor[rgb]{0.70,0.00,0.0}{${m}$}}
\rput[u](9.5,0.8){\textcolor[rgb]{0.70,0.00,0.0}{${{m}}$}}
\rput[l](0,4){Lossy source code}
\rput[l](0,3.5){block $b+1$}
\rput[l](5.2,4){Channel code}
\rput[l](5.2,3.5){block $b$}
\rput[l](9,4){Lossy source code}
\rput[l](9,3.5){block $b+1$}
\rput[l](3.6,-0.7){Two-sided state information}
\end{pspicture}
%\caption{This coding scheme is very similar to the one of Theorem \ref{theo:1CwithChannel}, in Sec. \ref{sec:ProofAchievability}. The random variable $W_2$ replace $V$. The decoder observes $Y$ and decodes $W_2$ in order to generate $V$ with i.i.d. distribution $\QQ(v|y,w_2)$.}
\caption{$\mc{D}$ observes $Y$, decodes $W_2$ and generates $V$ with $\QQ(v|y,w_2)$.}\label{fig:BlockMarkovSCDCD}
%\caption{This coding scheme is inspired by the block-Markov code represented by Fig. 6, in \cite{LetreustZaidiLasaulce(Allerton)11}. It is a concatenation of a source code for the block $b+1$ and of a channel code for the block $b$. The encoder determines  the index $m \in \mc{M}$, such that the sequences $\big(U^n_{b+1} , W^n_{2}(m) \big) \in A_{\varepsilon}^{{\star}{n}}(\QQ)$ of the block $b+1$, are jointly typical. The index $m \in \mc{M}$ is sent during the previous block $b$, using a channel code for the case of two-sided state information. The encoder observes $(U^n_{b} , W^n_{2,b}) $ and the decoder  observes  $W^n_{2,b} $. After decoding the index $m \in \mc{M}$, the decoder deduces the sequence $W^n_{2,b+1} = W^n_{2}(m)$ and returns the sequence $V^n_{b+1}$, drawn from the conditional probability distribution $\QQ_{\sf{v|yw_2}}^{\times n}$  depending on the sequences $\big(Y^n_{b+1} ,W^n_{2,b+1}\big)$. The  sequence $U^n_{b} $ does not influence the statistics of the  channel $\mc{T}_{\sf{y|x}}$ but the result still holds if we consider a state-dependent memoryless channel $\mc{T}_{\sf{y|ux}}$ where $U$ is the channel state, as mentioned by Theorem 14, in \cite{LarrousseLasaulceBloch(IT)14}. Unlike for the strictly causal decoding, the channel output $Y$ can not depends on the output $V$ of the decoder, because it is causal. }\label{fig:BlockMarkovSCDCD}
\end{center}
\end{figure}

 For all $\varepsilon_1>0$ there exists an $\bar{n} \in \N$ such that for all $n\geq\bar{n}$, the expected probability of the following error events are bounded by $\varepsilon_1$:
\begin{eqnarray}
&&\E_c\bigg[ \PP\bigg(U^n_{b}     \notin A_{\varepsilon}^{{\star}{n}} (\QQ) \bigg)\bigg]  \leq \varepsilon_1, \label{eq:AchievProbaCD1} \\
&&\E_c\bigg[ \PP\bigg( \forall  m \in \mc{M}  ,\quad 
\big(U^n_{b}, W_{2}^n(m) \big) \notin A_{\varepsilon}^{{\star}{n}}(\QQ) \bigg)\bigg]  \leq \varepsilon_1, \nonumber \\
&&\label{eq:AchievProbaCD2} \\
&&\E_c\bigg[ \PP\bigg( \forall  l \in \mc{M}_{\sf{L}}   ,\quad \nonumber \\
&&\big(U^n_{b-1}, W_{2,b-1}^n ,W_1^n(m,l) \big) \notin A_{\varepsilon}^{{\star}{n}}(\QQ) \bigg)\bigg]  \leq \varepsilon_1, \label{eq:AchievProbaCD3} \\
&&\E_c\bigg[ \PP\bigg(  \exists (m', l')\neq  (m ,l)  ,\text{ s.t. } \nonumber \\
&&
\big(Y^n_{b-1} , W^n_{2,b-1}  ,W_1^n(m',l') \big) \in A_{\varepsilon}^{{\star}{n}}( \QQ)\bigg)\bigg]   \leq \varepsilon_1.\label{eq:AchievProbaCD4}
\end{eqnarray}
Equation \eqref{eq:AchievProbaCD1} comes from the properties of typical sequences, stated in \cite[pp. 27]{ElGammalKim(book)11}.\\
Equation \eqref{eq:AchievProbaCD2} comes from equation \eqref{eq:AchievabilityCD1} and the covering lemma, stated in \cite[pp. 208]{ElGammalKim(book)11}.\\
Equation \eqref{eq:AchievProbaCD3} comes from equation \eqref{eq:AchievabilityCD2} and the covering lemma, stated in \cite[pp. 208]{ElGammalKim(book)11}.\\
Equation \eqref{eq:AchievProbaCD4} comes from equation \eqref{eq:AchievabilityCD3} and the packing lemma, stated in \cite[pp. 46]{ElGammalKim(book)11}.\\

The expected error probability of this block-Markov code is upper bounded, using the same arguments as for the achievability proof of Theorem \ref{theo:1CwithChannel}, stated in App. \ref{sec:ProofAchievability}.

%%%%%%%%%%%%%%%%%%%%%%%%%%%%%%%%%%%%%%%%%%%%%%%%%%%%%%%%%%%%%%%%%%%%%%%%%%%%

 %%%%%%%%%%%%%%%%%%%%%%%%%%%%%%%%%%%%%%%%%%
 %%%%%%%%%%%%%%%%%%%%%%%%%%%%%%%%%%%%%%%%%%
 %%%%%%%%%%%%%%%%%%%%%%%%%%%%%%%%%%%%%%%%%%
  %%%%%%%%%%%%%%%%%%%%%%%%%%%%%%%%%%%%%%%%%%
 %%%%%%%%%%%%%%%%%%%%%%%%%%%%%%%%%%%%%%%%%%
 %%%%%%%%%%%%%%%%%%%%%%%%%%%%%%%%%%%%%%%%%%

\section{Equality in the information constraint \eqref{eq:CwithChannel1}}\label{sec:ProofCEqualityIC}

We consider a target distribution $ \PP_{\sf{u}}(u)   \times \QQ(x | u) \times  \mc{T}(y | x )  \times \QQ(v | u,x,y)$ such that the maximum in \eqref{eq:CwithChannel1} is equal to zero:
\begin{eqnarray}
\max_{{\QQ}\in \Q_{\sf{c}}} \bigg( I( W_1;Y |W_2) - I(U;W_2,W_1) \bigg)=0. \label{eq:CEqualityIC}
\end{eqnarray}

\subsection{First case: channel capacity is strictly positive}
The channel is not trivial and it is possible to send some reliable information: 
\begin{eqnarray}
\max_{\PP(x)} I(X;Y)  >0. \label{eq:Capacity}
\end{eqnarray}
We denote by $\QQ^{\star}(u,x,y,v)=\PP_{\sf{u}}(u) \times \PP^{\star}(x) \times \mc{T}(y|x) \times \QQ(v)$ the distribution where $(X,Y)$, $U$ and $V$ are mutually independent and where $\PP^{\star}(x)$ achieves the maximum in \eqref{eq:Capacity}. We denote by $ \max_{{\QQ}^{\star}\in \Q^{\star}_{\sf{c}}} \bigg( I_{\QQ^{\star}}( W_1;Y |W_2) - I_{\QQ^{\star}}(U;W_2,W_1) \bigg)$ the information constraint corresponding to $\QQ^{\star}(u,x,y,v)$ and we show that it is strictly positive:
\begin{align}
&\max_{{\QQ}^{\star}\in \Q^{\star}_{\sf{c}}} \bigg( I_{\QQ^{\star}}( W_1;Y |W_2) - I_{\QQ^{\star}}(U;W_2,W_1) \bigg)\nonumber \\
&\geq  I_{\QQ^{\star}}( X;Y |W_2)\label{eq:StarC1}  \\
& = I(X;Y)>0.\label{eq:StarC2}  
\end{align}
Equation \eqref{eq:StarC1} comes from the choice of auxiliary random variables $W_1=X$ and $W_2$ such that 
$W_2$, $V$, $U$ and $(W_1,X,Y)$ are mutually independent: $\QQ^{\star}(u,x,w_1,w_2,y,v) = \PP_{\sf{u}}(u) \times \QQ(w_2) \times \PP^{\star}(x) \times \UN(x=w_1) \times  \mc{T}(y|x) \times  \QQ(v)$.\\ 
Equation \eqref{eq:StarC2} comes from the independence between $W_2$ and $(X,Y)$ and the hypothesis of strictly positive channel capacity.

We construct the sequence $\{\QQ^k(u,x,y,v)\}_{k\in\N}$ of convex combinations between the target distribution $\QQ(u,x,y,v)$ and the distribution $\QQ^{\star}(u,x,y,v)$, where for all $(u,x,y,v)$:
\begin{eqnarray}
\QQ^k(u,x,y,v) =\frac{1}{k} \cdot\bigg( (k-1) \cdot \QQ(u,x,y,v) +  \QQ^{\star}(u,x,y,v) \bigg).\nonumber
\end{eqnarray}
$\bullet$ The information constraint \eqref{eq:CEqualityIC} corresponding to $\QQ(u,x,y,v)$ is equal to zero.\\
$\bullet$ The information constraint \eqref{eq:StarC2} corresponding to $\QQ^{\star}(u,x,y,v)$ is strictly positive.\\
By Theorem \ref{theo:ConvexProblemCD}, the information constraint is concave with respect to the distribution. Hence, the information constraint corresponding to the distribution $\QQ^k(u,x,y,v)$ is strictly positive for all $k>1$:
\begin{eqnarray}
\max_{{\QQ}^{k}\in \Q^{k}_{\sf{c}}} \bigg( I_{\QQ^{k}}( W_1;Y |W_2) - I_{\QQ^{k}}(U;W_2,W_1) \bigg)>0.
\end{eqnarray}
App. \ref{sec:ProofAchievabilityC} concludes that the distribution $\QQ^k(u,x,y,v)$ is achievable, for all $k>1$. Since, the distribution $\QQ^k(u,x,y,v)$ converges to the target distribution $\QQ(u,x,y,v)$, as $k$ goes to $+\infty$. This proves that the distribution $\QQ(u,x,y,v) $ is achievable.

%%%%%%%%%%%%%%%%%%%%%%%%%%%%%%%%%%%%%%%%%%%%%%%%%%%%%%%%%%%%%%%%%%%%%%%%%%%%%%%%%%%%%%%%%%%%%%%%%%%%%%%%%%%%%%%%%%%%%%%%%%%%%%%%%%%%%%%%

\subsection{Second case: channel capacity is zero}

We assume that the channel capacity is equal to zero: $\max_{\PP(x)} I(X;Y) =0$ and we consider the auxiliary random variables $(W_1,W_2)$ with bounded support: 
$\max\big( |\mc{W}_1|,  |\mc{W}_2| \big)\leq  |\mc{U} \times \mc{X}  \times \mc{Y}\times \mc{V}  | +2$. We define the set of distributions $\mc{A}_{\textsf{c}}$ of the random variables $(U,X,W_2,Y,V)$ that satisfy the information constraint of Theorem \ref{theo:CwithChannel} and we show that $\mc{A}_{\textsf{c}}$ boils down to:
\begin{align}
&\mc{A}_{\textsf{c}} =\bigg\{ \PP_{\sf{u}}(u) \times  \QQ(x,w_2 | u)   \times  \mc{T}(y | x )\times \QQ(v| y,w_2),  \text{ s.t. } \nonumber\\
&\max_{\QQ(w_1|u,x,w_2),\atop |\mc{W}_1| \leq  |\mc{U} \times \mc{X}  \times \mc{Y}\times \mc{V}  | +2 } \bigg(I( W_1;Y  |W_2 ) - I( W_1, W_2 ; U   )  \bigg) \geq 0 \bigg\} \label{eq:BoundaryC01} \\
&= \bigg\{\PP_{\sf{u}}(u) \times  \QQ(x,w_2 | u)   \times  \mc{T}(y | x )\times \QQ(v| y,w_2),   \text{ s.t. }  \nonumber\\
&  I(  W_2 ; U   )  = 0 \;\; \bigg\} \label{eq:BoundaryC02} \\
&= \bigg\{ \PP_{\sf{u}}(u) \times  \QQ(w_2 ) \times  \QQ(x| u,w_2 ) \times  \mc{T}(y | x ) \times \QQ(v| y,w_2) \bigg\} \label{eq:BoundaryC03} .
\end{align}
Equation \eqref{eq:BoundaryC01} defines the set of distributions $\mc{A}_{\textsf{c}} \subset \Delta(\mc{U} \times \mc{X} \times \mc{W}_2 \times \mc{Y} \times \mc{V}  )$.\\
Equation \eqref{eq:BoundaryC02} comes from the hypothesis of channel capacity equal to zero $\max_{\PP(x)} I(X;Y) =0$ and Lemma \ref{lemma:ZeroCapacityCIC}.\\
Equation \eqref{eq:BoundaryC03} comes from Lemma \ref{lemma:decomposition2}, by considering the distribution $  \QQ(x,v | u) $ instead of $  \QQ(x,w_2 | u) $.

\textit{Coding Scheme:} We consider a target distribution: $ \PP_{\sf{u}}(u) \times  \QQ(w_2 ) \times  \QQ(x,w_1| u,w_2 )   \times  \mc{T}(y | x ) \times \QQ(v| y,w_2)$ that belongs to the set of distributions \eqref{eq:BoundaryC03}.\\
$\bullet$ The sequence $W_2^n$ is drawn with the i.i.d. probability $\QQ(w_2)$ and known in advance by both encoder and decoder,\\
$\bullet$ The encoder observes the sequence of source $U^n$ and generates the sequence $X^n$ with the i.i.d. probability distribution   $\QQ(x|u,w_2)$, depending on  the pair $(U^n,W_2^n)$. \\
$\bullet$ The decoder observes the sequence of channel outputs $Y^n$ and generates the sequence $V^n$ with the i.i.d. probability distribution   $\QQ(v|y,w_2)$, depending on  the pair $(Y^n,W_2^n)$. \\
$\bullet$ This coding scheme proves that any probability distribution that belongs to the set of distribution \eqref{eq:BoundaryC03} is achievable:
\begin{eqnarray} 
\PP_{\sf{u}}(u) \times  \QQ(w_2 ) \times  \QQ(x| u,w_2 )   \times  \mc{T}(y | x ) \times \QQ(v| y,w_2).
\end{eqnarray}

This proves that if $\max_{\PP(x)} I(X;Y) =0$, then any distribution $\PP_{\sf{u}}(u) \times  \QQ(x,w_2 | u)   \times  \mc{T}(y | x ) \times \QQ(v| y,w_2)  \in \mc{A}_{\textsf{c}} $ is achievable.

% such that $\max_{\QQ(w_1|u,x,w_2),\atop |\mc{W}_1| \leq  |\mc{U} \times \mc{X}  \times \mc{Y}\times \mc{V}  | +2 } \Big(I( W_1;Y  |W_2 )  - I( W_1, W_2 ; U   )  \Big) \geq 0$, 

\begin{lemma}\label{lemma:ZeroCapacityCIC}
We consider both sets of probability distributions:
\begin{align}
\mc{A}_{\textsf{c}} &= \bigg\{ \PP_{\sf{u}}(u) \times  \QQ(x,w_2 | u)   \times  \mc{T}(y | x )\times \QQ(v| y,w_2),  \text{ s.t. }\nonumber \\
& \max_{\QQ(w_1|u,x,w_2),\atop |\mc{W}_1| \leq  |\mc{U} \times \mc{X}  \times \mc{Y}\times \mc{V}  | +2 } \bigg(I( W_1;Y  |W_2 ) - I( W_1, W_2 ; U   )  \bigg)  \geq 0 \bigg\}, \\
\mc{B}_{\textsf{c}}&= \bigg\{\PP_{\sf{u}}(u) \times  \QQ(x,w_2 | u)   \times  \mc{T}(y | x )\times \QQ(v| y,w_2), \nonumber \\
&\qquad \text{ s.t. }  \qquad I(  W_2 ; U   )  = 0 \;\; \bigg\} .
\end{align}
If the channel capacity is equal to zero $\max_{\PP(x)} I(X;Y) =0$, then both sets of probability distributions are equal $\mc{A}_{\textsf{c}} = \mc{B}_{\textsf{c}}$.
\end{lemma}

\begin{proof}[Lemma \ref{lemma:ZeroCapacityCIC}]
\textit{First inclusion $\mc{A}_{\textsf{c}} \subset \mc{B}_{\textsf{c}}$.} We consider a distribution $\PP_{\sf{u}}(u) \times  \QQ(x,w_2 | u)   \times  \mc{T}(y | x ) \times \QQ(v| y,w_2)$ that belongs to $\mc{A}_{\textsf{c}}$ and we denote by $W_1$ the auxiliary random variable that achieves the maximum in the information constraint. Since the channel capacity is equal to zero, we have:
\begin{align}
 0 &= \max_{\PP(x)} I(X;Y)   \geq I(W_1;Y,W_2) \nonumber \\
 &\geq I(W_1;U,W_2) + I(U;W_2) \geq I(U;W_2)\geq0. \label{eq:CapacityC05}
\end{align}
This implies that $I(U;W_2)=0$, hence the distribution $\PP_{\sf{u}}(u) \times  \QQ(x,w_2 | u)   \times  \mc{T}(y | x ) \times \QQ(v| y,w_2)$ belongs to the set  $\mc{B}_{\textsf{c}}$. This proves the first inclusion $\mc{A}_{\textsf{c}} \subset \mc{B}_{\textsf{c}}$.\\
\textit{Second inclusion $\mc{A}_{\textsf{c}} \supset \mc{B}_{\textsf{c}}$.} We consider a distribution $ \PP_{\sf{u}}(u) \times  \QQ(x,w_2 | u)   \times  \mc{T}(y | x ) \times \QQ(v| y,w_2)$ that belongs to $\mc{B}_{\textsf{c}}$, hence for which $I(U;W_2)=0$. We introduce a deterministic auxiliary variable $\widetilde{W}_1$, for which $I(\widetilde{W}_1;Y,W_2) = I(\widetilde{W}_1;U,W_2)=0$. Hence the information constraint  satisfies:
\begin{align}
&\max_{\QQ(w_1|u,x,w_2),\atop |\mc{W}_1| \leq  |\mc{U} \times \mc{X}  \times \mc{Y}\times \mc{V}  | +2 } \bigg(I( W_1;Y  |W_2 )  - I( W_1, W_2 ; U   )  \bigg) \nonumber \\
&\geq I(\widetilde{W}_1;Y,W_2) - I(\widetilde{W}_1;U,W_2) - I(U;W_2) = 0.
\end{align}
Since the information constraint  is positive, the distribution $ \PP_{\sf{u}}(u) \times  \QQ(x,w_2 | u)   \times  \mc{T}(y | x ) \times \QQ(v| y,w_2)$ belongs to the set $\mc{A}_{\textsf{c}}$. This shows the second inclusion $\mc{A}_{\textsf{c}} \supset \mc{B}_{\textsf{c}}$.
\end{proof}

 %%%%%%%%%%%%%%%%%%%%%%%%%%%%%%%%%%%%%%%%%%
 %%%%%%%%%%%%%%%%%%%%%%%%%%%%%%%%%%%%%%%%%%
 
%%%%%%%%%%%%%%%%%%%%%%%%%%%%%%%%%%%%%%%%%%%
%%%%%%%%%%%%%%%%%%%%%%%%%%%%%%%%%%%%%%%%%%%

\section{Converse proof of Theorem \ref{theo:CwithChannel}.}\label{sec:ProofConverseC}

We suppose that the joint probability distribution $\PP_{\sf{u}}(u)   \times \QQ(x | u) \times  \mc{T}(y | x )\times \QQ(v | u,x,y)$ is achievable with a causal code. For all $\varepsilon>0$, there exists a minimal length $\bar{n}\in \N$, such that for all $n \geq \bar{n}$, there exists a code $c\in\mc{C}(n)$, such that the probability of error satisfies $\PP_{\sf{e}}(c) = \PP\big(||{Q}^n - \QQ  ||_{\sf{tv}} > \varepsilon \big) \leq \varepsilon$. The parameter $\varepsilon>0$ is involved in both the definition of the typical sequences and the upper bound of the error probability. We introduce the random  event of error $E \in \{0,1\}$ defined as follows:
\begin{eqnarray}
E = \Bigg\{
\begin{array}{lll}
0 &\text{ if }& \big|\big|{Q}^n - \QQ  \big|\big|_{\sf{tv}} \leq \varepsilon ,\\
1 &\text{ if }& \big|\big|{Q}^n - \QQ  \big|\big|_{\sf{tv}} > \varepsilon.
\end{array}
\Bigg.
\end{eqnarray}
%For all $\varepsilon>0$, there exists a minimal length $\bar{n}\in \N$ such that for all $n \geq \bar{n}$ there exists a code $c\in\mc{C}(n)$ such that the probability of error  $\PP_{\sf{e}}(c) = \PP(E=1) \leq \varepsilon$. 
%Consider a sequence of code $c(n) \in \mc{C}$ that achieves the probability distribution $\QQ(u,x,y,v)$ \textit{i.e.}, for which the probability of error  $\PP_{\sf{e}}(c) = \PP(E=1)$ goes to zero. We have the following equations:
We have the following equations:
\begin{eqnarray}
0=\sum_{i=1}^n I(U^n_{i+1} ; Y_i  | Y^{i-1}) - \sum_{i=1}^n  I(   Y^{i-1}  ; U_i | U^n_{i+1}) \label{eq:ConvCD1} \\
=  \sum_{i=1}^n I( U^n_{i+1} ; Y_i  | Y^{i-1} )- \sum_{i=1}^n I(    Y^{i-1} , U^n_{i+1} ; U_i) \label{eq:ConvCD2} \\
=  \sum_{i=1}^n I(   W_{1,i} ; Y_i      | W_{2,i}) - \sum_{i=1}^n I(  W_{1,i} , W_{2,i}   ; U_i ) \label{eq:ConvCD3} .
\end{eqnarray}
Equation \eqref{eq:ConvCD1} comes from Csisz\'{a}r Sum Identity stated  in \cite[pp. 25]{ElGammalKim(book)11}.\\
Equation \eqref{eq:ConvCD2} comes from the i.i.d. property of the information source $U$, that implies $ I(U_i  ; U^n_{i+1}) = 0$ for all $i\in\{1 , \ldots,n\}$.\\
Equation \eqref{eq:ConvCD3} comes from the introduction of the auxiliary random variables $W_{1,i} = U^n_{i+1}$ and $W_{2,i} = Y^{i-1} $. The two random variables  $(W_{1,i}, W_{2,i})$ satisfy the Markov Chains corresponding to the set of probability distributions  $\Q_{\textsf{c}}$.
\begin{eqnarray}
&&Y_i -\!\!\!\!\minuso\!\!\!\!- X_i -\!\!\!\!\minuso\!\!\!\!-  (U_i , W_{1,i} ,W_{2,i} ) ,\\
&&V_i -\!\!\!\!\minuso\!\!\!\!- (Y_i,W_{2,i} ) -\!\!\!\!\minuso\!\!\!\!-  (U_i, X_i,W_{1,i} ) .
\end{eqnarray}
$\bullet$ The first Markov chain comes from memoryless property of the channel and the fact that $Y_i$ does not belong to $(W_{1,i} ,W_{2,i})$.\\
$\bullet$ The second Markov chain $V_i -\!\!\!\!\minuso\!\!\!\!- ( Y_i , Y^{i-1} )-\!\!\!\!\minuso\!\!\!\!-  (U_i , X_i, U^n_{i+1} )$ comes from the causal decoding: the output of the decoder $V_i$ depends on the current symbols $(U_i , X_i)$ and the future symbols $U^n_{i+1}$ only through the past and current channel outputs $( Y_i , Y^{i-1} )$. \\

Hence, for all $i\in\{1 , \ldots,n\}$ we have:
%\begin{eqnarray}
%0&\leq&  \sum_{i=1}^n I(   W_{1,i} ; Y_i  | W_{2,i}) \nonumber\\
%&-& \sum_{i=1}^n I(  W_{1,i} , W_{2,i}   ; U_i )  \\
%&=&  n \cdot \bigg( I(  W_{1,T}  ; Y_T   | W_{2,T}, T) \nonumber\\
%&-&  I(  W_{1,T} , W_{2,T}  ; U_T  | T) \bigg)\label{eq:ConvCD4} \\
%&=&  n \cdot \bigg( I(W_{1,T}  ; Y_T  | W_{2,T}, T) \nonumber\\
%&-&  I(  W_{1,T} , W_{2,T} , T   ; U_T ) \bigg) \label{eq:ConvCD5} \\
%&\leq& n \cdot \max_{\QQ \in \Q_{\textsf{c}}} \bigg( I(W_1  ; Y_T | W_2 ) \nonumber\\
%&-& I(  W_1 , W_2   ; U_T ) \bigg) \label{eq:ConvCD6} \\
%&\leq&  n \cdot \max_{\QQ \in \Q_{\textsf{c}}} \bigg(   I(W_1  ; Y_T  | W_2,E =0 ) \nonumber\\
%&-&  I(  W_1 , W_2   ; U_T | E=0 ) + \varepsilon  \bigg) 
%\label{eq:ConvCD7} \\
%&\leq&  n \cdot \max_{\QQ \in \Q_{\textsf{c}}} \bigg( I(W_1  ; Y  | W_2  ) \nonumber\\
%&-& I(  W_1 , W_2   ; U ) + 2\varepsilon \bigg). \label{eq:ConvCD8} 
%\end{eqnarray}
\begin{align}
0&\leq  \sum_{i=1}^n I(   W_{1,i} ; Y_i  | W_{2,i}) - \sum_{i=1}^n I(  W_{1,i} , W_{2,i}   ; U_i )  \\
&=  n \cdot \bigg( I(  W_{1,T}  ; Y_T   | W_{2,T}, T) -  I(  W_{1,T} , W_{2,T}  ; U_T  | T) \bigg)\label{eq:ConvCD4} \\
&=  n \cdot \bigg( I(W_{1,T}  ; Y_T  | W_{2,T}, T)-  I(  W_{1,T} , W_{2,T} , T   ; U_T ) \bigg) \label{eq:ConvCD5} \\
&\leq n \cdot \max_{\QQ \in \Q_{\textsf{c}}} \bigg( I(W_1  ; Y_T | W_2 ) - I(  W_1 , W_2   ; U_T ) \bigg) \label{eq:ConvCD6} \\
&\leq  n \cdot \max_{\QQ \in \Q_{\textsf{c}}} \bigg(   I(W_1  ; Y_T  | W_2,E =0 ) \nonumber\\
&-  I(  W_1 , W_2   ; U_T | E=0 ) + \varepsilon  \bigg) 
\label{eq:ConvCD7} \\
&\leq  n \cdot \max_{\QQ \in \Q_{\textsf{c}}} \bigg( I(W_1  ; Y  | W_2  ) - I(  W_1 , W_2   ; U ) + 2\varepsilon \bigg). \label{eq:ConvCD8} 
\end{align}
Equation \eqref{eq:ConvCD4} comes from the introduction of the uniform random variable $T$ over the indices $\{1,\ldots,n\}$ and the introduction of the corresponding mean random variables $U_T$, $W_{1,T}$,  $W_{2,T}$, $X_T$, $Y_T$, $V_T$.\\
Equation \eqref{eq:ConvCD5} comes from the i.i.d. property of the information source that implies $I(T ; U_T) =0$.\\
Equation \eqref{eq:ConvCD6} comes from identifying $W_1$ with $W_{1,T}$ and $W_2$ with $( W_{2,T},T)$ and taking the  maximum over the probability distributions that belong to $\Q_{\textsf{c}}$. This is possible since the random variables $W_{1,T}$ and $( W_{2,T},T)$ satisfy the two Markov chains of the set of probability distributions $\Q_{\textsf{c}}$.\\
Equation \eqref{eq:ConvCD7} comes from the empirical coordination requirement, as stated in Lemma \ref{lemma:ErrorEventCoordinationSCD}. By hypothesis, the sequences are not jointly typical with small error probability $\PP_{\sf{e}}(c) = \PP(E=1)$. Lemma \ref{lemma:ErrorEventCoordinationSCD} adapts the proof of Fano's inequality to the empirical coordination requirement.\\
Equation \eqref{eq:ConvCD8} comes from Lemma \ref{lemma:3}. The probability distribution induced by the coding scheme $ \PP\big((U_T,X_T,Y_T ,V_T)= (u,x,y,v) \big| E=0\big)$ is closed to the target probability distribution $\QQ(u,x,y,v)$. The continuity of the entropy function stated in \cite[pp. 33]{CsiszarKorner(Book)11}  implies equation \eqref{eq:ConvCD8}.

If the joint probability distribution $\PP_{\sf{u}}(u)   \times \QQ(x | u) \times  \mc{T}(y | x )\times \QQ(v | u,x,y)$ is achievable with a causal code, then the following equation is satisfied for all $\varepsilon>0$:
\begin{eqnarray}
0\leq \max_{\QQ \in \Q_{\textsf{c}} } \bigg(  I(W_1  ; Y  | W_2  ) -  I(  W_1 , W_2   ; U ) + 2\varepsilon \bigg).
\end{eqnarray}
This concludes the converse proof of Theorem \ref{theo:CwithChannel}.

\begin{remark}[Stochastic encoder and decoder]
This converse result still holds when considering stochastic encoder and decoder instead of deterministic ones.
\end{remark}

\begin{remark}[Channel feedback observed by the encoder]
\label{remark:FeedbackConverseCD}
The converse proof of Theorem \ref{theo:CwithChannel} is based on the following assumptions:\\
$\bullet$ The information source $U$ is i.i.d distributed with $\PP_{\sf{u}}(u)$.\\
$\bullet$ The decoding function is causal $V_i = g_i(Y^{i})$, for all $i\in\{1,\ldots,n\}$.\\
$\bullet$ The auxiliary random variables $W_{1,i} = U^n_{i+1}$ and $W_{2,i} = Y^{i-1} $   satisfy the Markov chains $Y_i -\!\!\!\!\minuso\!\!\!\!- X_i -\!\!\!\!\minuso\!\!\!\!-  (U_i , W_{1,i} ,W_{2,i} )$ and $V_i -\!\!\!\!\minuso\!\!\!\!- (Y_i,W_{2,i} ) -\!\!\!\!\minuso\!\!\!\!-  (U_i, X_i,W_{1,i} ) $, for all $i\in \{1 , \ldots,n\}$. \\
$\bullet$ The sequences of random variables $(U^n,X^n,Y^n,V^n)$ are jointly typical for the target probability distribution $\PP_{\sf{u}}(u)   \times \QQ(x | u) \times  \mc{T}(y | x ) \times \QQ(v | u,x,y)$, with high probability.

%\begin{description}
%\item $\bullet$ The information source $U$ is i.i.d distributed with $\PP_{\sf{u}}(u)$.
%\item $\bullet$ The decoding function is causal $V_i = g_i(Y^{i})$, for all $i\in\{1,\ldots,n\}$.
%\item $\bullet$ The auxiliary random variables $W_{1,i} = U^n_{i+1}$ and $W_{2,i} = Y^{i-1} $   satisfy the Markov chains $Y_i -\!\!\!\!\minuso\!\!\!\!- X_i -\!\!\!\!\minuso\!\!\!\!-  (U_i , W_{1,i} ,W_{2,i} )$ and $V_i -\!\!\!\!\minuso\!\!\!\!- (Y_i,W_{2,i} ) -\!\!\!\!\minuso\!\!\!\!-  (U_i, X_i,W_{1,i} ) $, for all $i\in \{1 , \ldots,n\}$. 
%\item $\bullet$ The sequences of random variables $(U^n,X^n,Y^n,V^n)$ are jointly typical for the target probability distribution $\PP_{\sf{u}}(u)   \times \QQ(x | u) \times  \mc{T}(y | x ) \times \QQ(v | u,x,y)$, with high probability.
%\end{description}
As mentioned in Remark \ref{remark:FeedbackConverse}, each step of the converse holds when the encoder $X_i = f_i(U^n,Y_1^{i-1})$ observes the channel feedback $Y_1^{i-1}$ with $i\in\{1,\ldots,n\}$, drawn from the memoryless channel $\mc{T}(y_1,y | x )$. In fact, the encoder ignores the channel feedback since it arrives too late to be  exploited by the causal decoder. 
\end{remark}

\section{Bound on the cardinalities of $|\mc{W}_1|$ and $|\mc{W}_2|$ for Theorem \ref{theo:CwithChannel}}\label{sec:CardinalityBoundCD}

This section is similar to the App. C, in \cite[pp. 631]{ElGammalKim(book)11}. Lemma \ref{lemma:CardinalityBoundC} relies on the support Lemma and the Lemma of Fenchel-Eggleston-Carathéodory, stated in \cite[pp. 623]{ElGammalKim(book)11}.

\begin{lemma}[Cardinality bound for Theorem \ref{theo:CwithChannel}]\label{lemma:CardinalityBoundC}
The cardinality of the supports of the auxiliary random variables $W_1$ and $W_2$ of the Theorem \ref{theo:CwithChannel}, are bounded by $\max\big( |\mc{W}_1|,  |\mc{W}_2| \big)\leq  |\mc{U} \times \mc{X}  \times \mc{Y}\times \mc{V}  | +2 $.
\end{lemma}

\begin{proof}[Lemma \ref{lemma:CardinalityBoundC}] 
We consider the probability distribution $ \QQ(u,x,w_1,w_2,y,v) = \PP_{\sf{u}}(u)   \times \QQ(x,w_1,w_2 | u) \times  \mc{T}(y | x ) \times \QQ(v | y,w_2)$ that achieves the maximum in equation \eqref{eq:CwithChannel1} of Theorem \ref{theo:CwithChannel}. We fix a pair of symbols $(w_1,w_2)\in \mc{W}_1 \times  \mc{W}_2$ and we consider the conditional probability distribution $ \QQ(u,x,y,v |w_1,w_2) =  \QQ(u, x| w_1,w_2 ) \times  \mc{T}(y | x ) \times \QQ(v | y,w_2)$ that induces the  following continuous functions $h_i\Big(\QQ(u,x,y,v|w_1,w_2)\Big)$, from the set of joint probability distributions $\Delta(\mc{U} \times\mc{X} \times \mc{Y}\times \mc{V} )$ to $\R$:
\begin{eqnarray*}
h_i\Big(\QQ(u,x,y,v|w_1,w_2)\Big) = \qquad\qquad\qquad\qquad\qquad\qquad \nonumber \\
\begin{cases}
\QQ(u,x,y,v|w_1,w_2),\text{ for } i=\big\{ 1,\ldots,  |\mc{U} \times \mc{X}\times \mc{Y}\times \mc{V} | -1\big\},\\
H(Y|W_1=w_1,W_2=w_2) ,\quad \text{ for } i=  |\mc{U} \times \mc{X}\times \mc{Y}\times \mc{V} |,\\
H(Y|W_2=w_2) ,\quad \text{ for } i=  |\mc{U} \times \mc{X}\times \mc{Y}\times \mc{V} | +1,\\
H(U|W_1=w_1,W_2=w_2)  , \text{ for } i=  |\mc{U} \times \mc{X} \times \mc{Y}\times \mc{V}| +2.
\end{cases}
\end{eqnarray*}
The conditional entropies $H(Y|W_1=w_1,W_2=w_2)$, $H(Y|W_2=w_2)$, $H(U|W_1=w_1,W_2=w_2)$ are evaluated with respect to $ \QQ(u,x,y,v |w_1,w_2)$. The support Lemma, stated in \cite[pp. 631]{ElGammalKim(book)11}, implies that there exists a pair of auxiliary random variables $(W'_1 , W'_2) \sim \QQ(w_1',w_2')$  defined on the sets $\mc{W}'_1 \times \mc{W}'_2$ with bounded cardinality $\max\big( |\mc{W}'_1|,  |\mc{W}'_2| \big)\leq  |\mc{U} \times \mc{X}  \times \mc{Y}\times \mc{V}  | +2 $, such that:
\begin{align*}
&H(Y|W_1,W_2) \nonumber \\
  &= \int_{\mc{W}_1\times \mc{W}_2} H(Y|W_2=w_2,W_1 = w_1) d \PP_{\sf{w}_1\sf{w}_2}(w_1,w_2)  \\
&=   \sum_{(w'_1,w'_2) \in \mc{W}'_1 \times \mc{W}'_2} H(Y|W_2'=w_2',W_1' = w_1')   \nonumber \\
&\times \QQ_{\sf{w}_1'\sf{w}_2'}(w_1',w_2') = H(Y|W_1',W_2'),\\
&H(Y|W_2)  \nonumber \\
&= \int_{ \mc{W}_2}H(Y|W_2=w_2) d \PP_{\sf{w}_2}(w_2)  \\
&=   \sum_{w'_2 \in \mc{W}'_2} H(Y|W_2'=w_2') \cdot  \QQ_{\sf{w}_2'}(w_2')  = H(Y|W_2'),\\
&H(U|W_1,W_2)  \nonumber \\
&= \int_{\mc{W}_1\times \mc{W}_2} H(U|W_2=w_2,W_1 = w_1) d \PP_{\sf{w}_1\sf{w}_2}(w_1,w_2)  \\
&=   \sum_{(w'_1,w'_2) \in \mc{W}'_1 \times \mc{W}'_2} H(U|W_2'=w_2',W_1' = w_1')  \nonumber \\
&\times\QQ_{\sf{w}_1'\sf{w}_2'}(w_1',w_2') = H(U|W_1',W_2'),\\
&\QQ(u,x,y,v)  \nonumber \\
&=  \int_{\mc{W}_1\times \mc{W}_2} \QQ(u,x | w_1,  w_2) \times \mc{T}(y|x) \nonumber \\
&\times \QQ(v| y ,  w_2)  \times  \PP_{\sf{w}_1\sf{w}_2}(w_1,w_2) \\
&=  \sum_{(w'_1,w'_2) \in \atop \mc{W}'_1 \times \mc{W}'_2}  \QQ(u,x | w_1',  w'_2) \times \mc{T}(y|x) \nonumber \\
&\times \QQ(v| y ,  w'_2)   \times \QQ_{\sf{w}_1'\sf{w}_2'}(w_1',w_2'),
\end{align*}
for all $ (u,x,y,v)$ with index $i=\big\{ 1,\ldots,  |\mc{U} \times \mc{X} \times \mc{Y} \times \mc{V}  |  + 2\big\}$.
Hence, the probability distribution $\QQ(u,x,y,v)$ and the conditional entropies $H(Y|W_1,W_2)$, $H(Y|W_2)$ and $H(U|W_1,W_2)$ are  preserved. The information constraint writes:
\begin{eqnarray*}
&&I( W_1;Y  |W_2 )  -   I(  U ; W_1,W_2  ) \\
 &= & H(Y  |W_2 )  - H( Y  |W_1,W_2  )  -   H(  U  )  + H(  U |W_1,W_2   ) \\
 &= & H(Y  |W_2' )  - H( Y  |W_1',W_2'  )  -   H(  U  )  + H(  U |W_1',W_2'   ) \\
&=& I( W_1';Y  |W_2' )  -   I(  U ; W_1',W_2'  ) ,
\end{eqnarray*}
with $\max\big( |\mc{W}'_1|,  |\mc{W}'_2| \big)\leq  |\mc{U} \times \mc{X}  \times \mc{Y}\times \mc{V}  | +2 $. This concludes the proof of Lemma \ref{lemma:CardinalityBoundC}, for the cardinality bounds of the supports of the auxiliary random variables $(W_1,W_2)$ of Theorem \ref{theo:CwithChannel}.
\end{proof}

 %%%%%%%%%%%%%%%%%%%%%%%%%%%%%%%%%%%%%%%%%%
  %%%%%%%%%%%%%%%%%%%%%%%%%%%%%%%%%%%%%%%%%%
 %%%%%%%%%%%%%%%%%%%%%%%%%%%%%%%%%%%%%%%%%%
 %%%%%%%%%%%%%%%%%%%%%%%%%%%%%%%%%%%%%%%%%%
 
 \section{Proof of Theorem \ref{theo:ConvexProblemCD}}\label{sec:ProofTheoConvexCD}

We consider two joint distributions $\QQ^1(u,x,w_1,w_2,y,v)$ and $\QQ^2(u,x,w_1,w_2,y,v)$ that belong to $\Q_{\sf{c}}$ and that achieve the maximum in the information constraint \eqref{eq:CwithChannel1}. 
We denote by $I_{\QQ^1}( W_1;Y  |W_2 )$ and $I_{\QQ^2}( W_1;Y  |W_2 )$ the mutual informations corresponding to the  distributions $\QQ^1(u,x,w_1,w_2,y,v)$ and $\QQ^2(u,x,w_1,w_2,y,v)$. For all $\lambda\in[0,1]$, we prove that any convex combination of the distributions $\QQ^{\lambda} = \lambda \cdot \QQ^1 + (1-\lambda) \cdot \QQ^2$ provides a larger information constraint than the convex combination of the information constraints. We define an auxiliary random variable $Z\in\{1,2\}$, independent of $U$ such that $\PP(Z = 1) = \lambda$ and $\PP(Z = 2) = 1 - \lambda$ and we consider the general distribution $\QQ^{\lambda}(u,x,w_1,w_2,y,v,z)$.
\begin{align}
& \lambda \cdot \bigg( I_{\QQ^1}( W_1;Y  |W_2 )  -   I_{\QQ^1}(  U ; W_1,W_2  ) \bigg) \nonumber\\
&+ (1-\lambda) \cdot \bigg( I_{\QQ^2}( W_1;Y  |W_2 )  -   I_{\QQ^2}(  U ; W_1,W_2  ) \bigg) \nonumber \\
&= \PP(Z = 1) \cdot \bigg( I_{\QQ^{\lambda}}( W_1;Y  |W_2,Z=1 )  \nonumber\\
&-   I_{\QQ^{\lambda}}(  U ; W_1,W_2 |Z=1 ) \bigg) \nonumber\\
&+ \PP(Z = 2) \cdot \bigg( I_{\QQ^{\lambda}}( W_1;Y  |W_2,Z=2 ) \nonumber\\
&-   I_{\QQ^{\lambda}}(  U ; W_1,W_2 |Z=2 ) \bigg) \label{eq:Concave2}  \\
&=  I_{\QQ^{\lambda}}( W_1;Y  |W_2,Z )  -   I_{\QQ^{\lambda}}(  U ; W_1,W_2 |Z )  \label{eq:Concave3}  \\
&=  I_{\QQ^{\lambda}}( W_1;Y  |W_2,Z )  -   I_{\QQ^{\lambda}}(  U ; W_1,W_2 ,Z ) \label{eq:Concave4}  \\
&=  I_{\QQ^{\lambda}}( W_1;Y  |W'_2 )  -   I_{\QQ^{\lambda}}(  U ; W_1,W'_2 ) \label{eq:Concave5} \\
&\leq \max_{\QQ \in \Q_{\sf{c}}} \bigg( I_{\QQ^{\lambda}}( W_1;Y  |W"_2 )  -   I_{\QQ^{\lambda}}(  U ; W_1,W"_2  ) \bigg) .\label{eq:Concave6} 
\end{align}
Equations \eqref{eq:Concave2} and  \eqref{eq:Concave3} come from the definition of the general distribution $\QQ^{\lambda}(u,x,w_1,w_2,y,v,z)$ with random variable $Z$.\\
Equation \eqref{eq:Concave4} comes from the independence between $U$ and $Z$.\\
Equation \eqref{eq:Concave5} comes from replacing $W'_2=(W_2 ,Z)$.\\
Equation \eqref{eq:Concave6} comes from taking the maximum over the joint distributions $\QQ \in \Q_{\sf{c}}$.

This result extends to any convex combination and we conclude that the information constraint \eqref{eq:CwithChannel1} is concave over the set of achievable distributions for causal decoding.

  %%%%%%%%%%%%%%%%%%%%%%%%%%%%%%%%%%%%%%%%%%
 %%%%%%%%%%%%%%%%%%%%%%%%%%%%%%%%%%%%%%%%%%
 %%%%%%%%%%%%%%%%%%%%%%%%%%%%%%%%%%%%%%%%%%

\section*{Acknowledgment}

The author would like to thank B. Larrousse and S. Lasaulce for useful discussions regarding the achievability proof of Theorem \ref{theo:1CwithChannel}, L. Wang and C. Weidmann for fruitful conversations regarding the proof of Theorem \ref{theo:AchievableUtility}, M. Bloch and R. Peretz for fruitful remarks regarding the case of information constraint equal to zero and the anonymous reviewers for providing very insightful comments.

%
%The author would like to thank Samson Lasaulce and Benjamin Larrousse for useful conversations.

\bibliographystyle{ieeetr}

%\bibliography{/Users/maelletreust/Documents/Redaction/BiblioMael}

\begin{IEEEbiography}[{\includegraphics[width=1in,height=1.25in,clip,keepaspectratio]{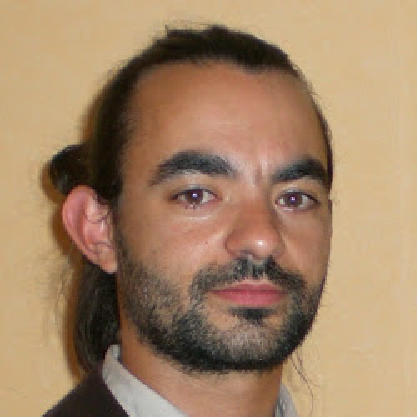}}]{Mael Le Treust}
Maël Le Treust earned his Diplôme d'Etude Approfondies (M.Sc.) degree in Optimization, Game Theory \& Economics (OJME) from the Université de Paris VI (UPMC), France in 2008 and his Ph.D. degree from the Université de Paris Sud XI in 2011, at the Laboratoire des signaux et systèmes (joint laboratory of CNRS, Supélec, Université de Paris Sud XI) in Gif-sur-Yvette, France. Since 2013, he is a CNRS researcher at ETIS laboratory UMR 8051, Université Paris Seine, Université Cergy-Pontoise, ENSEA, CNRS, in Cergy, France. In 2012, he was a post-doctoral researcher at the Institut d'électronique et d'informatique Gaspard Monge (Université Paris-Est) in Marne-la-Vallée, France. In 2012-2013, he was a post-doctoral researcher at the Centre Énergie, Matériaux et Télécommunication (Université INRS ) in Montréal, Canada. From 2008 to 2012, he was a Math T.A. at the Université de Paris I (Panthéon-Sorbonne), Université de Paris VI (UPMC) and Université Paris Est Marne-la-Vallée, France. His research interests are strategic coordination, information theory, Shannon theory, game theory, physical layer security and wireless communications.\end{IEEEbiography}

\end{document}